\documentclass[11pt,letterpaper]{article}
\usepackage{thmtools}
\usepackage{thm-restate}
\usepackage{amsthm}
\usepackage{mathrsfs}           %
\usepackage{vmargin,fancyhdr}   %
\usepackage{enumerate}

\usepackage{amsmath,amssymb}    %
\usepackage{verbatim}           %
\usepackage{xspace}             %
\usepackage{graphicx,float}     %
\usepackage{ifthen,calc}        %
\usepackage{textcomp}           %
\usepackage{fancybox}           %
\usepackage{bbm}
\usepackage{hhline}             %
\usepackage{float}              %
\usepackage{multirow}			%

\usepackage[vflt]{floatflt}     %
\usepackage[small,compact]{titlesec}
\usepackage{setspace}

\usepackage{color}
\definecolor{ForestGreen}{rgb}{0.1333,0.5451,0.1333}
\usepackage{thumbpdf}
\usepackage[letterpaper,
colorlinks,linkcolor=ForestGreen,citecolor=ForestGreen,
backref,
bookmarks,bookmarksopen,bookmarksnumbered]
{hyperref}

\usepackage{cleveref}

\setpapersize{USletter}
\setmarginsrb{1in}{.5in}        %
{1in}{.5in}        %
{.25in}{.25in}     %
{.25in}{.5in}      %
\newcommand{\showccc}[0]{0}
\newcommand{\ccc}[2][nothing]{%
	\ifthenelse{\showccc=0}{}{
		\ensuremath{^{\Lsh\Rsh}}\marginpar{\raggedright\tiny\textsf{%
				\ifthenelse{\equal{#1}{nothing}}{}{\textbf{#1}\\}#2}}}}

\newcounter{hours}\newcounter{minutes}
\newcommand{\hhmm}{%
	\setcounter{hours}{\time/60}%
	\setcounter{minutes}{\time-\value{hours}*60}%
	\ifthenelse{\value{hours}<10}{0}{}\thehours:%
	\ifthenelse{\value{minutes}<10}{0}{}\theminutes}
\ifthenelse{\showccc=0}{\rhead{}}{\rhead{\today \ [\hhmm]}}
\newtheorem{theorem}{Theorem}[section]

\newtheorem{corollary}[theorem]{Corollary}
\newtheorem{definition}[theorem]{Definition}

\newtheorem{lemma}[theorem]{Lemma}
\newtheorem{fact}[theorem]{Fact}

\usepackage{float}
\floatstyle{plain}\newfloat{myfig}{t}{figs}[section]
\floatname{myfig}{\textsc{Figure}}
\floatstyle{plain}\newfloat{myalg}{H}{algs}[section]
\floatname{myalg}{}
\newcommand{\R}[0]{\ensuremath{\mathbb{R}}}

\newcommand{\tvect}[2]{
   \ensuremath{\Bigl(\negthinspace\begin{smallmatrix}#1\\#2\end{smallmatrix}\Bigr)}}
\newcommand{\defeq}[0]{:=}
\newcommand{\norm}[1]{\left\lVert#1\right\rVert}
\newcommand{\inprod}[2]{\left\langle#1, #2\right\rangle}

 \global\long\def\normInf#1{\norm{#1}_{\infty}}

\global\long\def\argmin{\mathrm{argmin}}
\global\long\def\argmax{\mathrm{argmax}}
\global\long\def\nnz{\mathrm{nnz}}

\newcommand{\smax}{\textup{smax}}
\newcommand{\eps}{\epsilon}
\newcommand{\xset}{\mathcal{X}}
\newcommand{\pset}{\mathcal{P}}
\newcommand{\tO}[0]{\tilde{O}}
\newcommand{\ai}[0]{A_{i:}}
\newcommand{\aj}[0]{A_{:j}}
\newcommand{\pup}[0]{\textup{Prog}^{\uparrow}}
\newcommand{\pdown}[0]{\textup{Prog}^{\downarrow}}
\newcommand{\gup}[0]{g^{\uparrow}}
\newcommand{\gdown}[0]{g^{\downarrow}}
\newcommand{\half}[0]{\frac{1}{2}}
\newcommand{\E}{\mathbb{E}}
\newcommand{\tx}[0]{\tilde{x}}
\newcommand{\ty}[0]{\tilde{y}}
\newcommand{\tz}[0]{\tilde{z}}
\newcommand{\y}[0]{y_t}
\newcommand{\x}[0]{x_t}
\newcommand{\z}[0]{z_t}
\newcommand{\hx}[0]{x_{t + \half}^{(j)}}
\newcommand{\hy}[0]{y_{t + \half}}
\newcommand{\hz}[0]{w_t}
\newcommand{\py}[0]{y_{t + 1}^{(j)}}
\newcommand{\px}[0]{x_{t + 1}^{(j)}}
\newcommand{\pz}[0]{z_{t + 1}}
\newcommand{\hd}[0]{\delta_{t + \half}^{(j)}}

\newcommand{\yor}[0]{\texttt{Y-Oracle}}
\newcommand{\1}[0]{\mathbbm{1}}
\newcommand{\diag}{\textbf{\textup{diag}}}

\definecolor{burntorange}{rgb}{0.8, 0.33, 0.0}

\usepackage{tikz}
\usetikzlibrary{positioning,chains,fit,shapes,calc} %
  \usepackage{nth}
  \usepackage{intcalc}

\usepackage{url}

\begin{document}

\begin{titlepage}
\def\thepage{}
\thispagestyle{empty}

\title{Coordinate Methods for Accelerating $\ell_\infty$ Regression \\ and Faster Approximate Maximum Flow}

\date{}
\author{
Aaron Sidford \\
Stanford University \\
{\tt sidford@stanford.edu}
\and
Kevin Tian \\
Stanford University \\
{\tt kjtian@stanford.edu}
}

\maketitle

\abstract{
	
We provide faster algorithms for approximately solving $\ell_{\infty}$ regression, a fundamental problem prevalent in both combinatorial and continuous optimization. In particular, we provide accelerated coordinate descent methods capable of provably exploiting dynamic measures of coordinate smoothness, and apply them to $\ell_\infty$ regression over a box to give algorithms which converge in $k$ iterations at a $O(1/k)$ rate. Our algorithms can be viewed as an alternative approach to the recent breakthrough result of Sherman~\cite{Sherman17} which achieves a similar runtime improvement over classic algorithmic approaches, i.e. smoothing and gradient descent, which either converge at a $O(1/\sqrt{k})$ rate or have running times with a worse dependence on problem parameters. Our runtimes match those of \cite{Sherman17} across a broad range of parameters and achieve improvement in certain structured cases. 

We demonstrate the efficacy of our result by providing faster algorithms for the well-studied maximum flow problem. Directly leveraging our accelerated $\ell_\infty$ regression algorithms imply a $\tilde{O}\left(m + \sqrt{mn}/\eps\right)$ runtime to compute an $\epsilon$-approximate maximum flow for an undirected graph with $m$ edges and $n$ vertices, generically improving upon the previous best known runtime of $\tilde{O}\left(m/\eps\right)$ in \cite{Sherman17} whenever the graph is slightly dense.  We further design an algorithm adapted to the structure of the regression problem induced by maximum flow obtaining a runtime of $\tilde{O}\left(m + \max(n, \sqrt{ns})/\eps\right)$, where $s$ is the squared $\ell_2$ norm of the congestion of any optimal flow. Moreover, we show how to leverage this result to achieve improved exact algorithms for maximum flow on a variety of unit capacity graphs. We hope that our work serves as an important step towards achieving even faster maximum flow algorithms.
} 
\end{titlepage}

\section{Introduction}
\label{sec:intro}

The classic problem of \emph{$\ell_\infty$ regression} corresponds to finding a point $x^*$ such that
\begin{equation*}
x^* = \textrm{argmin}_{x \in \mathbb{R}^m} \|Ax - b\|_\infty,\text{ for } A \in \mathbb{R}^{n \times m}, \; b \in \mathbb{R}^n.
\end{equation*}
In this work, we are primarily concerned with developing iterative algorithms for approximately solving this problem. We use $\textsf{OPT}$ to denote $\|Ax^* - b\|_\infty$ and our goal is to find an \emph{$\epsilon$-approximate minimizer} of the $\ell_{\infty}$-regression function, i.e. a point $x \in \mathbb{R}^{m}$ such that 
\begin{equation*}
\textsf{OPT} \leq \|Ax - b\|_\infty \leq \textsf{OPT} + \epsilon.
\end{equation*}
This problem has fundamental implications in statistics and optimization \cite{Sherman13, LeeS14, LeeS15A, SidfordWWY18}. In many of these settings, it is also useful to design iterative method machinery for the following more general problem of finding
\begin{equation*}
\label{eqn:compmin}
x^* = \textrm{argmin}_{x \in S}
\, \|Ax - b\|_\infty, \text{ for } A \in \mathbb{R}^{n \times m},\; b \in \mathbb{R}^n,\;
S = \{x\in \R^m : x_{j} \in [l_j,r_j] ~ \forall j\in[m]\}
\end{equation*}
for some $m$ pairs of scalar $l_j \leq r_j$ (possibly infinite). Note that this constrained problem is strictly more general than the standard one as setting $l_j = -\infty$, $r_j = \infty,\; \forall j \in [m]$ recovers the unconstrained problem. In this work, for simplicity, the domain constraint will only be $x \in [-1, 1]^m$ (though our results apply to the more general case; see Appendix~\ref{appendix:reductionbox} for a formal statement). 

\begin{definition}[Box-constrained $\ell_\infty$ regression]
\label{def:boxedlinfreg}
We call the problem of solving, for regression matrix $A \in \R^{n \times m}$ and demands $b \in \R^n$,
\[\min_{x \in [-1, 1]^m} \norm{Ax - b}_\infty, \]
the box-constrained $\ell_\infty$ regression problem. We refer to any $x' \in [-1, 1]^m$ such that
\[\norm{Ax' - b}_\infty - \min_{x \in [-1, 1]^m} \norm{Ax - b}_\infty \le \eps\]
as an $\eps$-approximate minimizer.
\end{definition}

Many natural optimization problems can be written in the form of box-constrained $\ell_\infty$ regression, e.g. the maximum flow problem and more broadly linear programming \cite{LeeS15B}, and thus faster methods for solving box-constrained $\ell_\infty$ regression can imply faster algorithms for common problems in theoretical computer science. Therefore, the central goal of this paper is to provide faster algorithms for computing $\eps$-approximate minimizers to $\ell_{\infty}$-regression, that when specialized to the maximum flow problem, achieve faster running times.

\subsection{Regression results}

In this paper we show how to apply ideas from the literature on coordinate descent methods (see \Cref{sec:previous_work}) to obtain faster algorithms for approximately solving box-constrained $\ell_\infty$ regression. We show that by assuming particular sampling and smoothness oracles (which are implementable given sparsity assumptions on $A$), we obtain a randomized algorithm which improves upon the the classic gradient descent based methods across a broad range of parameters and attains an $\epsilon^{-1}$ dependence in the runtime. We show the following in Section~\ref{ssec:accelregressformal}.

\begin{theorem}[Accelerated box-constrained $\ell_\infty$ regression]
\label{thm:accelboxlinfreg}
There is an algorithm that $\epsilon$-approximately minimizes the box-constrained $\ell_\infty$ regression problem (Definition~\ref{def:boxedlinfreg})
in time
\[\tilde{O}\left(mc + \frac{\left(\min(m, n) + \sqrt{m\min(n, s)}\right)c\norm{A}_\infty}{\eps}\right),\] 
where each column of $A \in \mathbb{R}^{n \times m}$ has at most $c$ non-zero entries, and the optimizer $x^*$ has $\norm{x^*}_2 \le s$.
\end{theorem}

Note that since $s \le m$, the runtime is always at most $\tO(mc\norm{A}_\infty/\eps)$. Moreover, Theorem~\ref{thm:accelboxlinfreg} generically achieves a runtime of $\tO(mc + \sqrt{mn}c\norm{A}_\infty/\eps)$ in the case $n = O(m)$. We give a proof of the following simple extension, encapsulating the general box-constrained case as well as the unconstrained case, in Appendix~\ref{appendix:reductionbox}, which follows via a reduction to Theorem~\ref{thm:accelboxlinfreg}. We simplified the bounds for easy statement, but we remark that as they follow by a reduction, they admit similar improvements when e.g. $n, s \ll m$.

\begin{corollary}
\label{corr:generalbox}
There is an algorithm that $\epsilon$-approximately minimizes the box-constrained $\ell_\infty$ regression problem
\[\min_{x \in [-r, r]^m}\norm{Ax - b}_\infty \]
in $\tilde{O}\left(mcr\|A\|_\infty/\epsilon\right)$ time where each column of $A \in \mathbb{R}^{n \times m}$ has at most $c$ non-zero entries. Moreover, there is an algorithm that $\epsilon$-approximately minimizes the unconstrained $\ell_\infty$ regression problem
\[\min_{x \in \R^m}\norm{Ax - b}_\infty \]
in $\tilde{O}\left(mcr\|A\|_\infty/\epsilon\right)$ time, where the optimizer is $x^*$, and $\norm{x_0 - x^*}_\infty \le r$ for some given $x_0$.
\end{corollary}

The only other known box-constrained $\ell_\infty$ regression algorithm achieving an $\eps^{-1}$ dependence (improving upon the standard $\eps^{-2}$ dependence) without paying a dimension-dependent penalty is the recent breakthrough result of \cite{Sherman17}. Pessimistic bounds on our guarantees attain a runtime matching that of \cite{Sherman17} across a broad range of parameters (for example in the uniform sparsity case where $mc = O(\nnz(A))$). In instances with more structured regression matrices, with sharper bounds on parameters $n$, $s$, we obtain improved runtimes. These improvements are attainable by modifying the algorithm to take steps in a nonuniform diagonal norm, obtaining tighter dependences on sparsity measures of the matrix and optimal solution, which we elaborate on in Sections~\ref{sec:regression} and~\ref{sec:improvement}. Because of these tighter dependencies, in many parameter regimes, including those for the maximum flow problem for even slightly dense graphs, our result improves upon \cite{Sherman17}. 

Our work provides an alternative approach for accelerating $\ell_\infty$ gradient descent for certain highly structured optimization problems, i.e. $\ell_\infty$ regression. Whereas Sherman's work introduced an intriguing notion of area convexity and new regularizations of $\ell_\infty$ regression, our results are achieved by working with the classic smoothing of the $\ell_\infty$ norm and by providing a new accelerated coordinate descent method. We achieve our tighter bounds by exploiting local smoothness properties of the problem and dynamically sampling by these changing smoothnesses.

Our algorithm is inspired by, and builds upon, advances in non-uniform sampling for coordinate descent \cite{ZhuQRY16, QuR16, NesterovS17}, as well as extragradient proximal methods \cite{Nemirovski04, Nesterov07}, and is similar in spirit to work on accelerated algorithms for approximating packing and covering linear programs \cite{ZhuO15} which too works with non-standard notions of smoothness. Our paper overturns conventional wisdom that these techniques do not extend nicely to $\ell_\infty$ regression and the maximum flow problem. Interestingly, our algorithms gain an improved dependence on dimension and sparsity over \cite{Sherman17} in certain cases while losing the parallelism of \cite{Sherman17}. It is an open direction for future work as to see whether or not these approaches can be combined for a more general approach to minimizing $\ell_\infty$-smooth functions. 

\subsection{Maximum flow results}

The classic problem of maximum flow roughly asks for a graph $G$ with $m$ (capacitated) edges and $n$ vertices, how to send as many units of flow can be sent from a specified ``source'' vertex to a specified ``sink'' vertex while preserving flow conservation at all other vertices and without violating edge capacity constraints (i.e. the flow cannot put more units on an edge than the edge's capacity). 

The maximum flow problem is known to be easily reducible to the more general problem of \emph{minimum congestion flow}. Instead of specifying $s$ and $t$ this problem takes as input a vector $d \in \R^V$ such that $d^\top \mathbbm{1} = 0$, where $\mathbbm{1}$ is the all-ones vector. The goal of minimum congestion flow is to find a flow $f \in \R^E$ which routes $d$ meaning, mean that the imbalance of $f$ at vertex $v$ is given by $d_v$, and subject to this constraint minimizes the \emph{congestion},
\begin{equation*}
\max_{e \in E(G)} \left|f_e / u_e \right|
\end{equation*}
where $f_e$ is the flow on some edge, and $u_e$ is the capacity on that edge. We refer to the vector with entries $f_e/u_e$ as the \emph{congestion vector}. We call any flow which routes an amount within a $1 + \eps$ multiplicative factor to the optimum an \emph{$\eps$-approximate maximum flow}.

A recent line of work beginning in \cite{Sherman13,KelnerLOS14} solves the maximum flow problem by further reducing to constrained $\ell_\infty$ regression. To give intuition for the reduction used in this work, broadly inspired by \cite{Sherman13, KelnerLOS14}, we note that maximum flow in uncapacitated graphs can be rephrased as asking for the smallest congestion of a feasible flow, namely to solve the problem 
\begin{equation*}
f^* = \textrm{argmin}_{Bf = d} \|f\|_\infty
\end{equation*}
where the restriction $Bf = d$ for $B$ the edge-vertex incidence matrix of a graph, and $d$ the demands, enforces the flow constraints. This can be solved up to logarithmic factors in the running time by fixing some value $F$ for $\|f\|_\infty$ and asking to optimally solve the problem 
\begin{equation*}
f^* = \textrm{argmin}_{\|f\|_\infty \leq F} \|Bf - d\|_\infty
\end{equation*}
where we note that the constraint $\|f\|_\infty \leq F$ can be decomposed as the indicator of a box so that this objective matches the form of Equation~\ref{eqn:compmin}. The exact reduction we use has a few modifications: the box constraint is more simply replaced by $\|f\|_\infty \leq 1$, and the regression objective is in a matrix $RB$, where $R$ is a combinatorially-constructed preconditioner whose goal is to improve the condition number (and convergence rate) of the problem, and the problem is scaled for capacitated graphs (for a more detailed description, see \Cref{sec:flowtolinfregress}).

In this paper we show how to modify our algorithm for structured $\ell_\infty$ regression in order to obtain faster algorithms for maximum flow. We do so by leveraging the tighter dependence on the domain size (in the $\ell_2$ norm rather than $\ell_\infty$) and coordinate smoothness properties of the function to be minimized (due to the structure of the regression matrix). In particular we show the following.

\begin{theorem}[$\ell_2$ accelerated approximate maximum flow]
There is an algorithm that takes time $\tilde{O}(m + \max(n, \sqrt{ns})/\eps)$ to find an $\eps$-approximate maximum flow, where $s$ is the squared $\ell_2$ norm of the congestion vector of any optimal flow.
\end{theorem}

Our running time improves upon the previous fastest running time of this problem of $\tilde{O}(m/\eps)$. Since $s \leq m$ we achieve a faster running time whenever the graph is slightly dense, i.e. $m = \Omega(n^{1 + \delta})$ for any constant $\delta > 0$.

Interestingly our algorithm achieves even faster running times when there is a sparse maximum flow, i.e. a maximum flow in which the average path length in the flow decomposition of the optimal flow is small. Leveraging this, in Section~\ref{ssec:exactflows} we provide several new results on exact undirected and directed maximum flow on uncapacitated graphs as well.

\begin{theorem}[Improved algorithms for exact maximum flows]
There are algorithms for finding an exact maximum flow in the following types of uncapacitated graphs.
\begin{itemize}
\item There is an algorithm which finds a maximum flow in an undirected, uncapacitated graph with maximum flow value $F$ in time $\tilde{O}(m + \textup{min}(\sqrt{mn} F^{3/4}, m^{3/4} n^{1/4} \sqrt{F}))$.
\item There is an algorithm which finds a maximum flow in an undirected, uncapacitated graph with a maximum flow that uses at most $s$ edges in time $\tilde{O}(m + \sqrt{ms}n^{1/4}\max(n, s)^{1/4})$.
\end{itemize}
\end{theorem}
Each of these runtimes improves upon previous work in some range of parameters. For example, the bound of $\tilde{O}(m + m^{3/4} n^{1/4} \sqrt{F})$ for undirected, uncapacitated graphs improves upon the previous best running times of $\tilde{O}(m \sqrt{F})$ achievable by \cite{Sherman17} whenever $n = o(m)$ and of $\tilde{O}(m + n F)$ achievable by \cite{KargerL02} whenever $m = o(n F^{2/3})$. 

We also separately include the following result (which has no dependence on the sparsity $s$) for finding exact flows in general uncapcitated directed graphs, as it improves upon the running time of $\tilde{O}(m \cdot \max \{m^{1/2}, n^{2/3}\})$ achieved by \cite{GoldbergR98} whenever $m = \omega(n)$ and $m = o(n^{5/3})$.

\begin{theorem}[Exact maximum flow for directed uncapacitated graphs]
There is an algorithm which finds a maximum flow in a directed, uncapacitated graph in time $\tilde{O}(m^{5/4} n^{1/4})$. When the maximum flow is $s$-sparse, there is an algorithm which finds a maximum flow in a directed, uncapacitated graph in time $\tilde{O}(mn^{1/4}\max(n, s)^{1/4})$.
\end{theorem}

Although the runtime of \cite{GoldbergR98} has been improved by the recent works of \cite{Madry13} achieving runtime $O(m^{10/7})$ and of \cite{LeeS14} achieving runtime $\tilde{O}(m \sqrt{n})$, which dominate our $\tilde{O}(m^{5/4} n^{1/4})$ runtime, they do it using sophisticated advances in interior point methods, whereas our algorithm operates using a \emph{first-order method} which only queries gradient information of the objective function, rather than second-order Hessian information. In particular, our algorithm is the first to improve runtimes for directed graphs while relying only on first-order information of the objective function. We find it interesting that our result achieves any running time improvement for unit capacity maximum flow over \cite{GoldbergR98} without appealing to interior point machinery and think this may motivate further research in this area, namely designing first-order methods for structured linear programs.

\subsection{Previous work}
\label{sec:previous_work}

Here we embark on a deeper dive into the context of the problems and tools discussed in this paper.

\noindent{\bf Solving the $\ell_\infty$ regression problem.}
For a non-differentiable function such as $f(x) = \|x\|_\infty$, it is possible to use the toolkit for linear programming (including interior point and cutting plane \cite{LeeS14, LeeS15B}) to obtain iterative algorithms for approximate minimization. However, these particular algorithms have a larger dependence on dimension, and it is widely believed that the iteration complexity is inherently dimension-dependent. A first-order iterative algorithm with a better dependence on dimension for approximately solving the regression problem was developed by \cite{Nesterov05} and proceeds in two stages. First, the algorithm constructs a smooth approximation to the original function, which is typically explicitly derived via regularizing the dual function using a regularizer which is both smooth and bounded in range. The smooth approximation is constructed such that approximately minimizing the approximate function is sufficient to approximately minimize the original function. Second, a first-order method such as gradient descent in a particular norm, or one of its many variants, is applied to approximately minimize the smoothed function.

One of the earlier works to develop algorithms using first-order methods under this framework to solve the regression problem is \cite{Nesterov05}. One regularizer used in this work for optimization over a dual variable in the simplex was the entropy regularizer, which yields the smooth approximation to the $\ell_\infty$ norm defined by $\textrm{smax}_\alpha(x) = \alpha\log(\sum_j \exp(x_j/\alpha))$. Until recently, state-of-the-art gradient methods converged to an $\epsilon$-approximate solution in $O(\epsilon^{-2})$ or $O(\sqrt{m}\epsilon^{-1})$ iterations, hiding problem-specific dependencies on smoothness and domain size. The per-iteration cost of these methods involves computing a whole gradient, which incurs another multiplicative loss of dimension in runtime. 

Several other works which aimed to solve the regression problem via considering a smooth minimax formulation, including \cite{Nemirovski04} and \cite{Nesterov07}, incurred the same fundamental barrier in convergence rate. These works aimed to pose the (smooth) regression problem as finding the saddle point of a convex-concave function via a specially-constructed first-order method. The main barrier to improving prior work up to this point has been the inability to construct regularizers of small range which are strongly convex with respect to the $\ell_\infty$ norm. For some time, these issues posed a barrier towards finding faster algorithms for the regression problem, and many related problems.

Very recently, Sherman \cite{Sherman17} presented an alternative method which was able to break this barrier and attain an $O(1/\eps)$ iteration count for finding approximate solutions to the regression problem, where each iteration can be applied in time to compute a gradient. The algorithm used was a variation of Nesterov's dual extrapolation method \cite{Nesterov07} for approximately finding a saddle point in a convex-concave function, adapted to work for regularizers satisfying a weaker property known as area convexity, and an analysis of its convergence. As a corollary, this obtained the currently fastest-known algorithm for approximate maximum flow.%

\begin{table}
\begin{center}
 \begin{tabular}{|c|c|c|c|c|c|} 
 \hline
 \textbf{Year} & \textbf{Author} & \textbf{Method} & \textbf{Iteration Complexity} & \textbf{Iteration Cost} & \textbf{Norm} \\ [0.5ex] 
 \hline\hline
 2003 & \cite{Nesterov05} & Smoothing & $O(\epsilon^{-2})$ & $O(m)$ & $\ell_\infty$ \\ 
  & & & $O(\epsilon^{-1})$ & $O(m)$ & $\ell_2$ \\
 \hline
 2004 & \cite{Nemirovski04} & Mirror prox & $O(\epsilon^{-2})$ & $O(m)$ & $\ell_\infty$ \\
 &  & & $O(\sqrt{m}\epsilon^{-1})$ & $O(m)$ & $\ell_\infty$ \\
 \hline
 2005 & \cite{Nesterov07} & Dual extrapolation & $O(\epsilon^{-2})$ & $O(m)$ & $\ell_\infty$ \\
  &  & & $O(\sqrt{m}\epsilon^{-1})$ & $O(m)$ & $\ell_\infty$ \\
 \hline
 2017 & \cite{Sherman17} & Area-convexity & $O(\epsilon^{-1})$ & $O(m)$ & $\ell_\infty$ \\
 \hline
 2018 & This paper & Local smoothness & $O( \sqrt{m}\epsilon^{-1})$ & $\tilde{O}(c)$ & $\ell_2$ \\ %
 \hline
\end{tabular}
\end{center}
\caption{Dependencies of algorithms for $\ell_\infty$ regression in $A \in \mathbb{R}^{n \times m}$ on various problem parameters. Note that there is up to an $O(\sqrt{m})$ discrepancy between the $\ell_2$ and $\ell_\infty$ norms. Here, $c$ is the maximum number of nonzero entries in any column of $A$.
}
\end{table}

\noindent{\bf Abbreviated history of first-order methods, emphasizing coordinate-based methods.}
First-order methods for convex optimization have a long history. Gradient descent methods with error decaying in $k$ iterations as $O(1/\sqrt{k})$ for Lipschitz functions and $O(1/k)$ for smooth functions have been well studied (for example, see \cite{Nesterov03} or \cite{Bubeck15} for a more detailed exposition), and applied in many important settings.

Nesterov gave the first gradient-based algorithm for minimizing functions smooth in the Euclidean norm which converged at the rate $O(1/k^2)$. The method is optimal in the sense that it matched known lower bounds for smooth functions. Unfortunately, this method does not apply generically to functions which are smooth in other norms, in the same way that unaccelerated variants do, without possibly paying an additional dependence on the dimension. In particular, the accelerated convergence rate depends on the regularizer that the mirror descent steps use, and thus the analysis incurs a loss based on the size of the regularizer, which is the barrier in the aforementioned $\ell_\infty$-smooth function case. Specifically, it is a folklore result that any function strongly-convex over $[-1, 1]^n$ in the $\ell_\infty$ norm has range at least $n/2$, which we show in \Cref{sec:folklore}.

There has been much interest in applying \emph{randomized} first order methods to more efficiently obtain an approximate minimizer on expectation, when the convex optimization problem has certain structure. One example of these randomized methods in the literature is coordinate descent, studied first in \cite{Nesterov12}. The main idea is that using crude, computationally efficient, approximations to the full gradient, one is still able to find an approximate minimizer on expectation. One benefit is that coordinate descent admits a more fine-grained analysis of convergence rate, based on structural properties of the function, i.e. the smoothness of the function in each coordinate.

Generalizations of standard coordinate descent have received much attention recently, both for their powerful theoretical and practical implications. \cite{Nesterov12} provided an accelerated version of the standard coordinate descent algorithm, but the naive implementation of its steps were inefficient, taking linear time in the dimension. The study of efficient accelerated coordinate descent methods (which converge at the rate $O(1/k^2)$ without an additional dependence on dimension) was pioneered by \cite{LeeS13}, and since then a flurry of other works, including \cite{FercoqR15, ZhuQRY16, QuR16} have improved the rate of convergence and generalized the methods to composite functions with a separable composite term, of the form $F(x) = f(x) + \sum_j \psi_j(x_j)$. We remark that our box constraint can be represented as such a separable composite term in the objective, and our constrained accelerated coordinate descent algorithm is an adaptation of such composite methods. For a more detailed history of the study of coordinate descent methods, we refer the reader to \cite{FercoqR15}.

Accelerated coordinate based methods have proven to be useful in many ways when applied to problems in theoretical computer science. For example, the authors of \cite{LeeS13} framed graph Laplacian system solvers as a coordinate descent problem to give better runtime guarantees. One particularly interesting example that highlighted the potential for using accelerated coordinate descent in minimizing entropy-based functions was the work of \cite{ZhuO15} in solving packing and covering LPs, where the constraint matrix is nonnegative, in which they also attained a $O(1/\epsilon)$ method complexity. Conventional wisdom is that these results are specific to the structure of the particular problem, so any exploration of accelerated methods in greater generality is particularly interesting.

\noindent{\bf Maximum flow.}
The maximum flow problem is a fundamental problem in combinatorial optimization that has been studied extensively for several decades. Until recently, the toolkit used to solve the problem has been primarily combinatorial, culminating in algorithms with runtime roughly $\tilde{O}(\textrm{min}\{mn^{2/3}, m^{3/2}\})$ for finding a maximum flow in graphs with $m$ edges and $n$ vertices and polynomially bounded capacities \cite{GoldbergR98}, and $\tilde{O}(m + nF)$ for finding a maximum flow in undirected graphs with $m$ edges, $n$ vertices, and a maximum flow value of $F$ \cite{KargerL02}.

Breakthroughs in the related problem of electrical flow using tools from continuous optimization and numerical linear algebra were first achieved by Spielman and Teng \cite{SpielmanT04} who showed that solving a linear system in the Laplacian of a graph could be done in nearly linear time, which is equivalent to computing an electrical flow.
\newpage

Notably, the electric flow problem corresponds to approximately solving an $\ell_2$ regression problem $\norm{Ax - b}_2$, and the maximum flow problem corresponds to approximately solving an $\ell_\infty$ regression problem $\norm{Ax - b}_\infty$. Accordingly, using the faster algorithms for electric flow combined with a multiplicative weights approach, the authors of \cite{ChristianoKMST11} were able to make a breakthrough to approximately solve maximum flow with a runtime of $\tilde{O}(mn^{1/3})$, where $\tilde{O}$ hides logarithmic factors. Finally, using constructions presented in \cite{Madry10}, the authors of \cite{Sherman13} and \cite{KelnerLOS14} were able to reduce this runtime to almost linear, essentially using variants of preconditioned gradient descent in the $\ell_\infty$ norm. This runtime was reduced to $\tilde{O}(m/\epsilon^2)$ by Peng in \cite{Peng16} by using a recursive construction of the combinatorial preconditioner. As previously mentioned, the $\epsilon^{-2}$ dependence in the runtime was a barrier typical of algorithms for minimizing $\ell_\infty$-smooth functions without worse dimension dependence, and was broken in \cite{Sherman17}, who attained a runtime of $\tilde{O}(m/\epsilon)$.
\begin{table}
\begin{center}
 \begin{tabular}{|c|c|c|c|c|} 
 \hline
 \textbf{Year} & \textbf{Author} & \textbf{Complexity} & \textbf{Weighted} & \textbf{Directed} \\ [0.5ex] 
 \hline\hline
 1998 & \cite{GoldbergR98} & $\tilde{O}(\min(m^{3/2}, mn^{2/3}))$ & Yes & Yes \\
 \hline
 1998 & \cite{Karger98} & $\tilde{O}(m\sqrt{n} \epsilon^{-1})$ & Yes & No \\
 \hline
 2002 & \cite{KargerL02} & $\tilde{O}(m + nF)$ & Yes & No \\
 \hline
 2011 & \cite{ChristianoKMST11} & $\tilde{O}(mn^{1/3}\epsilon^{-11/3})$ & Yes & No \\
 \hline
 2012 & \cite{LeeRS13} & $\tilde{O}(mn^{1/3}\epsilon^{-2/3})$ & No & No \\
 \hline
 2013 & \cite{Sherman13}, \cite{KelnerLOS14} & $\tilde{O}(m^{1 + o(1)}\epsilon^{-2})$ & Yes & No \\
 \hline
 2013 & \cite{Madry13} & $\tilde{O}(m^{10/7})$ & No & Yes \\
 \hline
 2014 & \cite{LeeS14} & $\tilde{O}(mn^{1/2})$ & Yes & Yes \\
 \hline
 2016 & \cite{Peng16} & $\tilde{O}(m\epsilon^{-2})$ & Yes & No \\
 \hline
 2017 & \cite{Sherman17} & $\tilde{O}(m\epsilon^{-1})$ & Yes & No \\
 \hline
 2018 & This paper & $\tilde{O}(m + (n + \sqrt{ns})\epsilon^{-1})$ & Yes & No \\
 \hline
 
\end{tabular}
\end{center}
\caption{Complexity of maximum flow since \cite{GoldbergR98} for undirected graphs with $n$ vertices, $m$ edges, where $s$ is the $\ell_2^2$ of the maximum flow's congestion, and $F$ is the maximum flow value.}
\end{table}

\subsection{Revision since initial publication}

The original version of this manuscript claimed a runtime of $\tilde{O}\left(m + \sqrt{ns}/\eps\right)$ for the approximate maximum flow problem. Since its original conference publication, a mistake in the analysis of the accelerated coordinate descent method used, under the dynamic sampling scheme based on local coordinate smoothnesses, was pointed out to us by Kent Quanrud. The mistake was in the modification of the analysis of the accelerated method of \cite{QuR16}, in which the iterates of the algorithm were shown to be a convex combination of prior iterates; under dynamic sampling probabilities, this may no longer be the case. In this revision, we show that a modification using our original algorithm, under a proximal-point reduction inspired by the extragradient algorithm of \cite{Nemirovski04}, yields a runtime of $\tilde{O}(m + \max(n, \sqrt{m\min(n, s)})/\eps)$. This algorithm retains the improvement upon the state-of-the-art approximate maximum flow runtimes for slightly-dense graphs, and has an improved complexity for the more general problem of box-constrained $\ell_\infty$ regression in terms of the dependence on the column sparsity $c$, improving the dependence from $c^{2.5}$ to $c$. 

 Moreover, we provide a randomized primal-dual algorithm, more closely related to the algorithms of \cite{Nemirovski04, Nesterov07, Sherman17}, obtaining a runtime of $\tilde{O}(m + \max(n, \sqrt{ns})/\eps)$, i.e. the originally claimed runtime for flow sparsities at least $n$. This algorithm builds upon our local smoothness-based sampling scheme, and introduces several new algorithmic and analytic techniques, including a ``locally variance-reduced'' randomized extragradient method which preserves the $\eps^{-1}$ convergence rate, and a data structure which allows for entry queries and sampling from a simplex variable in nearly-constant time, under structured dense updates. We believe these contributions will be of independent interest to the community, and hope that they will find use in designing further improved algorithms for $\ell_\infty$ regression and related problems.
 
 Some of the ideas used in developing our revised algorithms were inspired by the approach of our independent work \cite{CarmonJST19} with our collaborators, Yair Carmon and Yujia Jin.

\subsection{Organization}
The rest of this paper is organized as follows. Many proofs are deferred to the appendices.
\begin{itemize}
\item \textbf{Section~\ref{sec:overview}: Overview.} We introduce the definitions and notation we use throughout the paper, and give a general framework motivating our work.
\item \textbf{Section~\ref{sec:regression}: Regression.} We first give a framework for accelerated randomized algorithms which minimize the box-constrained $\ell_\infty$ regression function based on uniform sampling, as well as a faster one based on non-uniform sampling which assumes access to a coordinate smoothness and sampling oracle. To do so, we develop a new analysis of coordinate descent under a box constraint, amenable to dynamic coordinate sampling distributions, and show how to accelerate it via a primal-dual proximal point method. We then give efficient implementations for these oracles for structured problems. 
\item \textbf{Section~\ref{sec:maxflow}: Maximum Flow.} We state the reduction from the maximum flow problem to box-constrained $\ell_\infty$ regression problem. We first show how to attain a faster algorithm for maximum flow by  exploiting combinatorial structure of the flow regression problem, using the regression algorithm we developed in the prior section. We then state the improved runtimes which follow from a randomized primal-dual variation of our regression algorithm, given in Section~\ref{sec:improvement}. Further, we give the exact maximum flow runtimes achieved via rounding the resulting approximate flow of our improved method.
\item \textbf{Section~\ref{sec:improvement}: Primal-Dual Coordinate Acceleration.} We develop an algorithm with improved runtimes for the structured $\ell_\infty$ regression problem which results from the maximum flow reduction, and correspondingly yields further-improved flow runtimes. 
\end{itemize}

\section{Overview}
\label{sec:overview}

\subsection{Basic definitions}

First, we define some basic objects and properties which we use throughout this paper.

\noindent {\bf General Notations.}
We use $\tilde{O}(f(n))$ to denote runtimes of the following form: $O(f(n) \log^c f(n))$ where $c$ is a constant. With an abuse of notation, we let $\tilde{O}(1)$ denote runtimes hiding polynomials in $\log n$ when the variable $n$ is clear from context, and refer to such runtimes as ``nearly constant.''  

Generally, we work with functions whose arguments are vector-valued variables in $m$-dimensional space, and may depend on a linear operator $A: \mathbb{R}^m \rightarrow \mathbb{R}^n$. Correspondingly we use $j \in [m]$ and $i \in [n]$ to index into these sets of dimensions, where $[m]$ is the set $\{1, 2, \ldots m\}$. We use $e_j$ to denote the $j$th standard basis vector, i.e. the vector which is 1 in dimension $j$ and 0 everywhere else. We use $u \circ v$ to denote the vector which is the coordinate-wise product, i.e. its $j^{th}$ coordinate is $u_j v_j$. 

\noindent {\bf Matrices.}
In this work, we deal with matrices $A \in \mathbb{R}^{n \times m}$ unless otherwise specified. Accordingly, we index into rows of $A$ with $i \in [n]$, and into columns with $j \in [m]$. We refer to rows of $A$ via $A_{i:}$ or $a_i$ when it is clear from context, and columns via $A_{:j}$. We use $\nnz(A)$ to denote the number of nonzero entries of $A$, and assume $\nnz(A) \ge n + m - 1$, else we may drop a row or column.

We use $\textrm{diag}(w)$ to denote the diagonal matrix whose diagonal entries are the coordinates of a vector $w$. We call a square symmetric matrix $A$ positive semi-definite if for all vectors $x$, $x^\top A x \geq 0$ holds. For positive semi-definite matrices $A, B$ we apply the Loewner ordering and write $A \preceq B$ if for all vectors $x$, $x^\top A x \leq x^\top B x$ holds.

Finally, we say that a matrix is $c$-column-sparse if no column of $A$ has more than $c$ nonzero entries.

\noindent {\bf Norms.}
We use $\| \cdot \|$ to denote an arbitrary norm when one is not specified. For scalar valued $p \geq 1$, including $p = \infty$, we use $\|x\|_p \defeq (\sum_j x_j^p)^{1/p}$ to denote the $\ell_p$ norm. For vector valued $w \in \mathbb{R}_{\geq 0}^m$, we use $\|x\|_w^2 \defeq \sum_j w_j x_j^2$ to denote the weighted quadratic norm, and for positive semidefinite matrix $A$, we define $\norm{x}_A^2 = x^\top A x$. Further, we let $\Delta^n$ be the simplex in $n$ dimensions, e.g. $p \in \Delta^n \iff \norm{p}_1 = 1,$ $p \geq 0$ entrywise.

For a norm $\| \cdot \|$, the dual norm $\| \cdot \|_*$ is defined by $\|x\|_* \defeq \textrm{max}_{\|y\| \leq 1} y^\top x$. It is well known that the dual norm of $\ell_p$ is $\ell_q$ for $1/p + 1/q = 1$. For matrix $A$ and a vector norm $\| \cdot \|$, we define the matrix norm $\|A\| \defeq \max_{\|x\| = 1} \|Ax\|$. For example, $\|A\|_\infty$ is the largest $\ell_1$ norm of a row of $A$.

\noindent {\bf Functions.}
We will primarily be concerned with minimizing convex functions $f(x)$ subject to the argument being restricted by a box constraint, where the domain is some scaled box $B^c_\infty$ unless otherwise specified. Whenever the function is clear from context, $x^*$ will refer to any minimizing argument of the function. We use the term $\epsilon$-approximate minimizer of a function $f$ to mean any point $x$ such that $f(x^*) \leq f(x) \leq f(x^*) + \epsilon$. Furthermore, we define the $\textsf{OPT}$ operator to be such that $\textsf{OPT}(f)$ is the optimal value of $f$, when this optimal value is well-defined. 

For differentiable functions $f$ we let $\nabla f(x)$ be the gradient and let $\nabla^2 f(x)$ be the Hessian. We let $\nabla_j f(x)$ be the value of the $j^{th}$ partial derivative; we also abuse notation and use it to denote the vector $\nabla_j f(x) e_j$ when it is clear from context. 

\noindent {\bf Properties of functions.}
We say that a function is $L$-smooth with respect to some norm $\| \cdot \|$ if it obeys $\|\nabla f(x) - \nabla f(y)\|_* \leq L \| x - y \|$, the dual norm of the gradient is Lipschitz continuous. It is well known in the optimization literature that when $f$ is convex, this is equivalent to $f(y) \leq f(x) + \nabla f(x)^\top (y - x) + \frac{L}{2}\|y - x\|^2$ for $y, x \in \textrm{dom}(f)$ and, for twice-differentiable $f$, $y^\top \nabla^2 f(x) y \leq L \|y\|^2$.

We say that a function is $L_j$-coordinate smooth in the $j^{th}$ coordinate if the restriction of the function to the coordinate is smooth, i.e. $|\nabla_j f(x + ce_j) - \nabla_j f(x)| \leq L_j |c|$ $\forall x \in \textrm{dom}(f), c \in \mathbb{R}$. Equivalently, for twice-differentiable convex $f$, $\nabla^2_{jj} f(x) \leq L_j$. 

Finally, we say a function is $\mu$-strongly convex with respect to $\| \cdot \|$ if for all $x, y$, $f(y) \ge f(x) + \nabla f(x)^\top(y-x) + \frac{\mu}{2}\|y - x\|^2$. When $f$ is twice-differentiable, equivalently $y^\top \nabla^2 f(x) y \ge \mu \norm{y}^2$. 

\noindent {\bf Graphs.}
We primarily study capacitated undirected graphs $G = (V, E, u)$ with edge set $E \subseteq V \times V$, edge capacities $u: E \to \R_+$. When referring to graphs, we let $m = |E|$ and $n = |V|$. Throughout this paper, we assume that $G$ is strongly connected. 

We associate the following matrices with the graph $G$, when the graph is clear from context. The matrix of edge weights $U \in \R^{E \times E}$ is defined as $U \defeq \textbf{diag}(u)$. Orienting the edges of the graph arbitrarily, the vertex-edge incidence matrix $B \in \R^{V \times E}$ is defined as $B_{s, (u,v)} \defeq -1$ if $s=u$, $1$ if $s=v$ and $0$ otherwise.

\noindent {\bf Divergences.}
In the analysis of mirror descent variants, a first-order method flexible to geometric constraints on its arguments, we require the concept of a Bregman divergence with respect to a regularizer $r$. For a convex function $r$, we define the (nonnegative) Bregman divergence to be
\[V^r_x(y) = r(y) - r(x) - \nabla r(x)^\top(y - x).\]
We drop the $r$ for convenience when it is clear from context. The Bregman divergence satisfies the well-known equality
\begin{equation}\label{eq:threepoint}\inprod{\nabla V^r_x(y)}{u - y} = V^r_x(u) - V^r_x(y) - V^r_y(u).  \end{equation}

\subsection{Overview of our algorithms}

Here, we give an overview of the main ideas used in our algorithms for approximately solving $\ell_\infty$ regression problems. The main ideological contribution of this work is that it uses a new variation of coordinate descent which uses the novel concept of \emph{local coordinate smoothness} in order to get tighter guarantees for accelerated algorithms. 

\subsubsection{$\ell_\infty$ regression algorithm}

The first piece of our algorithm is developed in Section~\ref{ssec:accelanalysis}, where we show how to use a \emph{primal-dual proximal point method} inspired by the ``conceptual mirror-prox'' algorithm of \cite{Nemirovski04} to reduce the task of designing an accelerated scheme for the $\ell_\infty$ regression problem to designing an unaccelerated procedure for minimizing a regularized approximation of the regression objective. Next, we show in Section~\ref{ssec:accprox} how to improve the complexity of the standard coordinate descent algorithm for an appropriately regularized $\ell_\infty$-smooth approximation to the regression problem by using the concept of local coordinate smoothnesses, which we introduce. To analyze its convergence, we develop a novel analysis of coordinate descent under dynamic sampling probabilities subject to a box constraint. Finally, in order to implement the steps of the algorithm, it is necessary to efficiently compute overestimates to the local coordinate smoothnesses, and furthermore sample coordinates proportional to these overestimates; this procedure is given in \Cref{sec:fastimpl}. 

\noindent{\bf Acceleration via proximal point reduction.} In Section~\ref{ssec:accelanalysis}, we show how we can reduce minimizing the original $\ell_\infty$ objective to efficiently finding high-precision minimizers to a sequence of regularized approximations, via a proximal scheme of \cite{Nemirovski04}, which we refer to as the primal-dual proximal point method, or proximal point method for short.\footnote{The proximal point method in this paper is slightly different than the ``conceptual mirror-prox'' algorithm of \cite{Nemirovski04}. In \cite{Nemirovski04}, each iteration takes two steps, the first of which solves a regularized proximal problem to sufficiently high accuracy, and the second of which is an extragradient adjustment step. We bypass the need for this adjustment step via more stringent requirements on the accuracy level of the solution of the proximal problem.} This reduction constructs a sequence of iterates by calling a high-precision minimization oracle for each regularized approximation, where the regularization amount is parameterized by a scalar quantity $\alpha > 0$. A larger $\alpha$ will result in simpler subproblems, but will require more calls to the oracle; trading off these complexities via the parameter $\alpha$ results in our accelerated runtime. More formally, note that we may rewrite the original regression problem by introducing a dual variable (after appropriately doubling the constraints to account for signs; see discussion in Section~\ref{ssec:smoothapprox})
\[\min_{x \in [-1, 1]^{m}} \norm{Ax - b}_\infty = \min_{x \in [-1, 1]^m} \max_{p \in \Delta^n} p^\top (Ax - b).\]
The proximal point method with parameter $\alpha$ constructs a sequence of points $\{z_t\}$ as follows:
 from an iterate $z_t = (x_t, p_t)$, define the next iterate $z_{t + 1} = (x_{t + 1}, p_{t + 1})$ as the solution to a proximal subproblem (throughout, $s \defeq \norm{x^*}_2^2$ where $x^*$ is the optimizer of the box-constrained $\ell_\infty$ regression).
\begin{equation}\label{eq:ztp1def}
z_{t + 1} = \argmin_{x \in [-1, 1]^m}\; \argmax_{p \in \Delta^n}\; p^\top (Ax - b) + \frac{\alpha}{2s}\norm{x - x_t}_2^2 - \alpha \sum_i p_i \log \frac{p_i}{[p_t]_i}.
\end{equation}
To explain further, the most prevalent first-order method approach to convex optimization, and its primal-dual generalization (for example found in mirror descent and gradient descent) for solving a problem of the form $\min_{x \in [-1, 1]^m} \max_{p \in \Delta^n} p^\top (Ax - b)$ with gradient operator $g(x, p)$, is to repeatedly construct regularized linearizations of the form, for some regularizer function $r$,
\[z_{t + 1} = \argmin_z \; \inprod{g(z_t)}{z} + V^r_{z_t}(z).\]
The proximal method instead sets the next iterate $z_{t + 1}$ to be the result of a proximal problem, without the linearization; we set the regularizer $r(x, p)$ to be $\frac{1}{2s}\norm{x}_2^2 + \sum_i p_i \log p_i$. Overall, if the regularizer $r$ has range bounded by $\Theta$, then the proximal point method converges in roughly $\alpha\Theta/\eps$ iterations to an $\eps$-approximate saddle point, which suffices for our purposes.

We give the convergence analysis of the proximal point method under approximate solutions to the subproblems defining the iterates $\{z_t\}$ in Section~\ref{ssec:accelanalysis}. Therefore, the main algorithmic workhorse can be reduced to computing high-accuracy saddle points to problems of the form 
\begin{equation}\label{eq:subproblem}\argmin_{x \in [-1, 1]^m}\; \alpha \log \sum_{i \in [n]}\exp\left(\frac{1}{\alpha}\left[Ax - b_t\right]_i \right) + \frac{\alpha}{2s}\norm{x - x_t}_2^2.\end{equation}
Note that the problem \eqref{eq:subproblem} is the same as \eqref{eq:ztp1def}, where we maximized over $p$ explicitly; the vector $b_t$ is obtained via a linear shift of the vector $b$ (details can be found in Section~\ref{ssec:accelanalysis}). Our remaining algorithmic development deals with this subproblem; combining a fast iterative method for this subproblem with the optimal choice of $\alpha$ yields the runtime for regression. To obtain our more fine-grained runtimes in Section~\ref{sec:maxflow}, we also generalize to diagonally-reweighted $\ell_2^2$ regularizers.

\noindent {\bf Local coordinate smoothness.}
In this work, we introduce the concept of \emph{local coordinate smoothness} at a point $x$. This generalizes the concept of global coordinate smoothness to a particular point. This definition is crucial to the analysis throughout the rest of the paper.

\begin{definition}[Local coordinate smoothness]
\label{def:lcs}
Twice-differentiable function $f$ is $L_j(x)$ locally coordinate smooth in coordinate $j$ at $x$, if for all $|c| \leq \left|\nabla_j f(x)/L_j(x)\right|$, $\nabla^2_{jj} f(x + ce_j) \le L_j(x)$.
\end{definition}

We state a useful equivalent characterization to Definition~\ref{def:lcs}; the proof is standard and follows by integration (once and twice respectively). 

\begin{lemma}
\label{lem:lcsequiv}
For twice-differentiable $f$, $f$ is $L_j(x)$ locally coordinate smooth if and only if $|\nabla_j f(y) - \nabla_j f(y')| \le L_j(x)|y - y'|$ for all $y$, $y'$ between $x \pm \nabla_j f(x)/L_j(x)e_j$. If $f$ is $L_j(x)$ locally coordinate smooth then for all $y$ between $x \pm \nabla_j f(x)/L_j(x) e_j$, $f(y) \leq f(x) + \nabla f_j(x) (y_j - x_j) + \frac{L_j(x)}{2}|y_j - x_j|^2$. 
\end{lemma}

Note that this says that a coordinate descent step using local smoothnesses at a point exhibits roughly the same behavior as a single step of coordinate descent with global smoothnesses. In particular, for the point which the coordinate descent algorithm would step to, the function values exhibit the same quadratic upper bound along the coordinate. For a more motivating discussion of this definition, we refer the reader to an analysis of coordinate descent presented in \Cref{sec:cdproof}. We will drop the $x$ from the notation $L_j(x)$ when the point we are discussing is clear, i.e. a particular iterate of one of our algorithms.

\noindent {\bf Bounding the progress of coordinate descent in $\ell_\infty$-smooth functions.}
Here, we sketch the main idea underlying our improved runtime for the problem \eqref{eq:subproblem}, whose first component is $\ell_\infty$-smooth. Why is it possible to hope to improve gradient methods in the $\ell_\infty$ norm via coordinate descent? One immediate reason is that smoothness in this norm is a strong assumption on the sum $S$ of the local coordinate smoothness values of $f$.

As we recall in \Cref{appendix:a}, gradient descent for an $\ell_\infty$-smooth function initialized at $x^0 \in \mathbb{R}^m$ takes roughly $\frac{L\|x^0 - x^*\|_\infty^2}{\epsilon}$ iterations to converge to a solution which has $\epsilon$ additive error, whereas coordinate descent with appropriate sampling probabilities $\frac{L_j}{S}$, for $S = \sum_j L_j$, takes $\frac{S\|x^0 - x^*\|_2^2}{\epsilon}$ iterations to converge to the same quality of solution.

When the norm in the gradient descent method is $\| \cdot \|_{\infty}$, we have $\|x^0 - x^*\|_2^2 \leq m \|x^0 - x^*\|_{\infty}^2$, but the iterates can be $m$ times cheaper because they do not require a full gradient computation. So, if we can demonstrate $S \leq L$, we can hope to match and improve the runtime. To be more concrete, we will demonstrate the following fact.

\begin{lemma}
\label{lem:boundedtrace}
Suppose for some point $x$, $f : \R^{m} \rightarrow \R$ is convex and $L$-smooth with respect to $\| \cdot \|_{\infty}$, $\Lambda_j(x) = \nabla^2_{jj} f(x)$, and $S = \sum_j \Lambda_j(x)$. Then $S \leq L$. 
\end{lemma}

\begin{proof}
Fix $x$, and define $M \defeq \nabla^2 f(x)$ and $S \defeq \textrm{Tr}(M)$. Consider drawing $y$ uniformly at random from $\{-1, 1\}^m$. By the smoothness assumption, we have $y^\top M y \leq L\|y\|^2_{\infty} = L$. Also, note that
\begin{equation*}
\E[y^\top M y] = \E \left[\sum_{i, j} M_{ij} y_i y_j \right] = \textrm{Tr}(M) = S
\end{equation*}
Thus, by the probabilistic method, there exists some $y$ such that $S \leq y^\top M y \leq L$, as desired.
\end{proof}

While this gives a bound on the number of iterations required by a coordinate descent algorithm, it requires being able to compute and sample by the $L_j(x)$; as we take coordinate descent steps, it is not clear how the local coordinate smoothnesses $L_j(x^k)$ will change, and how to update and compute them. Naively, at each iteration, we could recompute the local smoothnesses, but this requires as much work as a full gradient computation if not more. Furthermore, we need to implement sampling the coordinates in an appropriate way, and show how the algorithm behaves under acceleration. However, a key idea in our work is that if we can take steps within regions where the smoothness values do not change by much, we can still make iterates computationally cheap, which we will show. 

\noindent{\bf Box-constrained coordinate descent under dynamic sampling.} One technical difficulty that arises in the analysis of coordinate descent methods under local coordinate smoothnesses is the fact that the sampling distribution changes from iteration to iteration. In prior analyses of coordinate descent subject to a separable convex (i.e. box) constraint \cite{FercoqR15, QuR16}, a key technical fact of the iterates was the fact that they could be written as a convex combination of prior iterates. Under dynamic sampling distributions, this may no longer be the case. In this work, we give a new analysis of coordinate descent under a box constraint, and show that the progress of each iteration can be directly analyzed by using the geometry of the box constraint. We develop this analysis in Section~\ref{ssec:accprox}, and combining it with our local coordinate smoothness analysis yields the faster oracle for minimizing problem \eqref{eq:subproblem}.

\noindent {\bf Implementation of local smoothness estimates.}
One useful property of coordinate descent is that as long as we implement the algorithm with overestimates to the local smoothness values, the convergence rate scales with the sum of the overestimates. Our full algorithm for solving \eqref{eq:subproblem} proceeds by showing how to compute and sample proportional to slight overestimates to the local smoothnesses, for regression problems in a column-sparse matrix. We do so by first proving that the smooth approximation to $\ell_\infty$ regression admits local smoothnesses which can be bounded in a structured way, in Section~\ref{ssec:accprox}. Further, using a lightweight data structure, we are able to maintain these overestimates and sample by them in nearly-constant time, yielding a very efficient implementation, which we show in Section~\ref{ssec:cheapiter}.

\subsubsection{Maximum flow algorithm}

In Section~\ref{sec:maxflow}, we study the maximum flow problem as an example of a problem which can be reduced to $\ell_\infty$ regression in a column-sparse matrix. We first describe a reduction from approximate maximum flow to structured instances of $\ell_\infty$ regression, already-present in the literature \cite{Sherman13, KelnerLOS14, Peng16}. We first show that a direct application of our accelerated $\ell_\infty$ regression algorithm yields the fastest currently known approximate maximum flow algorithm, roughly giving a runtime of $\tilde{O}(m + (n + \sqrt{ms})/\eps)$. We also show that a slight modification of our accelerated regression algorithm, where the norm we measure smoothness and strong-convexity of the box-constrained variable is weighted by columns of the matrix, yields a runtime of $\tilde{O}(m + \sqrt{mn}/\eps)$, generically improving upon the runtime of \cite{Sherman17} for slightly-dense graphs.

Finally, in Section~\ref{sec:improvement}, we show that by opening up the algorithm further into a fully primal-dual method, we can use a novel analysis of a variance-reduced mirror prox method based on local coordinate smoothness estimates in order to obtain an improved runtime of $\tilde{O}(m + (n + \sqrt{ns})/\eps)$. Our randomized mirror prox method requires the development of a somewhat more-complicated data structure, based on efficient polynomial approximations to the exponential, in order to approximately query and sample from a simplex variable under dense updates.  %
\section{Minimizing $\|Ax - b\|_\infty$ subject to a box constraint}
\label{sec:regression}

We now show how to turn the framework presented in the previous section into improved algorithms for the problem of box-constrained regression in the $\ell_\infty$ norm. Recall that our goal is to compute an $\epsilon$-approximate minimizer of the constrained $\ell_\infty$ regression problem with a $O(1/\epsilon)$ method complexity (see Definition~\ref{def:boxedlinfreg}).

In the style of previous approaches to solving $\ell_\infty$ regression, because $\|x\|_\infty$ is not a smooth function, we choose to minimize a suitable smooth approximation instead. Intuitively, the $O(1/\epsilon)$ rate comes from accelerating gradient descent for a function which is $O(1/\epsilon)$-smooth. One would then expect the function error of the $T^{th}$ iterate with respect to $\textsc{OPT}$ is proportional to $(1/\eps)/T^2$, so if we wish for an $\epsilon$-approximate minimizer, it suffices to pick $T = O(1/\epsilon)$. Because our method is not a typical accelerated method, and is instead based on reducing the proximal point method to solving a series of subproblems \eqref{eq:subproblem}, the runtime analysis proceeds somewhat differently. We will show (roughly speaking) how to solve a subproblem of type \eqref{eq:subproblem} in
\[\tilde{O}\left(m + \frac{\min(m, n)}{\alpha} + \frac{s}{\alpha^2}\right)\]
iterations, where each iteration can be implemented in time $\tilde{O}(c)$, where $c$ is the maximum number of nonzero entries in any column of $A$. Because each problem \eqref{eq:subproblem} results from a regularization based on a regularizer $r$ of nearly-constant range, it suffices to solve $\tilde{O}(\alpha/\eps)$ such problems to yield an $\eps$-approximate solution. Finally, each reduction to the subproblem is complemented by an extragradient step, which takes time $O(\nnz(A))$. The accelerated runtime is then roughly 
\[\tilde{O}\left(\left(\nnz(A) + \frac{\min(m, n)}{\alpha} + \frac{s}{\alpha^2}\right)\frac{\alpha}{\eps}\right) = \tilde{O}\left(\nnz(A) + \frac{\min(m, n) + \sqrt{\nnz(A)s}}{\eps}\right),\]
where the choice of $\alpha$ was to appropriately balance the terms. 

\subsection{Constructing the smooth approximation to regression}
\label{ssec:smoothapprox}
In this section, we define the smooth approximation for $\ell_\infty$ regression we use through the paper and provide some technical facts about this approximation. Note that these approximations are standard in the literature. First, we define the $\textrm{smax}$ function which is used throughout. This function is smooth in the $\ell_\infty$ norm, which can be seen because it is the result of the following conjugate problem
\[\max_{p \in \Delta^n} \inprod{p}{x} - \alpha\sum_i p_i \log p_i;\]
because the function $r(p) = \sum_i p_i \log p_i$ is 1-strongly convex in the $\ell_1$ norm, its dual, the softmax function, is smooth in the $\ell_\infty$ norm.

\begin{definition}[Softmax] For all real valued vectors $x$ we let $\smax_\alpha(x) \defeq \alpha \log(\sum_j \exp(\frac{x_j}{\alpha}))$.
\end{definition}

\begin{fact}[Softmax additive error]
\label{fact:adderr}
$\forall x \in \mathbb{R}^m$, $\max_{j \in [m]} x_j \leq \smax_\alpha(x) \leq  \alpha\log m + \textup{max}_{j \in [m]} x_j$.
\end{fact}
\begin{proof}
It follows from monotonicity of $\log$ and positivity of $\exp$: letting $j^*$ be the maximal index of $x$, $\smax_\alpha(x) \geq \alpha \log(\exp(x_{j^*}/\alpha)) = x_{j^*}$, and $\smax_\alpha(x) \leq \alpha\log(m \exp(x_{j^*}/\alpha)) = \alpha \log m + x_{j^*}$.
\end{proof}

Note that these properties are about the quality of approximation $\textrm{smax}$ provides on the maximum element of a vector, instead of its $\ell_{\infty}$ norm. To apply this to an $\ell_\infty$ objective, we used the standard reduction of applying it to the regression problem in twice the original dimension, defined with a proxy matrix $A' = \tvect{A}{-A}$ and a proxy vector $b' = \tvect{b}{-b}$. For notational convenience, we will focus on minimizing $f(x)$ defined above, but with $A \in \mathbb{R}^{n \times m}$ and $b \in \mathbb{R}^n$ in the original dimensionalities, which preserves all dependencies on the dimension and structural sparsity assumptions used later in this work up to a constant. Next, we state some technical properties of our approximation. We drop the $\alpha$ from many definitions because the $\alpha$ we choose for all our methods is fixed.

\begin{definition}
For $x \in \mathbb{R}^m$ let $p(x) \in \mathbb{R}^m$ be defined as $p_j(x) \defeq \frac{\exp(x_j/\alpha)}{\sum_{j'} \exp(x_{j'}/\alpha)}$.
\end{definition}

Note that for any $x$ the above  $p_j(x)$ form a probability distribution. Moreover, they are defined in this way because they directly are used in the calculation of the gradient and Hessian of $\smax$. The following facts can be verified by direct calculation.

\begin{fact}[Softmax calculus]
\label{fact:smaxfacts}
$\nabla \smax_\alpha(x) = p(x)$, $0 \preceq \nabla^2 \smax_\alpha (x) \preceq \alpha^{-1} \textup{\textbf{diag}}(p(x))$.
\end{fact}

\subsection{Acceleration via proximal point method}
\label{ssec:accelanalysis}

In this section, we give an analysis of a proximal point method inspired by \cite{Nemirovski04}, tailored to our purposes. The method reduces the problem of finding an $\eps$-approximate saddle point to a minimax convex-concave objective to iteratively solving a proximal subproblem to sufficiently high accuracy. Consider a saddle point problem of the form 
\[\min_{x \in \xset} \max_{p \in \pset} f(x, p), \]
where $f$ is convex in its restriction to the first argument and concave in its restricton to the second. Define the \emph{duality gap} of a pair $(x, p)$ to be
\[\max_{p' \in \pset} f(x, p') - \min_{x' \in \xset} f(x', p).\]
Note that when we define the associated gradient operator
\[g(x, p) \defeq \left(\nabla_x f(x, p), -\nabla_p f(x, p)\right),\]
convexity-concavity shows we may upper bound the duality gap with respect to some pair $(x', p')$ by the \emph{regret} $\inprod{g(x, p)}{(x, p) - (x', p')}$, in the sense of
\[f(x, p') - f(x', p) \le \inprod{\nabla_x f(x, p)}{x - x'} - \inprod{\nabla_p f(x, p)}{p - p'} = \inprod{g(x, p)}{(x, p) - (x', p')}. \]

The proximal point algorithm, with a possibly randomized prox oracle, defines a sequence $\{z_t\}$, where each iterate is the result of calling a proximal oracle on the previous iterate. Formally, the method is defined as follows.

\begin{definition}[Primal-dual proximal point method]
\label{def:conceptmp}
Initalize some $z_0 = (x_0, p_0)$, and let $q(x)$ and $r(p)$ be convex distance generating functions; let $V_z(w)$ be the Bregman divergence on the joint space with respect to their sum, i.e. for $z = (x, p)$ and $z' = (x', p')$,
\[V_z(z') \defeq V^q_x(x') + V^r_p(p').\]
We define the \emph{primal-dual proximal point method} to be the iteration of the following procedure: on iteration $t$, from the point $z_t = (x_t, p_t)$, let $z_{t + 1} = (x_{t + 1}, p_{t + 1})$ be any point such that 
\[\max_{u \in \xset \times \pset} \left\{\inprod{g(z_{t + 1})}{z_{t + 1} - u} - \alpha V_{z_t}(u) + \alpha V_{z_{t + 1}}(u)\right\} \le \eps.\]
\end{definition}
We remark that this definition of $z_{t + 1}$ is motivated by the fact that the (exact) solution of 
\begin{equation}\label{eq:gensubproblem}\min_{x \in \xset} \max_{p \in \pset} f(x, p) + \alpha V^q_{x_t}(x) - \alpha V^r_{p_t}(p)\end{equation}
has this property with $\eps = 0$; the proximal point method implies that any efficient algorithm for finding a high-precision saddle point to the prox problem suffices. Our algorithm for computing iterates will ultimately be randomized; we will union bound the probability that the iterate produced does not have the necessary property over all iterations by an inverse polynomial in $n$.

\begin{lemma}
\label{lem:conceptmp}
The iterates resulting from running the primal-dual proximal point method for $T$ iterations satisfy, for any $u \in \xset \times \pset$,
\[\frac{1}{T} \sum_{t \in [T]} \inprod{g(z_t)}{z_t - u} \le \frac{\alpha V_{z_0}(u)}{T} + \eps.\]

\end{lemma}
\begin{proof}
Consider some particular iterate $t$. By the definition of $z_{t + 1}$, we have for all $u$,
\[\inprod{g(z_{t + 1})}{z_{t + 1} - u} \le \alpha\left(V_{z_t}(u) - V_{z_{t + 1}}(u)\right) + \eps.\]
Summing over all iterations, taking an average, and using nonnegativity of $V$ yields the conclusion.
\end{proof}

We now specialize the required oracle for computing the $\{z_t\}$ to our particular saddle-point problem (in the case of $\ell_\infty$ regression). In our setting (after the constraints $A$ have been appropriately doubled to account for sign), we wish to solve
\[\min_{x \in [-1, 1]^m} \norm{Ax - b}_\infty = \min_{x \in [-1, 1]^m} \max_{p \in \Delta^n} p^\top(Ax - b). \]
The associated gradient operator for a point $(x, p)$ is
\begin{equation}\label{eq:gregress}g(x, p) = \left(A^\top p, b - Ax\right).\end{equation}
For the rest of this section, whenever we write $g(x, p)$ and the associated $A, b$ in the regression problem are clear from context, we mean \eqref{eq:gregress}. We note that to solve the original (primal-only) regression problem, it suffices to obtain duality gap in the primal-dual regression problem with respect to $(x^*, p')$ for any $p'$, where $x^* = \argmin_{x \in [-1, 1]^m} \norm{Ax - b}_\infty$, as quantified in the following.
\begin{lemma}
\label{lem:boundagainstu}
Let $z = (x, p) \in [-1, 1]^m \times \Delta^n$ be a pair such that for all $u = (x^*, p')$, where $x^* \in  [-1, 1]^m$ is fixed and $p' \in \Delta^n$ is arbitrary,
\[\inprod{g(z)}{z - u} \le \eps.\]
Then, we have
\[\norm{Ax - b}_\infty - \norm{Ax^* - b}_\infty \le \eps.\]
\end{lemma}
\begin{proof}
Choose $p'$ so that $p'^\top(Ax - b) = \norm{Ax - b}_\infty$. Then,
\[\inprod{g(z)}{z - u} = p^\top \left((Ax - b) - (Ax^* - b)\right) + (b - Ax)^\top (p - p') \ge \norm{Ax - b}_\infty - \norm{Ax^* - b}_\infty. \]
The only inequality follows from $p^\top(Ax^* - b) \le \norm{Ax^* - b}_\infty$ for any $p \in \Delta^n$.
\end{proof}
In our definition of the proximal point method (Definition \ref{def:conceptmp}), we choose $q(x) = \frac{1}{2s}\norm{x}_2^2$, and $r(p) = \sum_{i \in [n]} p_i \log p_i$, where $s \defeq \norm{x^*}_2^2$. It is simple to compute that from these definitions, 
\[V^q_{x}(x') = \frac{1}{2s}\norm{x - x'}_2^2,\; V^r_{p}(p') = \sum_{i \in [n]} p'_i \log \frac{p'_i}{p_i}.\]
Moreover, it is well-known that when $p \in \Delta^n$ is the uniform distribution $\frac{1}{n}\mathbbm{1}$, the range of $V^r_{p}(p')$ is bounded by $\log n$. Therefore, Lemma~\ref{lem:conceptmp} and Lemma~\ref{lem:boundagainstu} imply that we only need to take $\tilde{O}(\alpha/\eps)$ iterations of the proximal point method to obtain an $\eps$-approximate minimizer to the regression problem. We complete the analysis of this framework by showing that in order to return a sequence $\{z_t\}$ with the necessary properties, it suffices to approximately compute the saddle point to problems of the form \eqref{eq:gensubproblem}.

\begin{lemma}
\label{lem:sufficientx}
From a point $z = (x, p)$, let $\bar{z} = (\bar{x}, \bar{p})$ be the solution to the problem
\[\argmin_{x' \in  [-1, 1]^m} \argmax_{p' \in \Delta^n} p'^\top(Ax' - b) + \frac{\alpha}{2s}\norm{x - x'}_2^2 - \alpha \sum_{i \in [n]}p'_i \log \frac{p'_i}{p_i}.\]
Then, for $\eps < 1$, any $x'$ with 
\[\norm{x' - \bar{x}}_\infty \le \min\left(\frac{\eps }{16\norm{A}_\infty}, \frac{\eps s}{8\alpha m}, \frac{\eps\alpha}{64\norm{A}_\infty^2 }\right),\]
and setting $p' \in \Delta^n$ to be
\[p' \propto \exp\left(\frac{1}{\alpha}\left(Ax' - b + \alpha \log p\right)\right),\]
letting $z' = (x', p')$, for all $u \in  [-1, 1]^m \times \Delta^n$,
\[\inprod{g(z')}{z' - u} - \alpha V_{z}(u) + \alpha V_{z'}(u) \le \eps.\]
\end{lemma}
\begin{proof}
By the optimality conditions of the definition of $\bar{z}$, we see that for all $u \in  [-1, 1]^m \times \Delta^n$,
\[\inprod{g(\bar{z})}{\bar{z} - u} \le \alpha (V_z(u) - V_{\bar{z}}(u) - V_z(\bar{z})) \le \alpha (V_z(u) - V_{\bar{z}}(u)).\]
Therefore, it suffices to show that
\begin{equation}\label{eq:bound4parts}\inprod{g(z') - g(\bar{z})}{\bar{z} - u} + \inprod{g(z')}{z' - \bar{z}} + \alpha (V_{\bar{z}}(u) - V_{z'}(u))\le \eps.\end{equation}
We first derive a simple bound on $\norm{\bar{p} - p'}_1$. Note that by the definition of $\bar{p}$ as the optimal response to $\bar{x}$, we have that $A\bar{x} - b + \alpha\log(p/\bar{p})$ is a multiple of the all-ones vector, so
\[\bar{p} \propto \exp\left(\frac{1}{\alpha}\left(A\bar{x} - b + \alpha \log p\right)\right).\]
Therefore, the multiplicative ratio between each entry of $p'$ and $\bar{p}$ is bounded by
\begin{equation}\label{eq:pvsp}\exp\left(\frac{2}{\alpha}\norm{A}_\infty\norm{\bar{x} - x'}_\infty\right) \le \exp\left(\frac{\eps}{32\norm{A}_\infty }\right) \le 1 + \frac{\eps}{16\norm{A}_\infty }.\end{equation}
This immediately implies that $\norm{\bar{p} - p'}_1 \le  \eps\norm{p'}_1/(16\norm{A}_\infty) = \eps/(16\norm{A}_\infty)$. Finally, we conclude by noting that by $\ell_1$-$\ell_\infty$ H\"older, and $\norm{x' - x}_\infty \le \eps/(16\norm{A}_\infty)$,
\begin{align*}
\inprod{g(\bar{z}) - g(z')}{\bar{z} - u} &\le 2\norm{A}_\infty\norm{\bar{x} - x'}_\infty + 2\norm{A}_\infty\norm{\bar{p} - p'}_1 \le \frac{\eps}{4},\\
\inprod{g(z')}{z' - \bar{z}} &\le \norm{A}_\infty\norm{\bar{x} - x'}_1 + \norm{A}_\infty\norm{\bar{p} - p'}_1 \le \frac{\eps}{4}.
\end{align*}
Moreover, by using the definitions of Bregman divergences and $\norm{x}_1 \le m$ for $x \in [-1, 1]^m$, and noting that similarly to the derivation of \eqref{eq:pvsp}, $p'/p$ is entrywise bounded by $\exp(\eps/4\alpha)$ via $\norm{x' - x}_\infty \le \eps/(16\norm{A}_\infty)$,
\begin{align*}
\alpha (V^q_{\bar{x}}(u_x) - V^q_{x'}(u_x)) &= \frac{\alpha}{2s}\norm{\bar{x} - u_x}_2^2 - \frac{\alpha}{2s}\norm{x' - u_x}_2^2 = \frac{\alpha}{s} \inprod{u_x}{x' - \bar{x}} + \frac{\alpha}{2s}\inprod{x' + \bar{x}}{x' - \bar{x}} \\
&\le \frac{\alpha}{s}\left(\norm{u_x}_1 + \half\norm{x' + \bar{x}}_1\right)\norm{x' - \bar{x}}_\infty \le \frac{2\alpha m}{s}\norm{x' - \bar{x}}_\infty \le \frac{\eps}{4}, \\
\alpha (V^r_{\bar{p}}(u_p) - V^r_{p'}(u_p)) &= \alpha\sum_{i \in [n]}[u_p]_i\log\frac{p'_i}{\bar{p}_i} \le \alpha \max_{i \in [n]} \log\frac{p'_i}{\bar{p}_i} \le \frac{\eps}{4}.
\end{align*}
Finally, \eqref{eq:bound4parts} follows by combining the above bounds.
\end{proof}

Finally, note that for $\bar{z} = (\bar{x}, \bar{p})$ the solution to the problem
\[\argmin_{x' \in  [-1, 1]^m} \argmax_{p' \in \Delta^n} p'^\top(Ax' - b) + \frac{\alpha}{2s}\norm{x - x'}_2^2 - \alpha \sum_{i \in [n]}p'_i \log \frac{p'_i}{p_i},\]
we can equivalently write that $\bar{x}$ is the solution to the problem
\begin{equation*}\argmin_{x \in [-1, 1]^m}\; \alpha \log \sum_{i \in [n]}\exp\left(\frac{1}{\alpha}\left[Ax - \tilde{b}\right]_i \right) + \frac{\alpha}{2s}\norm{x' - x}_2^2,\; \tilde{b} \defeq b - \alpha \log p.\end{equation*}
We will show in the following section how to efficiently compute an approximate minimizer to this problem with high probability.

\subsection{Constructing the subproblem oracle}
\label{ssec:accprox}

In this section, we develop a new analysis of (unaccelerated) coordinate descent under local coordinate smoothness estimates and a box constraint, and show how to use it to compute a high-accuracy solution to the subproblems required by our proximal point method. More specifically, we develop an efficient iterative method for solving the problem (abusing some notation for simplicity of this self-contained section)
\begin{equation}\label{eq:subproblemgen}\argmin_{x \in [-1, 1]^m}\; \alpha \log \sum_{i \in [n]}\exp\left(\frac{1}{\alpha}\left[Ax - b\right]_i \right) + \frac{\alpha}{2s}\norm{x - \bar{x}}_2^2.\end{equation}

\subsubsection{Box-constrained coordinate descent under dynamic sampling}

In this section, we first develop a general coordinate descent analysis under a box constraint, amenable to dynamic sampling probabilities. Let $\xset$ be an arbitrary box, e.g. product of one-dimensional intervals, and let $f$ be an $\ell_2$ $\mu$-strongly convex function. Suppose at each point $x$, we have local coordinate smoothness estimates $\{L_j(x)\}$ such that for
\[x' = \argmin_{x' \in \xset}\left\{f(x) + \inprod{\nabla_j f(x)}{x' - x} + \frac{L_j(x)}{2}\norm{x' - x}_2^2 \right\}, \]
we have that the upper bound (recalling Definition~\ref{def:lcs}) holds, e.g.
\[f(x') \le f(x) + \inprod{\nabla_j f(x)}{x' - x} + \frac{L_j(x)}{2}\norm{x' - x}_2^2. \]
Further, define
\[S(x) = \sum_{j \in [m]} L_j(x),\]
and assume that there is a global upper bound $S$ on $S(x)$. Consider the following ``local smoothness'' variant of the standard coordinate descent algorithm.
\begin{definition}[Locally smooth coordinate descent]
\label{def:lcd}
Given a function $f$ with local coordinate smoothnesses $\{L_j(x)\}$ at each point $x$, define the local smoothness coordinate descent algorithm as iteratively performing the following (resetting $x' \gets x$) every iteration:
\begin{enumerate}
	\item Sample $j \propto L_j(x)$.
	\item Update $x' \gets \argmin_{x' \in \xset}\left\{f(x) + \inprod{\nabla_j f(x)}{x' - x} + \frac{L_j(x)}{2}\norm{x' - x}_2^2 \right\}$.
\end{enumerate}
\end{definition}
We will now prove a bound on its multiplicative progress in a single iteration.
\begin{lemma}
\label{lem:dynamicsampleprog}
	\[\E[f(x')] - f(x^*) \le \left(1 - \frac{\mu}{2S}\right)(f(x) - f(x^*)). \]
\end{lemma}
\begin{proof}
	First, define 
	\begin{align*}\pdown &\defeq f(x) - \min_{x^\downarrow \in \xset}\left\{f(x) + \inprod{\nabla f(x)}{x^\downarrow - x} + \frac{\mu}{2}\norm{x^\downarrow - x}_2^2 \right\},\\
	x^\downarrow &\defeq \argmin_{x^\downarrow \in \xset}\left\{f(x) + \inprod{\nabla f(x)}{x^\downarrow - x} + \frac{\mu}{2}\norm{x^\downarrow - x}_2^2 \right\}.
	\end{align*}
	We have by strong convexity that $f(x) - f(x^*) \le \pdown$. We also define $\gdown \defeq x - x^\downarrow$, and note that $\gdown$ agrees with $\nabla f(x)$ in the sign of each coordinate. Further, by separability of the box constraint,
	\begin{equation}
	\label{eq:downbound}0 \le |\gdown_j| \le \frac{1}{\mu} |\nabla_j f(x)|,\; \forall j \in [m]. 
	\end{equation}
	We can explicitly write that
	\[\pdown = \sum_{j \in [m]} \pdown_j, \text{ where } \pdown_j \defeq \gdown_j \left(\nabla_j f(x) - \frac{\mu}{2}\gdown_j \right).\]
	Similarly, we define for each $j \in [m]$,
	\[\pup_j \defeq f(x) - \min_{x' \in \xset}\left\{f(x) + \inprod{\nabla_j f(x)}{x' - x} + \frac{L_j(x)}{2}\norm{x' - x}_2^2 \right\}. \]
	We let $\gup$ be the vector such that $\gup_j$ agrees with $[x - x']_j$ if coordinate $j$ was sampled. In particular, $\gup$ agrees with $\nabla f(x)$ in the sign of each coordinate and by separability $\forall j \in [m]$,
	\begin{equation}\label{eq:upbound}0 \le |\gup_j| \le \frac{1}{L_j(x)} |\nabla_j f(x)|. \end{equation}
	We can explicitly write
	\[\pup_j = \gup_j\left(\nabla_j f(x) - \frac{L_j(x)}{2}\gup_j \right). \]
	First, we claim that for each $j \in [m]$,
	\begin{equation}\label{eq:pupdownbound}\pup_j \ge \frac{\mu}{2L_j(x)}\pdown_j. \end{equation}
	Note that if coordinate $j$ was sampled, and $x'_j$ is on the boundary of $\xset$, then $\gup_j = \gdown_j$, since the minimization problem defining $\gdown$ involves a larger step size. Conversely, if $[x^\downarrow]_j$ is not on the boundary of $\xset$, then neither is $x'_j$, and the upper bounds of \eqref{eq:downbound}, \eqref{eq:upbound} are tight. In both these cases and the third where $[x^\downarrow]_j$ is on the boundary and $x'_j$ is not, the following inequality holds:
	\[|\gup_j| \ge \frac{\mu}{L_j(x)} |\gdown_j|.\]
	We further note that
	\[\left|\nabla_j f(x) - \frac{L_j(x)}{2}\gup_j \right| \ge \half |\nabla_j f(x)| \ge \half \left|\nabla_j f(x) - \frac{\mu}{2}\gdown_j \right|.\]
	Combining these two facts with the definitions of $\pup_j$, $\pdown_j$ shows $\eqref{eq:pupdownbound}$. Now, we have
	\begin{align*}\E[f(x')] &\le f(x) - \E[\pup_j] \\
	&= f(x) - \sum_{j \in [m]} \frac{L_j(x)}{2S(x)} \pup_j \\
	&\le f(x) - \sum_{j \in [m]} \frac{\mu}{2S(x)} \pdown_j \\
	&= f(x) - \frac{\mu}{2S}\pdown.
	\end{align*}
	Subtracting $f(x^*)$ from both sides and using the lower bound on $\pdown$ gives the result.
\end{proof}

By iteratively applying Lemma~\ref{lem:dynamicsampleprog}, and Markov's inequality, we have the following corollary.

\begin{corollary}
\label{corr:boxconstrainedcd}
Box-constrained locally smooth coordinate descent initialized at $x_0$, applied to an $\ell_2$ $\mu$-strongly convex function $f$ converges to an $\eps$-approximate minimizer with probability at least $1 - \delta$ in 
\[O\left(\frac{S}{\mu}\log\left(\frac{f(x_0) - f(x^*)}{\eps\delta}\right)\right) \text{ iterations.}\]
\end{corollary}

We also remark that this analysis generalizes easily to strong convexity in any diagonal norm given by a (nonnegative) diagonal matrix $D$. In particular, let $f$ be $\mu$-strongly-convex in the $D$ norm, where $D = \diag(d)$ is some (positive) diagonal matrix, and assume we have the local coordinate smoothness bounds $\{L_j(x)\}_{j \in [m]}$ (note that the smoothness bound is still in the $\ell_2$ norm, i.e. independent of the strong convexity measurement matrix). We briefly discuss how to modify the guarantee of Lemma~\ref{lem:dynamicsampleprog}. The algorithm is given as follows.
\begin{definition}[Local smoothness coordinate descent in a diagonal norm]
	\label{def:lcdd}
	Given a function $f$ with local coordinate smoothnesses $\{L_j(x)\}$ at each point $x$, define the local smoothness coordinate descent algorithm in the $D$ norm as iteratively performing the following (resetting $x' \gets x$) every iteration:
\begin{enumerate}
	\item Sample $j \propto \kappa_j(x) \defeq \frac{L_j(x)}{d_j}$. 
	\item Update $x' \gets \argmin_{x' \in \xset}\left\{f(x) + \inprod{\nabla_j f(x)}{x' - x} + \frac{L_j(x)}{2}\norm{x' - x}_2^2 \right\}$.
\end{enumerate}
\end{definition}
We also define $S(x) = \sum_{j \in [m]}\kappa_j(x)$, and let $S$ be a global upper bound. We modify the definitions
\begin{align*}\pdown &\defeq f(x) - \min_{x_* \in \xset}\left\{f(x) + \inprod{\nabla f(x)}{x_* - x} + \frac{\mu}{2}\norm{x_* - x}_D^2 \right\},\\
x_* &\defeq \argmin_{x_* \in \xset}\left\{f(x) + \inprod{\nabla f(x)}{x_* - x} + \frac{\mu}{2}\norm{x_* - x}_D^2 \right\},\\
\gdown &\defeq x - x_*.
\end{align*}
We also clearly have by the same argument that
\[0 \le |\gdown_j| \le \frac{1}{\mu d_j}|\nabla_j f(x)|. \]
Therefore, the same arguments allow us to conclude that for each $j \in [m]$,
\[\pup_j \ge \frac{\mu d_j}{2L_j(x)}\pdown_j,\text{ where we recall }\pup_j = \gup_j\left(\nabla_j f(x) - \frac{L_j(x)}{2}\gup_j \right).\]
Finally, our given sampling probabilities imply that we have the desired
\[\E[f(x')] - f(x^*) \le \left(1 - \frac{\mu}{2S} \right)(f(x) - f(x^*)). \]
This yields the following corollary.
\begin{corollary}
	\label{corr:boxcddiagnorm}
	Box-constrained local smoothness coordinate descent in the diagonal $D$ norm initialized at $x_0$, applied to a $\mu$-strongly convex function $f$ in the $D$ norm, converges to an $\eps$-approximate minimizer with probability at least $1 - \delta$ in 
	\[O\left(\frac{S}{\mu}\log\left(\frac{f(x_0) - f(x^*)}{\eps\delta}\right)\right) \text{ iterations.}\]
\end{corollary}

\subsubsection{Minimizing the regularized softmax objective}

We now use the developments of the prior section to obtain the runtime of an efficient oracle for solving \eqref{eq:subproblemgen} to high precision; we restate the objective here:
\[h(x) \defeq \alpha \log \sum_{i \in [n]}\exp\left(\frac{1}{\alpha}\left[Ax - b\right]_i \right) + \frac{\alpha}{2s}\norm{x - \bar{x}}_2^2.\]
The complexity of minimizing this objective function using the box-constrained coordinate descent under local coordinate smoothnesses follows from estimates given in the following lemma.
\begin{lemma}[Local coordinate smoothnesses of regularized softmax]
\label{lem:lcsregsmax}
At a point $x \in [-1, 1]^m$, and for all $j \in [m]$, define
\[L_j(x) = \frac{8}{\alpha}\norm{\aj}_\infty\left(\inprod{|\aj|}{p(x)} + \frac{2\alpha}{s}\right) + \frac{\alpha}{s}.\]
Then, $h$ is $L_j(x)$ locally-coordinate smooth at $x$ for all $j \in [m]$.
\end{lemma}
\begin{proof}
Recalling Definition~\ref{def:lcs}, we prove the following: for $y = x + \gamma e_j$, $\gamma \in \left[\pm\frac{1}{L_j(x)}|\nabla_j h(x)|\right]$,
\begin{equation}\label{eq:hessbounds}\nabla^2_{jj}h(y) \leq L_j(x).\end{equation}
Defining $p(x) \propto \exp((Ax - b)/\alpha)$, by Fact~\ref{fact:smaxfacts}, $\nabla^2_{jj} h(y) \leq \frac{1}{\alpha}\norm{\aj}_{p(y)}^2 + \frac{\alpha}{s}$. Therefore, it clearly suffices \ to show that $p(y) \leq 8p(x)$ entrywise. Note that as long as we show that entrywise
\[\exp\left(\frac{Ay - b}{\alpha}\right) \in \left[\frac{1}{e}, e\right]\exp\left(\frac{Ax - b}{\alpha}\right),\]
we have the conclusion by $e^2 < 8$. Now, using the bound on $\gamma$, this is equivalent to showing for all $i$ that $|\inprod{\ai}{x - y}| \leq \alpha$. Recalling $\nabla_j h(x) = \inprod{\aj}{p(x)} + \frac{\alpha}{s}(x_j - \bar{x}_j)$, the following suffices:
\begin{align*}
|A_{ij}|\left|\frac{\nabla_j h(x)}{L_j(x)}\right| = \left|\frac{|A_{ij}|\left(\inprod{\aj}{p(x)} + \frac{\alpha}{s}(x_j - \bar{x}_j)\right)}{\frac{8}{\alpha}\norm{\aj}_\infty \left(\inprod{|\aj|}{p(x)} + \frac{2\alpha}{s}\right) + \frac{\alpha}{s}}\right| \leq \alpha.
\end{align*}
The conclusion follows.
\end{proof}

\subsubsection{Minimizing the diagonally regularized softmax objective}

By a simple modification of the regularizer $q(x)$ used in the proximal point method, we show how to obtain improved smoothness parameters in the regime $n < m$, independent of the sparsity of the optimal solution. In particular, for $D = \diag(\{\norm{\aj}_\infty\})$, the diagonal matrix whose entries are the $\{\norm{\aj}_\infty\}$, consider running the mirror prox procedure with the regularizer $q(x) = \frac{1}{2n\norm{A}_\infty}\norm{x}_D^2$; the range of $q(x)$ over the box $[-1, 1]^m$ is clearly at most a constant, since the sum of (absolute values of) entries of $A$ is bounded by $\tilde{O}(n\norm{A}_\infty)$. Therefore, it suffices to design an efficient iterative method for, in the vein of \eqref{eq:subproblemgen}, solving subproblems
\begin{equation}\label{eq:subproblemdiag}h_d(x) \defeq \alpha \log \sum_{i \in [n]}\exp\left(\frac{1}{\alpha}\left[Ax - b\right]_i \right) + \frac{\alpha}{2n\norm{A}_\infty}\norm{x - \bar{x}}_D^2.\end{equation}
In lieu of Lemma~\ref{lem:lcsregsmax}, we have the following local smoothness bounds on this subproblem.

\begin{lemma}[Local coordinate smoothnesses of diagonally regularized softmax]
	\label{lem:lcsregsmaxdiag}
	Let $d$ be the vector whose entries are $\norm{\aj}_\infty$ such that $D = \diag(d)$. At a point $x \in [-1, 1]^m$, and for all $j \in [m]$, define
	\[L_j(x) = \frac{8}{\alpha}\norm{\aj}_\infty\left(\inprod{|\aj|}{p(x)} + \frac{2\alpha}{n\norm{A}_\infty}\norm{\aj}_\infty\right) + \frac{\alpha}{n\norm{A}_\infty} \norm{\aj}_\infty.\]
	Then, $h_d$ is $L_j(x)$ locally-coordinate smooth at $x$ for all $j \in [m]$.
\end{lemma}
\begin{proof}
	The proof is similar to that of Lemma~\ref{lem:lcsregsmax}. Recalling Definition~\ref{def:lcs}, for $y = x + \gamma e_j$, $\gamma \in \left[\pm\frac{1}{L_j(x)}|\nabla_j h_d(x)|\right]$, we wish to show
	\begin{equation}\label{eq:hessbounds}\nabla^2_{jj}h_d(y) \leq L_j(x).\end{equation}
	Defining $p(x) \propto \exp((Ax - b)/\alpha)$, by Fact~\ref{fact:smaxfacts}, $\nabla^2_{jj} h_d(y) \leq \frac{1}{\alpha}\norm{\aj}_{p(y)}^2 + \frac{\alpha}{n\norm{A}_\infty}\norm{\aj}_\infty$. Therefore, it clearly suffices to show that $p(y) \leq 8p(x)$ entrywise. It suffices to show that for all $i$ that $|\inprod{\ai}{x - y}| \leq \alpha$, or recalling the definition of $\nabla_j h(x) = \inprod{\aj}{p(x)} + \frac{\alpha}{n\norm{A}_\infty}\norm{\aj}_\infty(x_j - \bar{x}_j)$,
	\begin{align*}
	|A_{ij}|\left|\frac{\nabla_j h(x)}{L_j(x)}\right| = \left|\frac{|A_{ij}|\left(\inprod{\aj}{p(x)} + \frac{\alpha}{n\norm{A}_\infty}\norm{\aj}_\infty(x_j - \bar{x}_j)\right)}{\frac{8}{\alpha}\norm{\aj}_\infty \left(\inprod{|\aj|}{p(x)} + \frac{2\alpha }{n\norm{A}_\infty}\norm{\aj}_\infty\right) + \frac{\alpha}{n\norm{A}_\infty}\norm{\aj}_\infty}\right| \leq \alpha.
	\end{align*}
	The conclusion follows.
\end{proof}

\subsection{Putting it all together: accelerated $\ell_\infty$ regression}
\label{ssec:accelregressformal}

We now state our main runtime result for $\ell_\infty$ regression. We combine previous developments to bound the number of coordinate descent iterations needed under local coordinate smoothness estimates needed to find an $\eps$-approximate minimizer to the box-constrained $\ell_\infty$ regression problem (Definition~\ref{def:boxedlinfreg}). We remark that the theorem statement assumes access to query and sampling oracles for the local coordinate smoothnesses; we show how to design efficient oracles for column-sparse $A$ in Section~\ref{ssec:cheapiter}. The combination of the following two theorems formally show Theorem~\ref{thm:accelboxlinfreg}.

\begin{theorem}[Coordinate acceleration for $\ell_\infty$ regression]
\label{thm:mainregressionthm}
The proximal point method (Definition~\ref{def:conceptmp}) with regularizers $q(x) = \frac{1}{2s}\norm{x}_2^2$ and $r(p) = \sum_{i \in [n]}p_i \log p_i$, with each iterate defined by the local smoothness coordinate descent method (Definition~\ref{def:lcd}) applied to the appropriate subproblem, results (with high probability) in an $\eps$-approximate minimizer to the box-constrained $\ell_\infty$ regression problem
in time
\[\tilde{O}\left(\left(\left(\frac{s\norm{A}_\infty^2}{\alpha^2} + \frac{\min(m, n)\norm{A}_\infty}{\alpha} + m\right) \cdot \mathcal{T}_{\text{iter}} + \nnz(A)\right)\cdot \frac{\alpha}{\eps}\right),\]
where $\mathcal{T}_{\text{iter}}$ is the cost of sampling proportional to $L_j(x)$ and computing the value of $L_j(x)$ for an iterate $x$ of local smoothness coordinate descent. For $\alpha = \max(\eps, \sqrt{s/m}\norm{A}_\infty)$ and $\mathcal{T}_{\text{iter}} = O(c\log n)$ from Section~\ref{ssec:cheapiter}, where $c$ is the column sparsity of $A$, the runtime is
\[\tilde{O}\left(mc + \frac{\left(\min(m, n) + \sqrt{ms}\right)c\norm{A}_\infty}{\eps}\right).\]
\end{theorem}
\begin{proof}
We first discuss the complexity of returning an iterate of the proximal point method. Lemma~\ref{lem:sufficientx}, and the discussion following, imply that for a function of form \eqref{eq:subproblemgen} with optimal argument $x^*_t$, it suffices to find any point $x'$ 
with $\norm{x' - x^*_t}_\infty$ bounded by an inverse polynomial in parameters $\norm{A}_\infty, m, s, \eps^{-1}, \alpha^{-1}$ to implement the proximal point method. By Fact~\ref{fact:adderr}, the range of the function (where the linear term $b_t$ is appropriately shifted)
\[h(x) = \alpha \log \sum_{i \in [n]}\exp\left(\frac{1}{\alpha}\left[Ax - b_t\right]_i \right) + \frac{\alpha}{2s}\norm{x - \bar{x}}_2^2\]
is at most $\alpha \log n + 2\norm{A}_\infty + \frac{\alpha m}{2s}$, where the second term comes from the range of $\norm{Ax - b}_\infty$ over $[-1, 1]^m$, and the third from a simple bound $\norm{x - \bar{x}}_2^2 \le m$ for all $x \in [-1, 1]^m$. Moreover, strong-convexity of $h(x)$ in the $\ell_2$ norm, and optimality of $x_t^*$, yields
\[h(x') - h(x_t^*) \ge \inprod{\nabla h(x_t^*)}{x' - x_t^*} + \frac{\alpha}{2s}\norm{x' - x_t^*}_2^2 \ge \frac{\alpha}{2s}\norm{x' - x_t^*}_2^2, \]
which implies
\begin{equation}\label{eq:linf_from_error}\norm{x' - x^*_t}_\infty \le \norm{x' - x^*_t}_2 \le \sqrt{\frac{2s\left(h(x') - h(x^*_t)\right)}{\alpha}}.\end{equation}
Note that for any point $x$, we can define an upper bound on the sum of values $L_j(x)$ in Lemma~\ref{lem:lcsregsmax},
\begin{equation}\label{eq:sbound}S \defeq \frac{8}{\alpha}\norm{A}_\infty^2 + \frac{16\min(m, n)}{s}\norm{A}_\infty + \frac{m\alpha}{s} \ge \sum_{j \in [m]} \frac{8}{\alpha}\norm{\aj}_\infty\left(\inprod{|\aj|}{p(x)} + \frac{2\alpha}{s}\right) + \frac{\alpha}{s}.\end{equation}
Here, we used that the sum of entries in the matrix is at most $n\norm{A}_\infty$, and the largest entry in any column is at most $\norm{A}_\infty$. Then, by Corollary~\ref{corr:boxconstrainedcd} with strong convexity parameter $\mu = \alpha/s$, we see that a sufficient $x'$ may be found with high probability in time
\[\tilde{O}\left(\left(\frac{s\norm{A}_\infty^2}{\alpha^2} + \frac{\min(m, n)\norm{A}_\infty}{\alpha} + m\right) \cdot \mathcal{T}_{\text{iter}}\right).\]
Finally, due to our choice of the regularizers $q$ and $r$ in the proximal point method and Lemma~\ref{lem:conceptmp}, $\tO(\alpha/\eps)$ iterations of proximal point suffice, for any $\alpha \ge \eps$. Combining these bounds yields the first runtime claim. To see the second, we simplified using $\nnz(A) \le mc$.
\end{proof}

Before we give our result for accelerating $\ell_\infty$ regression via a diagonal norm regularization, we state a technical result on the degree of accuracy required by solutions of the proximal point subproblems, the analog of Lemma~\ref{lem:sufficientx} in the diagonal norm (throughout, $D \defeq \textbf{diag}(\{\norm{\aj}_\infty\})$).

\begin{lemma}
	\label{lem:sufficientxdiag}
	From a point $z = (x, p)$, let $\bar{z} = (\bar{x}, \bar{p})$ be the solution to the problem 
	\[\argmin_{x' \in  [-1, 1]^m} \argmax_{p' \in \Delta^n} p'^\top(Ax' - b) + \frac{\alpha}{2n\norm{A}_\infty}\norm{x - x'}_D^2 - \alpha \sum_{i \in [n]}p'_i \log \frac{p'_i}{p_i}.\]
	Then, for $\eps < 1$, any $x'$ with 
	\[\norm{x' - \bar{x}}_\infty \le \min\left(\frac{\eps}{16\norm{A}_\infty}, \frac{\eps n}{8\alpha m}, \frac{\eps\alpha}{64\norm{A}_\infty^2 }\right),\]
	and setting $p' \in \Delta^n$ to be
	\[p' \propto \exp\left(\frac{1}{\alpha}\left(Ax' - b + \alpha \log p\right)\right),\]
	letting $z' = (x', p')$, for all $u \in  [-1, 1]^m \times \Delta^n$, and where divergences are with respect to $q(x) \defeq \frac{1}{2n\norm{A}_\infty}\norm{x}_D^2$ and $r(p) = \sum_{i \in [n]}p_i \log p_i$,
	\[\inprod{g(z')}{z' - u} - \alpha V_{z}(u) + \alpha V_{z'}(u) \le \eps.\]
\end{lemma}
\begin{proof}
	By the optimality conditions of the definition of $\bar{z}$, we see that for all $u \in  [-1, 1]^m \times \Delta^n$,
	\[\inprod{g(\bar{z})}{\bar{z} - u} \le \alpha (V_z(u) - V_{\bar{z}}(u) - V_z(\bar{z})) \le \alpha (V_z(u) - V_{\bar{z}}(u)).\]
	Therefore, it suffices to show that
	\begin{equation}\label{eq:bound4partsdiag}\inprod{g(z') - g(\bar{z})}{\bar{z} - u} + \inprod{g(z')}{z' - \bar{z}} + \alpha (V_{\bar{z}}(u) - V_{z'}(u))\le \eps.\end{equation}
	By exactly the same logic as in Lemma~\ref{lem:sufficientx}, we have the multiplicative ratio between every entry of $p'$ and $\bar{p}$ is bounded by $1 + \eps/(16\norm{A}_\infty)$, and $\norm{\bar{p} - p'}_1 \le \eps/(16\norm{A}_\infty)$. Finally, we conclude by noting that by $\ell_1$-$\ell_\infty$ H\"older, and $\norm{x' - x}_\infty \le \eps/(16\norm{A}_\infty)$,
	\begin{align*}
	\inprod{g(\bar{z}) - g(z')}{\bar{z} - u} &\le 2\norm{A}_\infty\norm{\bar{x} - x'}_\infty + 2\norm{A}_\infty\norm{\bar{p} - p'}_1 \le \frac{\eps}{4},\\
	\inprod{g(z')}{z' - \bar{z}} &\le \norm{A}_\infty\norm{\bar{x} - x'}_1 + \norm{A}_\infty\norm{\bar{p} - p'}_1 \le \frac{\eps}{4}.
	\end{align*}
	Moreover, by using the definitions of Bregman divergences, $\norm{x}_1 \le m$ for $x \in [-1, 1]^m$, $p'/p$ is entrywise bounded by $\exp(\eps/4\alpha)$, and $\norm{A}_\infty$ is larger than every entry of $D$,
	\begin{align*}
	\alpha (V^q_{\bar{x}}(u_x) - V^q_{x'}(u_x)) &= \frac{\alpha}{2n\norm{A}_\infty}\norm{\bar{x} - u_x}_D^2 - \frac{\alpha}{2n\norm{A}_\infty}\norm{x' - u_x}_D^2  \\
	&=\frac{\alpha}{n\norm{A}_\infty} u_x D(x' - \bar{x}) + \frac{\alpha}{2n\norm{A}_\infty}(x' + \bar{x})D(x' - \bar{x})\\
	&\le \frac{\alpha}{n}\left(\norm{u_x}_1 + \half\norm{x' + \bar{x}}_1\right)\norm{x' - \bar{x}}_\infty \le \frac{2\alpha m}{n}\norm{x' - \bar{x}}_\infty \le \frac{\eps}{4}, \\
	\alpha (V^r_{\bar{p}}(u_p) - V^r_{p'}(u_p)) &= \alpha\sum_{i \in [n]}[u_p]_i\log\frac{p'_i}{\bar{p}_i} \le \alpha \max_{i \in [n]} \log\frac{p'_i}{\bar{p}_i} \le \frac{\eps}{4}.
	\end{align*}
	Finally, \eqref{eq:bound4partsdiag} follows by combining the above bounds.
\end{proof}

\begin{theorem}[Coordinate acceleration for $\ell_\infty$ regression in a diagonal norm]
\label{thm:mainregressionthmdiagonal}
The proximal point method (Definition~\ref{def:conceptmp}) with regularizers $q(x) \defeq \frac{1}{2n\norm{A}_\infty}\norm{x}_D^2$ for $D = \textbf{diag}(\{\norm{\aj}_\infty\})$ and $r(p) = \sum_{i \in [n]}p_i \log p_i$, with each iterate defined by the local smoothness coordinate descent method in the $D$ norm (Definition~\ref{def:lcdd}) applied to the appropriate subproblem, results (with high probability) in an $\eps$-approximate minimizer to the box-constrained $\ell_\infty$ regression problem
in time
\[\tilde{O}\left(\left(\left(\frac{n\norm{A}_\infty^2}{\alpha^2} + \frac{n\norm{A}_\infty}{\alpha} + m\right) \cdot \mathcal{T}_{\text{iter}} + \nnz(A)\right)\cdot \frac{\alpha}{\eps}\right),\]
where $\mathcal{T}_{\text{iter}}$ is the cost of sampling proportional to $L_j(x)/d_j$ and computing the value of $L_j(x)$ for an iterate $x$ of local smoothness coordinate descent. For $\alpha = \max(\eps, \sqrt{n/m}\norm{A}_\infty)$ and $\mathcal{T}_{\text{iter}} = O(c\log n)$ from Section~\ref{ssec:cheapiter}, where $c$ is the column sparsity of $A$, the runtime is
\[\tilde{O}\left(mc + \frac{\left(n + \sqrt{mn}\right)c\norm{A}_\infty}{\eps}\right).\]
\end{theorem}
\begin{proof}
We first note that without loss of generality, every entry of $D$ is at least $\eps/m$; indeed, adding $\eps/m$ to an arbitrary nonzero entry of each column only perturbs the value of $\norm{Ax - b}_\infty$ over $[-1, 1]^m$ by an additive $\eps$. Thus, in lieu of \eqref{eq:linf_from_error} in the proof of Theorem~\ref{thm:mainregressionthm}, it suffices to solve to a degree of accuracy polynomially larger in problem parameters, where we use that the objective is $1/(n\norm{A}_\infty)$-strongly convex in the $D$ norm, and the $D$ norm is at most $\eps/m$ times smaller than the $\ell_2$ norm (Lemma~\ref{lem:sufficientxdiag} bounds the accuracy we require in our subproblem solutions in $\ell_\infty$).

We next bound the sum of local smoothnesses (relative to $d_j$) induced by Lemma~\ref{lem:lcsregsmaxdiag}, and the resulting complexity of solving the subproblems induced by the proximal point method with a diagonal regularizer; the remainder of the proof follows identically from Theorem~\ref{thm:mainregressionthm}. Note that
\[\sum_{j \in [m]} \frac{L_j(x)}{d_j} = \sum_{j \in [m]} \frac{8}{\alpha}\left(\inprod{|\aj|}{p(x)} + \frac{2\alpha }{n\norm{A}_\infty}\norm{\aj}_\infty\right) + \frac{\alpha}{n\norm{A}_\infty} = O\left(\frac{\norm{A}_\infty}{\alpha} + 1 + \frac{m\alpha}{n\norm{A}_\infty}\right).\]
The strong convexity parameter of the induced subproblems in the diagonal norm, of the form \eqref{eq:subproblemdiag}, is $\alpha/(n\norm{A}_\infty)$. Thus, applying Corollary~\ref{corr:boxcddiagnorm} implies each iterate of the mirror prox method can be found with high probability in time
\[\tilde{O}\left(\left(\frac{n\norm{A}_\infty^2}{\alpha^2} + \frac{n\norm{A}_\infty}{\alpha} + m\right) \cdot \mathcal{T}_{\text{iter}}\right).\]
Finally, due to the choice of regularizers, the domain size is still $\tilde{O}(1)$, so $\tilde{O}(\alpha/\eps)$ iterations of proximal point suffice, giving the first claim; the second claim follows by choice of $\alpha$.
\end{proof}

\subsection{Cheap iterations for $\ell_\infty$ regression in column-sparse $A$}
\label{ssec:cheapiter}

In this section, we show how to attain cheap iterations for $A$ whose columns have bounded sparsity. In particular, suppose $A$ is $c$-column-sparse. We show how to, for the local coordinate smoothness estimates 
\begin{equation}\label{eq:ljx}L_j(x) = \frac{8}{\alpha}\norm{\aj}_\infty\left(\inprod{|\aj|}{p(x)} + \frac{2\alpha}{s}\right) + \frac{\alpha}{s}\end{equation}
defined in Lemma~\ref{lem:lcsregsmax}, implement maintenance of the $L_j(x)$ and sampling by the quantities $L_j(x)$ for each iteration of the local smoothness coordinate descent procedure applied to the problem \eqref{eq:subproblemgen}, in time $\mathcal{T}_{\text{iter}} = O(c\log n)$. This shows that the runtime of the efficient implementation of our algorithm is, up to a $\tilde{O}(c)$ multiplicative factor, the same as the iteration count; in particular, for $c = \tilde{O}(1)$, we are able to implement each step in $\tilde{O}(1)$ time, without affecting the number of iterations by more than a $\tilde{O}(1)$ factor. More formally, in this section we show the following.

\begin{lemma}[Efficient implementation of iterates]
\label{sec:fastimpl}
Suppose we implement local smoothness coordinate descent (Definition~\ref{def:lcd}) for the problem \eqref{eq:subproblemgen} for some $c$-column-sparse $A$. Then, with $\nnz(A)$ precomputation cost, throughout the lifetime of the algorithm for local coordinate smoothness estimates $L_j(x)$ \eqref{eq:ljx} where $x$ is an iterate, it is possible to (1) maintain the sum $\sum_{j \in [m]} L_j(x)$, (2) compute for any $j$ the value $L_j(x)$, and (3) sample from the distribution $\{p_j \propto L_j(x)\}$ in time $O(c\log n)$ per iteration.
\end{lemma}

\begin{proof} We will describe the $L_j(x)$ maintenance and sampling procedures separately.

\noindent{\bf Maintaining smoothness overestimates.}

We first show how to (implicitly) maintain the quantities
\[p_i(x) = \frac{\exp(\frac{1}{\alpha}[Ax - b]_i)}{\sum_{i' \in [n]} \exp(\frac{1}{\alpha}[Ax - b]_{i'})}\]
in $O(c)$ time per iteration. In particular, because each iteration of (local smoothness) coordinate descent, starting at $x$ and stepping to $x'$, only affects a single coordinate, and by column-sparsity this only affects at most $c$ of the values $\exp(\frac{1}{\alpha}[Ax - b]_i)$, we can maintain their sum in $O(c)$ time, and also maintain the vector $\exp(\frac{1}{\alpha}(Ax - b))$.

Next, we discuss how to maintain $\sum_{j \in [m]} L_j(x)$ and query any $L_j(x)$ in $O(c)$ time per iteration. In $O(\nnz(A))$ time we precompute and store all values
\[\frac{16\norm{\aj}_\infty}{s} + \frac{\alpha}{s},\]
and there are at most $c$ entries in $\aj$, so querying $L_j(x)$ can be performed in $O(c)$ time, because we can compute any entry of $p(x)$ using the stored $\exp(\frac{1}{\alpha}(Ax - b))$ and its maintained sum. Moreover, in computing the sum
\[\sum_{j \in [m]} L_j(x) = \left(\sum_{j \in [m]}\frac{16\norm{\aj}_\infty}{s} + \frac{m\alpha}{s}\right) + \left(\frac{8}{\alpha}\sum_{i \in [n]} p_i(x) \sum_{j \in [m]} |A_{ij}|\norm{\aj}_\infty \right),\]
all quantities other than the $p_i(x)$ can be precomputed; the second summand can be computed with respect to the unnormalized vector $\exp(\frac{1}{\alpha}(Ax - b))$, and then scaled uniformly using its sum.

\noindent{\bf Sampling from the distribution.} 

In this part of the proof, we describe how to implement sampling from the distribution proportional to $L_j(x)$. First, in the prior discussion note that we maintain the sum of the $L_j(x)$ by computing the values of the two summands
\[\sum_{j \in [m]}\frac{16\norm{\aj}_\infty}{s} + \frac{m\alpha}{s},\; \frac{8}{\alpha}\sum_{i \in [n]} p_i(x) \sum_{j \in [m]} |A_{ij}|\norm{\aj}_\infty.\]
We first flip an appropriately biased coin to choose a summand. If the first is selected, then we sample a coordinate $j \in [m]$ with probability proportional to
\[\frac{16\norm{\aj}_\infty}{s} + \frac{\alpha}{s};\]
this can be done in constant time via precomputation \cite{Walker77}. 

To sample from the second summand, it clearly suffices to first sample the rows of $A$ by a distribution proportional to $p(x)$, and then sample the indices of that row proportional to $|A_{ij}|\norm{\aj}_\infty$, the latter of which takes constant time via precomputation \cite{Walker77}. To sample the rows, we use the well-known strategy that it suffices to maintain an augmented binary search tree data structure whose leaves dynamically maintain the set of $\exp(\frac{1}{\alpha}[Ax - b]_i)$ for the current iterate $x$. As previously argued, each iteration changes only $c$ of these values, so maintaining the augmented binary search tree takes $O(c \log n)$ per iteration.

\end{proof}
\section{Accelerating Maximum Flow}
\label{sec:maxflow}

The primary goal of this section is to show how to use the development of Section~\ref{sec:regression}, tailored to the regression problem associated with maximum flow, and give tighter analyses on its runtime guarantees to demonstrate how it yields faster algorithms. The reduction to $\ell_\infty$ regression is the same as introduced in \cite{Sherman13}, and is included for completeness.

\subsection{Maximum flow preliminaries}

The maximum flow problem is defined as follows: given a graph, and two of its vertices $s$ and $t$ labeled as source and sink, find a flow $f \in \R^m$ which satisfies the capacity constraints such that the discrete divergence at the sink, $(Bf)_t$, is as large as possible, and $(Bf)_s = -(Bf)_t$, $(Bf)_v = 0$ for $v \neq s, t$.

Following the framework of \cite{Sherman13}, we consider instead the equivalent problem of finding a minimum congestion flow; intuitively, if we route 1 unit of flow from $s$ to $t$ and congest edges as little as possible, we can find the maximum flow by just taking the multiple of the minimum congestion flow which just saturates edges. The congestion incurred by a flow $f$ is $\normInf{U^{-1}f}$ where $U$ is the diagonal matrix of edge capacities, and we say $f$ routes demands $d$ if $Bf = d$. The problem of finding a minimum congestion flow for a given demand vector, and its dual, the maximum congested cut, can be formulated as follows:

\begin{equation}
\label{eq:maxflowmincut}
    \begin{aligned}
        & \underset{f}{\text{min.}}
        & & \|U^{-1}f\|_\infty
        & \text{s.t.}
        & & Bf = d,
          f \geq 0.\\
        & \underset{v}{\text{max.}}
        & & d^\top v
        & \text{s.t.}
        & & \|UB^\top v\|_1 \leq 1.\\
    \end{aligned}
\end{equation}

Let $d_S \defeq \sum_{u\in S} d_u$ and $c(S, T)$ denote the total weight of edges from $S$ to $T$. It is well-known that for the second problem, one of the threshold cuts with respect to $v$ achieves $d_S/c(S, V - S)\geq d^\top v$. Whenever the flow problem is clear from context, we will refer to any optimal flow by $f^{\textsc{OPT}}$.

\subsection{From maximum flow to constrained $\ell_\infty$ regression}
\label{sec:flowtolinfregress}

First, we show how to transform the maximum flow problem into a constrained regression problem. The key tool used here is the concept of a good \emph{congestion approximator} \cite{Sherman13}, and associated properties.

\begin{definition}[Congestion approximator]
    An $\alpha$-congestion approximator for $G$ is a matrix $R$ such that for any demand vector $d$,
    $
        \normInf{Rd} \leq \textup{OPT}_d\leq \alpha \normInf{Rd}.
    $
\end{definition}

For undirected graphs, it is known that $\tilde{O}(1)$-congestion approximators can be computed in nearly linear time \cite{Madry10,Sherman13,KelnerLOS14,Peng16}. Further, the certain variants of these congestion approximator have additional nice properties. We use the following construction from \cite{Peng16}.

\begin{theorem}[Summary of results in \cite{Peng16}]
\label{thm:caproperties} There is an algorithm which given an $m$-edge $n$-vertex undirected graph runs in time $\tilde{O}(m)$ and with high probability produces an $\alpha$-congestion approximator $R$, for $\alpha = \tilde{O}(1)$. Furthermore, the matrix $A \defeq 2\alpha RBU$ has the following properties: (1) each column of $A$ has at most $\tilde{O}(1)$ nonzero entries, (2) $\|A\|_\infty = \tilde{O}(1)$, (3) $A$ has $O(n)$ rows, and (4) $A$ can be computed in time $\tilde{O}(m)$.
\end{theorem}

The above theorem is the result of a construction in $\cite{Peng16}$. Properties 2, 3, and 4 are direct results of the construction given in the paper (where 3 follows from the fact that the congestion approximator comes from routing on a graph which is a tree). Property 1 results from the way in which the tree is constructed, such that the depth of the congestion-approximating tree is $\tilde{O}(1)$, so each edge in the original graph $G$ is only routed onto a polylogarithmic number of edges.

Our analysis of reducing the flow problem to the regression problem follows that of \cite{Sherman13}. In particular, the reduction is given as follows.

\begin{lemma}
\label{lemma:flowreduction}
Let $G$ be an undirected graph and $d$ be a demand vector. Assume we are given an $\alpha$-congestion approximator $R$, and the associated matrix $A = 2\alpha RBU$. Furthermore, let $2\alpha Rd \defeq b$. In order to multiplicatively approximately solve the maximum flow problem given by \Cref{eq:maxflowmincut}, it suffices to solve an associated box-constrained regression problem $\|Ax - b\|_\infty$ over $x \in [-1, 1]^m$ a nearly-constant number of times to an $\eps$-additive approximation, and pay an additional $\tilde{O}(m)$ cost, which under the change of variables $x \defeq U^{-1} f$ recovers a corresponding flow. We call the full algorithm $\textsc{Flow-To-Regress}$.
\end{lemma}

In particular, we are able to use $R$ from the statement of \Cref{thm:caproperties}. For completeness, we will prove \Cref{lemma:flowreduction} in the appendices, but on a first read one may skip the proof and use the reduction statement as a black box result for the remaining analysis.

\subsection{Runtimes for accelerated maximum flow}
\label{section:impldetails}

Here, we provide a full description of how to implement relevant machinery for applying the tools from \Cref{sec:regression} for accelerating the minimization of a constrained $\ell_\infty$ function to the regression problem given in \Cref{lemma:flowreduction}. Due to the arguments presented in \Cref{appendix:b}, it suffices to bound the runtime of approximately solving the initial regression problem. 

\begin{definition}[Flow regression problem]
\label{def:flowregress}
The maximum flow regression problem asks to $\eps$-approximately minimize the function $\norm{Ax - b}_\infty$ subject to $x \in [-1, 1]^m$, $\norm{b}_\infty \le 1$, and for $\tilde{O}(1)$-column-sparse $A$ with $\norm{A}_\infty = \tilde{O}(1)$.
\end{definition}

Lemma~\ref{lemma:flowreduction} implies that the cost of finding an $\eps$-approximate maximum flow is (up to logarithmic factors) the same as solving the flow regression problem once.

\subsubsection{Applications of Section~\ref{sec:regression}}

We first show how to use the methods of Section~\ref{sec:regression} to obtain an improved maximum flow algorithm. First, note that by applying Theorem~\ref{thm:mainregressionthm} directly, combining with the properties given in  Theorem~\ref{thm:caproperties} of the regression matrix $A$, we immediately obtain a runtime of $\tilde{O}\left(m + (n + \sqrt{ms})/\eps\right)$ for the maximum flow problem, where the additive factor $\tilde{O}(m)$ comes from the preprocessing required in Section~\ref{ssec:cheapiter}, as well as the cost of computing the matrix $A$. Here, we used $\min(m, n) = n$ in the case of the flow regression matrix. We further can apply Theorem~\ref{thm:mainregressionthmdiagonal} to obtain a runtime of $\tilde{O}\left(m + \sqrt{mn}/\eps\right)$, where the dominant term is $\sqrt{mn}$ as $m = \Omega(n)$. Taking the better of these runtimes implies the following.

\begin{theorem}
\label{thm:fastishmaxflow}
There is an algorithm that takes time $\tilde{O}(m + (n + \sqrt{m\min(n, s)})/\eps)$ to find, with high probability, an $\eps$-approximate maximum flow, where $s$ is the squared $\ell_2$ norm of the congestion vector of any optimal flow.
\end{theorem}

\subsubsection{Tighter runtime dependence}

We develop an algorithm with an improved runtime for the flow regression problem in Section~\ref{sec:improvement}, based on directly applying a randomized mirror prox method to the primal-dual regression objective. Its runtime guarantee is stated here, and its full details are given in Section~\ref{sec:improvement}.

\begin{theorem}
\label{thm:improvement}
There is an algorithm, initialized at $x_0$, for finding an $\eps$-approximate minimizer to the flow regression problem (Definition~\ref{def:flowregress}), with high probability, in time $\tilde{O}(m +\max(n, \sqrt{ns})/\eps)$, where $s = \norm{x_0 - x^*}_2^2$.
\end{theorem}

By combining this improved algorithm with the reduction procedure of Lemma~\ref{lemma:flowreduction}, we obtain our fastest algorithm for maximum flow, generically improving upon Theorem~\ref{thm:fastishmaxflow}.

\begin{theorem}
	\label{thm:fastestmaxflow}
	There is an algorithm that takes time $\tilde{O}(m +\max(n, \sqrt{ns})/\eps)$ to find, with high probability, an $\eps$-approximate maximum flow, where $s$ is the squared $\ell_2$ norm of the congestion vector of any optimal flow.
\end{theorem}

\subsection{Exact maximum flows in uncapacitated graphs}
\label{ssec:exactflows}
Here, we describe several corollaries of our approach, for rounding to an exact maximum flow for several types of uncapacitated graphs. In an uncapacitated graph, $s = \|f^*\|_2^2 \leq Fn$ where $F$ is the maximum flow value, because the maximum flow is a 0-1 flow, and thus can be decomposed into $F$ $s-t$ paths with length at most $n$. We assume here that all the graphs are simple, and thus $m \leq n^2$; it is not difficult to generalize these results to non-simple graphs. As preliminaries, we state the following standard techniques for rounding to exact maximum flows.

\begin{lemma}[Theorem 5 in \cite{LeeRS13}]
There is a randomized algorithm that runs in expected time $\tilde{O}(m)$ which takes a fractional flow of value $F$ on an uncapacitated graph, and returns an integral flow of value $\left \lfloor{F}\right \rfloor$.
\end{lemma}

We will thus always assume that we have applied the rounding to an integral flow as a pre-processing step, as it will not affect our asymptotic runtime.

\begin{lemma}[Augmenting paths]
There is an algorithm that runs in time $O(m)$ which takes a non-maximal integral flow of value $F$ on an uncapacitated graph, and returns an integral flow of value $F + 1$.
\end{lemma}

Suppose we have a flow with value $(1 - \epsilon)F$, where the maximum flow value is $F$. The two lemmas for rounding and augmenting a flow therefore imply that the additional runtime required to attain an exact maximum flow is $O(\epsilon F m)$.

\subsubsection{Undirected uncapacitated graphs}

We state several corollaries of \Cref{thm:fastestmaxflow} which apply to finding exact maximum flows in various types of undirected uncapacitated graphs. All of these results only hold with high probability.

\begin{corollary}[Undirected graphs]
\label{cor:exactug}
There is an algorithm which finds a maximum flow in an undirected, uncapacitated graph in time $\tilde{O}(m^{5/4} n^{1/4})$.
\end{corollary}

\begin{proof}
We run the algorithm from \Cref{thm:fastestmaxflow} for $\epsilon = n^{1/4}/m^{3/4}$, and then run augmenting paths for $O(\epsilon m^2)$ iterations. Note that the maximum flow value and sparsity are bounded by $m$, and thus this will yield a maximum flow. Furthermore the runtime of the approximate algorithm is bounded by $m + \sqrt{nm}/\eps$. Putting together these two runtimes yields the result.
\end{proof}

\begin{corollary}[Undirected graphs with small maximum flow value]
\label{cor:exactsmallug}
There is an algorithm which finds a maximum flow in an undirected, uncapacitated graph with maximum flow value $F$ in time $\tilde{O}(m + \textup{min}(\sqrt{mn} F^{3/4}, m^{3/4} n^{1/4} \sqrt{F}))$.
\end{corollary}
\begin{proof}
The analysis here is the same as in \Cref{cor:exactug}, but instead we note that the bound on $s$ is $\textrm{min}(m, Fn)$, where the latter factor results from combining $F$ paths of length at most $n$. If the better bound is $Fn$, our runtime is bounded by $\tilde{O}(m + (n + \sqrt{n^2 F})/\epsilon + \epsilon Fm)$, and choosing $\epsilon = n^{1/2}/(F^{1/4}m^{1/2})$ yields the result. If the better bound is $m$, our runtime is bounded by $\tilde{O}(m + \sqrt{nm}/\eps + \epsilon Fm)$, and choosing $\epsilon = n^{1/4}/(\sqrt{F}m^{1/4})$ yields the result.
\end{proof}

\begin{corollary}[Undirected graphs with sparse optimal flow]
\label{cor:exactsparseug}
There is an algorithm which finds a maximum flow in an undirected, uncapacitated graph with a maximum flow that uses at most $s$ edges in time $\tilde{O}(m + \sqrt{ms}n^{1/4}\max(n, s)^{1/4})$.
\end{corollary}
\begin{proof}
The analysis here is the same as in \Cref{cor:exactug}, but instead we note that the bound on the maximum flow value is also $s$. Thus, our runtime is bounded by $\tilde{O}(m + \max(n, \sqrt{ns})/\eps + \epsilon sm)$. If $n \leq s$, choosing $\epsilon = n^{1/4}/(s^{1/4} m^{1/2})$ yields the result; otherwise, we choose $\epsilon = \sqrt{n/ms}$.
\end{proof}

\subsubsection{Directed graphs}

We follow the standard reduction of finding a maximum flow in a directed graph to finding a maximum flow in an undirected graph described in, for example, \cite{Lin09}. In short, an undirected graph with maximum flow value $O(m)$ is created, such that we can initialize the algorithm in \Cref{thm:fastestmaxflow} at a flow which is off from the true maximum flow by $s$ in $\ell_2^2$ distance. We give this reduction in \Cref{ssec:dirtoundir}, and refer the reader to \cite{Lin09} for a more detailed exposition.

Thus, after applying this reduction, the only difference in the runtimes given by the previous section are that the rounding algorithm will always take time $O(\epsilon m^2)$ instead of $O(\epsilon Fm)$. This immediately yields the following runtimes for exact maximum flows in directed graphs.

\begin{corollary}[Directed graphs]
\label{cor:exactdg}
There is an algorithm which finds a maximum flow in a directed, uncapacitated graph in time $\tilde{O}(m^{5/4} n^{1/4})$.
\end{corollary}

\begin{corollary}[Directed graphs with a sparse optimal flow]
\label{cor:exactsparsedg}
There is an algorithm which finds a maximum flow in a directed, uncapacitated graph in time $\tilde{O}(mn^{1/4}\max(n, s)^{1/4})$.
\end{corollary} %
\section{Improved Flow Runtimes via Primal-Dual Coordinate Regression}
\label{sec:improvement}

In this section, we prove Theorem~\ref{thm:improvement} by giving the algorithm and analyzing its runtime. Throughout, as in the statement of the flow regression problem (Definition~\ref{def:flowregress}), $A \in \R^{n \times m}$ has $\tilde{O}(1)$-sparse columns, $\norm{A}_{\infty} \leq 1$, and $\norm{b}_{\infty} \leq 1$, where we drop logarithmic factors in $\norm{A}_\infty$ for simplicity. We describe how to obtain a point $\hat{x}$ with $\norm{\hat{x}}_\infty \leq 1$, and
\begin{equation*}
\norm{A\hat{x} - b}_\infty - \epsilon \leq \textup{OPT} \defeq \norm{Ax^* - b}_\infty, \text{ where } x^* \defeq \argmin_{x \mid \norm{x}_\infty \leq 1} \norm{Ax - b}_\infty.
\end{equation*}
The runtime we will prove for the algorithm (initialized at the origin) is, as in Theorem~\ref{thm:fastestmaxflow},
\begin{equation*}
\tilde{O}\left(m + \frac{n + \sqrt{ns}}{\epsilon}\right), \text{ where } s \defeq \norm{x^*}_2^2.
\end{equation*}
Note if we wish to supply the algorithm with an initial point which is not the origin, as is the case for our results on maximum flow in directed graphs, it suffices to modify the definition of $b$ appropriately and shift by the initial point (see Appendix~\ref{appendix:reductionbox} for a more formal treatment). 

\subsection{Overview}

We first give an outline of our algorithm. The main motivation for the form it takes is to obtain the ``best of both worlds'' runtime of the form $\sqrt{ns}/\eps$. In terms of the dependence of Theorem~\ref{thm:mainregressionthm} on $s$, i.e. the sparsity of the optimal point, a standard (unweighted) Euclidean regularizer is necessary for the primal point $x \in [-1, 1]^m$. In terms of the dependence of Theorem~\ref{thm:mainregressionthmdiagonal} trading off an $n$ factor for an $m$, we require more fine-grained estimates on local coordinate smoothnesses based on dual information and properties of the matrix. We obtain both of these improvements in our final runtime via a fully primal-dual coordinate regression algorithm. 

Throughout, all divergences on $x$ space are with respect to $q(x) = \frac{1}{2s}\norm{x}_2^2$, on $y$ space\footnote{In this section, we use $y$ rather than $p$ to denote dual points, as they evolve separately; in our previous algorithms, $p$ was typically a probability distribution induced by a primal point $x$.} are with respect to $r(y) = \sum_i y_i \log y_i$, and on the product space are with respect to the direct sum (we drop superscripts in definitions of Bregman divergences in this section, as the regularizer will be fixed). 

\paragraph{Regularized subproblem.} The first step of our method is to define the following function, a regularized variant of the primal-dual formulation of the box-constrained $\norm{Ax - b}_\infty$ objective:
\begin{equation}
\label{eq:pdobjective}
h(x, y) \defeq y^\top (A x - b) + \frac{\epsilon}{2} q(x) - \frac{\epsilon}{4\log n} r(y).
\end{equation}
Throughout, we refer to the saddle point of the regularized objective $h$ by $\tilde{x} \in [-1, 1]^m, \tilde{y} \in \Delta^n$. The motivation for considering the regularized problem is related to technical issues which arise when generalizing Lemma~\ref{lem:conceptmp} to interact with a randomized algorithm; as we will see, returning the average iterate is computationally expensive for our coordinate method. We bypass this by providing a last-iterate guarantee via regularization, by arguing we can repeatedly return a point in each phase halving the distance to the saddle point.

The following lemma shows that to solve the box-constrained $\ell_\infty$ regression problem, it suffices to solve the regularized problem
$$\min_{x \in [-1, 1]^m} \max_{y \in \Delta^n} h(x, y)$$
to high accuracy. We also show that the regularized optimizer's $\ell_2$ sparsity is not too large.
\begin{lemma}
	\label{lem:txokay}
	$\norm{\tilde{x}}_2^2 \leq 2s$, and $\norm{A\tilde{x} - b}_{\infty} \leq \textup{OPT} + \frac{\epsilon}{2}$.
\end{lemma}
\begin{proof}
	Recall that the definition of $\smax_{\alpha}(x)$ implies
	\begin{equation*}
	\smax_{\alpha}(x) = \max_{y \in \Delta^n} y^\top (A x -  b) - \alpha r(y).
	\end{equation*}
	By Fact~\ref{fact:adderr},
	\begin{equation*}
	\norm{Ax - b}_\infty \leq \smax_{\epsilon/4\log n}(x) \leq \norm{Ax - b}_{\infty} + \frac{\epsilon}{4}.
	\end{equation*}
	Correspondingly, we have the following chain of inequalities:
	\begin{equation*}
	h(\tilde{x}, \tilde{y}) \leq h(x^*, \tilde{y}) \leq \smax_{\epsilon/4\log n}(x^*) + \frac{\epsilon}{2}q(x^*) \leq \textup{OPT} + \frac{\epsilon}{4} + \frac{\epsilon}{4} = \textup{OPT} + \frac{\epsilon}{2}.
	\end{equation*}
	The first inequality follows from minimality of $\tilde{x}$ with respect to $\tilde{y}$, the second from considering the terms in $h$ corresponding to $y$, and the last by the definition of $\textup{OPT}$ and $s$. Now, we also have
	\begin{equation*}
	h(\tilde{x}, \tilde{y}) = \smax_{\epsilon/4 \log n}(\tilde{x}) + \frac{\epsilon}{2} q(\tilde{x}) \geq \norm{A\tilde{x} - b}_{\infty} + \frac{\epsilon}{4s}\norm{\tilde{x}}_2^2.
	\end{equation*}
	Putting these together and using $\norm{A\tilde{x} - b}_{\infty} \geq \textup{OPT}$ by definition,
	\begin{equation*}
	\textup{OPT} + \frac{\epsilon}{4s}\norm{\tilde{x}}_2^2 \leq \textup{OPT} + \frac{\epsilon}{2} \Rightarrow \norm{\tilde{x}}_2^2 \leq 2s.
	\end{equation*}
	Similarly, the other conclusion follows by nonnegativity of $\frac{\epsilon}{4s}\norm{\tilde{x}}_2^2$.
\end{proof}

Consequently, an algorithm which is capable of obtaining a high-accuracy saddle point to $h$ suffices for minimizing the original objective.

\paragraph{Randomized mirror prox method.} We now describe one phase of our algorithm, which takes an initial point $z_{k, 0} = (x_{k , 0}, y_{k, 0})$, and returns a point $z_{k + 1, 0} = (x_{k + 1, 0}, y_{k + 1, 0})$ with
\begin{equation}\label{eq:halfv}\E[V_{z_{k + 1, 0}}(\tx, \ty)] \le \half V_{z_{k, 0}}(\tx, \ty).\end{equation}
Here, the expectation is over randomness used in the $k^{th}$ phase, i.e. the randomness used to define the point $z_{k + 1, 0}$. Combining this recursive guarantee via iterating expectations with the following initial bound (which uses Lemma~\ref{lem:txokay}) gives us a logarithmic bound on the number of phases.

\begin{lemma}
	\label{lem:range}
	Let $x_{0,0}$ be the all-zeroes vector and $y_{0,0} = \frac{1}{n}\1$. Then, $V_{x_{0, 0}, y_{0, 0}}(\tx, \ty) \leq \Theta_0 \defeq 1 + \log n$.
\end{lemma}

In order to obtain the guarantee \eqref{eq:halfv}, our starting point is Nemirovski's \emph{mirror prox} method \cite{Nemirovski04}, which can be viewed as a fixed-point iteration approximating the proximal point method (Definition~\ref{def:conceptmp}). Note that optimality conditions imply that iterating \eqref{eq:gensubproblem} in the proximal point method produces a sequence of iterates satisfying
\[z_{t + 1} \gets \argmin_{z}\left\{\inprod{g(z_{t + 1})}{z} + \alpha V_{z_t}(z).\right\}\]
However, this method is not implementable, as $z_{t + 1}$ uses its own gradient operator in its definition. Nemirovski's mirror prox approximates this process via a fixed-point iteration, by defining a two-step sequence
\begin{equation}\label{eq:mirrorprox}w_t \gets \argmin_{w}\left\{\frac{1}{\kappa}\inprod{g(z_{t})}{w} + V_{z_t}(w)\right\},\; z_{t + 1} \gets \argmin_{z}\left\{\frac{1}{\kappa}\inprod{g(w_t)}{z} + V_{z_t}(z)\right\}.\end{equation}
Here, the parameter $\kappa$ must be chosen to meet certain criteria so that the fixed-point iteration provably converges to a sufficient quality, and also governs the iteration count. Typically, $\kappa$ depends on the strong convexity of the regularizers $q$ and $r$, which leads to a dimension dependence in the runtime in the case of $\ell_\infty$ regression. \cite{Sherman17}
bypassed this by identifying a weaker criteria for the sequence \eqref{eq:mirrorprox} to converge. We obtain further improvements via a randomized variation of \eqref{eq:mirrorprox}.

Note that the gradient operator of the problem \eqref{eq:pdobjective} is:
\[g(x, y) \defeq \left(A^\top y + \frac{\epsilon}{2s} x, b - Ax + \frac{\epsilon}{4\log n}\log y\right).\]

A natural attempt unbiased estimator for $g$, inspired by the algorithm of Section~\ref{sec:regression}, is (for some sampling probabilities $\{p_j\}$) to randomly sample a coordinate of the primal block of $g$, i.e.
\begin{equation}\label{eq:gjideal}g_j(x, y) \defeq \left(\frac{1}{p_j} \left(\aj^\top y + \frac{\epsilon}{2s}x_j\right)e_j, b - Ax + \frac{\epsilon}{4\log n} \log y\right),\end{equation}
We would then define a step by: sample $j \sim \{p_j\}$, then iterate
\[w_t \gets \argmin_{w}\left\{\frac{1}{\kappa}\inprod{g_j(z_{t})}{w} + V_{z_t}(w)\right\},\; z_{t + 1} \gets \argmin_{z}\left\{\frac{1}{\kappa}\inprod{g_j(w_t)}{z} + V_{z_t}(z)\right\}.\]
However, in order to obtain our tight runtimes by leveraging a primal-dual analog of local coordinate smoothnesses, we require ``sharing randomness'' between these iterates, i.e. using the same coordinate $j$ in both steps. Note that in doing so, it no longer makes sense to say that $g_j(w_t)$ is an unbiased estimator for $g(w_t)$, as the choice of $j$ was used in the definition of $w_t$. We bypass this by defining an ``aggregate point'' $\bar{w}_t$ which $g_j(w_t)$ \emph{is} unbiased for, over the randomness of $w_t$. 

We then use a tight characterization of the convergence of our randomized method via local coordinate smoothnesses to argue about the quality of the average iterates $\bar{w}_t$, and show that randomly sampling one over $\tO(m + (n + \sqrt{ns})/\eps)$ iterations halves the divergence to $(\tx, \ty)$ in expectation. For this last step, we use the strong monotonicity\footnote{Strong monotonicity is a primal-dual analog of strong convexity.} of the objective $h$ to convert regret bounds into divergence bounds. Our complete algorithm concludes by repeating this procedure for $\tO(1)$ phases.

\paragraph{Roadmap.} Section~\ref{ssec:alg} states the algorithm, a randomized variation of mirror prox which uses the local coordinate smoothness ideas developed in  Section~\ref{sec:regression} in its analysis. It first develops a one-phase analysis, which leverages strong monotonicity of the objective $h$ in order to halve the distance to the true saddle point $(\tx, \ty)$ in $\tO(m + \max(n, \sqrt{ns})/\eps)$ iterations constituting a phase. It then uses the output of each phase as the starting point for the next phase, culminating in a high-accuracy saddle point in a logarithmic number of phases.

A key technical hurdle is that the iterates of the algorithm no longer have the sparse update structure used in the data structure development of Section~\ref{ssec:cheapiter}. In Section~\ref{ssec:runtime}, we show how to carefully use the structure of the updates to design a data structure based around Taylor approximation to perform iterations in batches, using nearly-constant amoritized time per iteration.

\subsection{Algorithm}
\label{ssec:alg}

Throughout, we index phases of the algorithm by $k \in [K]$, and iterates within a phase by $t \in [T]$. As discussed in the overview, we will choose $T = \tO(m + (n + \sqrt{ns})/\eps)$, and $K = \tO(1)$. 

Section~\ref{ssec:pd_prelim} defines local coordinate smoothness quantities which will factor into the algorithm. Section~\ref{ssec:singlephase} gives an analysis of a single phase of the algorithm, which outputs a point with expected divergence halved from the phase input. At the end of this section, we give a complete implementation of the phase, where we highlight issues with inexact implementation (which will be treated formally in Section~\ref{ssec:runtime}). Section~\ref{ssec:multiphase} leverages this single-phase method to give the complete algorithm, and proves the final runtime guarantee.

\subsubsection{Preliminaries}
\label{ssec:pd_prelim}
We first define some parameters used in the algorithm. For any $y \in \Delta^n$ we define for all $j$,
\begin{align*}
L_j(y) \defeq s\norm{\aj}_\infty |\aj|^\top y + \epsilon\norm{\aj}_\infty, \\
\tilde{L}_j(y) \defeq \left(\sqrt{s\norm{\aj}_\infty}\sum_{i \in [n]} \sqrt{|A_{ij}|y_i} + \sqrt{\epsilon\norm{\aj}_\infty}\right)^2.
\end{align*}
where $|\aj|$ is element-wise. These quantities will serve the role of local coordinate smoothness estimates in our algorithm and analysis. It is immediate that for all $j$,
\begin{equation}
\label{eq:ljcomp}
\tilde{O}(1) L_j(y) \geq \tilde{L}_j(y) \geq L_j(y),
\end{equation}
where the $\tilde{O}(1)$ factor is due to Cauchy-Schwarz and that each $\aj$ has $\tilde{O}(1)$ non-zero entries.
\begin{lemma}
	\label{lem:sumsqrtlj}
	For any $y$, $\sum_j \sqrt{\tilde{L}_j(y)} \leq C\sqrt{ns} + \sqrt{mn\epsilon}$, for some $C = \tilde{O}(1)$.
\end{lemma}
\begin{proof}
	Let all columns of $A$ have at most $c = \tilde{O}(1)$ nonzero entries. Then,
	\begin{equation}
	\begin{split}
	\left(\sum_{j \in [m]} \sum_{i \in [n]} \sqrt{\norm{\aj}_\infty y_i |A_{ij}|}\right)^2 & = \left( \sum_{j \in [m]} \sqrt{\norm{\aj}_\infty} 
	\cdot 
	\left[
	\sqrt{c} \cdot
	\sqrt{\sum_{i \in[n]}y_i |A_{ij}|}
	\right]
	\right)^2 \ \\
	& \leq c \left(\sum_{j \in [m]} \norm{\aj}_\infty\right) \left(\sum_{j \in [m]} \sum_{i \in [n]} y_i |A_{ij}|
	\right) \\
	& \leq nc \left(\sum_{i \in [n]} y_i \|A_i\|_1\right) \leq nc.
	\end{split}
	\end{equation}
	Here, the first line follows from Cauchy-Schwarz and using the fact that the sum $\sum_{i \in [n]} \sqrt{y_i |A_{ij}|}$ is $c$-sparse, the second line follows from Cauchy-Schwarz again, and the third line follows from the assumption $\norm{A}_\infty \le 1$. Thus,
	\begin{align*}
	\sum_{j \in [m]} \sqrt{\tilde{L}_j(y)} = \sqrt{s}\sum_{j \in [m]} \sum_{i \in [n]}\sqrt{\norm{\aj}_\infty |A_{ij}|y_i} + \sum_{j \in [m]} \sqrt{\epsilon \norm{\aj}_\infty} \\
	\leq \sqrt{nsc} + \sqrt{\epsilon}\sqrt{m\sum_{j \in [m]} \norm{\aj}_\infty} \leq \sqrt{nsc} + \sqrt{mn\epsilon}.
	\end{align*}
	It suffices to choose $C = \sqrt{c}$, where we used the column sparsity assumption.
\end{proof}
Finally, we define the following sampling distribution at any point $y$:
\begin{equation}
\label{eq:pjdef}
\begin{aligned}
p_j(y)= \frac{C\sqrt{ns}}{C\sqrt{ns} + \sqrt{mn\epsilon}} \cdot \frac{\sqrt{s\norm{\aj}_\infty}\sum_i \sqrt{|A_{ij}|y_i}}{\sum_j \sqrt{s\norm{\aj}_\infty}\sum_i \sqrt{|A_{ij}|y_i}} + \frac{\sqrt{mn\epsilon}}{C\sqrt{ns} + \sqrt{mn\epsilon}} \cdot \frac{\sqrt{\epsilon\norm{\aj}_\infty}}{\sum_j \sqrt{\epsilon\norm{\aj}_\infty}} \\
\geq \frac{\sqrt{\tilde{L}_j(y)}}{C\sqrt{ns} + \sqrt{mn\epsilon}}.
\end{aligned}
\end{equation}
The last inequality follows from the bounds from Lemma~\ref{lem:sumsqrtlj},
\[\sum_j \sqrt{s\norm{\aj}_\infty}\sum_i \sqrt{|A_{ij}|y_i} \le C\sqrt{ns},\; \sum_j \sqrt{\epsilon\norm{\aj}_\infty} \le \sqrt{mn\eps}.\]

We also make the simplifying assumption that at any $y$, all the sampling probabilities $p_j(y)$ are at least $1/(2m)$. To see why this is a valid assumption, our algorithm ultimately has a runtime depending linearly on our bound on $\sum_{j \in [m]} \sqrt{\tilde{L}_j(y)}$. By treating each $\sqrt{\tilde{L}_j(y)}$ as its sum with the average square root coordinate smoothness, this only doubles the overall sum (and therefore the bound in Lemma~\ref{lem:sumsqrtlj}), but enforces the lower bound on the sampling probabilities. This can be always be implemented by uniform sampling with half probability.

\subsubsection{Single phase analysis}
\label{ssec:singlephase}

In this section we give an analysis of the $k^{th}$ phase. We drop subscript $k$ from all iterates for simplicity, within the context of this section, until the very end. We define
\[\kappa \defeq m\epsilon + 8\sqrt{mn\epsilon} + 8C\sqrt{ns} + 16n,\]
the parameter which will ultimately govern the iteration count of the phase. We briefly discuss where each summand comes from in the analysis.
\begin{enumerate}
	\item The factor of $m\epsilon$ is used to account for terms of the form $\frac{\epsilon}{2p_j}[x_t]_j$ showing up in the randomized gradient estimator, which can be as large as $m\epsilon$, in Lemma~\ref{lem:threepoint}. This is the key lemma used to bound the progress of a single iteration.
	\item The factor of $8\sqrt{mn\epsilon}$ is used to ensure the stability of the simplex variable in a single iteration, due to the effects of terms of the form $\frac{\epsilon}{2p_j}[x_t]_j$, in Lemmas~\ref{lem:updatebound} and~\ref{lem:stable}. It is never the leading-order term, due to the terms $m\eps$ and $16n$.
	\item The factor of $8C\sqrt{ns}$ is used for both the error analysis and stability, due to effects of terms of the form $\frac{1}{p_j}A\Delta_t^{(j)}$, in Lemmas~\ref{lem:updatebound},~\ref{lem:stable}, and~\ref{lem:threepoint}.
	\item The factor of $16n$ is used to guarantee that $\kappa = \Omega(n)$. This is necessary in bounding the movement due to a fixed, dense term in the gradient updates, in the runtime analysis of Section~\ref{ssec:runtime}. In particular, it ensures we do not have to restart the data structure for simplex variable maintenance too frequently.
\end{enumerate} 
We now give one iteration of the phase, starting at a point $z_t = (x_t, y_t)$.
\begin{enumerate}
	\item Sample $j \propto p_j(y_t)$ 
	\item $x_{t + \half}^{(j)} \gets \argmin_{x \in [-1, 1]^m}\left\{\inprod{\frac{1}{\kappa p_j}(\aj^\top  y_t + \frac{\epsilon}{2s}[x_t]_j)e_j}{x} + V_{x_t}(x)\right\}$.
	\item $\hy \gets \argmin_{y \in \Delta^n}\left\{\inprod{\frac{1}{\kappa}\left(b - Ax_t + \frac{\epsilon}{4\log n} \log y_t\right)}{y} + V_{y_t}(y) \right\}$.
	\item $\Delta^{(j)}_t \defeq x_{t + \half}^{(j)} - x_t$.
	\item $\px \gets \argmin_{x \in [-1, 1]^m}\left\{\inprod{\frac{1}{\kappa p_j}(\aj^\top  y_{t + \half} + \frac{\epsilon}{2s}[\hx]_j)e_j}{x} + V_{x_t}(x)\right\}$.
	\item $\py \gets \argmin_{y \in \Delta^n}\left\{\inprod{\frac{1}{\kappa}\left(b - A\left(x_t + \frac{1}{p_j}\Delta_t^{(j)}\right) + \frac{\epsilon}{4\log n} \log \hy\right)}{y} + V_{y_t}(y)\right\}$.
\end{enumerate}
We remark that in all but possibly the $j^{th}$ coordinate, $\hx$ and $\px$ are identical to $\x$. We write
$$z_t = (x_t, y_t),\; w_t^{(j)} = \left(x_{t + \half}^{(j)}, y_{t + \half}\right),\; z_{t + 1}^{(j)} = \left(x_{t + 1}^{(j)}, y_{t + 1}^{(j)}\right).$$ 
We briefly remark on the form of the iterates. The gradient estimators inducing the points $\hx$, $\px$ are precisely those described by \eqref{eq:gjideal}, where we note that the point $\hy$ is deterministic (conditioned on $z_t$). Moreover, the gradient estimator inducing $\py$ is chosen so that our algorithm has the following property, which implies in each iteration, there is an ``aggregate point'' $\bar{w}_t$ whose regret we can bound. In this sense, the term $\frac{1}{p_j}\Delta_t^{(j)}$ can be viewed as a debiasing step.
\begin{lemma}
	\label{lem:average}
	Let $\bar{x}_{t + \half} = x_t + \sum_j \Delta_t^{(j)}$, the point taking all coordinate steps from $x_t$, and denote
	\begin{equation*}
	g_j(w_t^{(j)}) = \left(\frac{1}{p_j}\left(\aj^\top \hy + \frac{\epsilon}{2s} \left[\hx\right]_j\right)e_j,\; b - A\left(x_t + \frac{1}{p_j}\Delta_t^{(j)}\right) + \frac{\epsilon}{4\log n}\log \hy\right).
	\end{equation*}
	Then, we have for $\tilde{z} = (\tx, \ty)$,
	\begin{equation*}
	\E_j\left[\inprod{g_j(w_t^{(j)})}{w_t^{(j)} - \tz} \right] = \inprod{g(\bar{w}_t)}{\bar{w}_t - \tz},
	\end{equation*}
	where $\bar{w}_t = (\bar{x}_{t + \half}, y_{t + \half})$. 
\end{lemma}
\begin{proof}
	Recall that $\hx$ and $\bar{x}_{t + \half}$ agree in the $j^{th}$ coordinate. Then, expanding we have 
	\begin{align*}
	\E_j\left[\inprod{g_j(w_t^{(j)})}{w_t^{(j)} - \tz} \right] \\
	= \sum_j p_j \left(\inprod{\frac{1}{p_j}\aj^\top y_{t + \half}}{\bar{x}_{t + \half} - \tx} + \inprod{\frac{1}{p_j}\frac{\epsilon}{2s}[\bar{x}_{t + \half}]_j}{\bar{x}_{t + \half} - \tx} \right.\\
	\left. +\inprod{b - A\left(\x + \frac{1}{p_j}\Delta_t^{(j)}\right)}{y_{t + \half} - \ty} + \inprod{\frac{\epsilon}{4\log n}\log \hy}{\hy - \ty}\right)\\
	= \inprod{A^\top y_{t + \half}}{\bar{x}_{t + \half} - \tx} + \inprod{\frac{\epsilon}{2s}\bar{x}_{t + \half}}{\bar{x}_{t + \half} - \tx} \\
	+ \inprod{b - A\bar{x}_{t + \half}}{y_{t + \half} - \ty} + \inprod{\frac{\epsilon}{4\log n}\log \hy}{y_{t + \half} - \ty} \\
	= \inprod{g(\bar{w}_t)}{\bar{w}_t - \tz}.
	\end{align*}
\end{proof}

Next, we require the following bound on the size of the updates.

\begin{lemma}
	\label{lem:updatebound}
	For any $t$, call the updates to the simplex variables due to the bilinear term
	\begin{equation*}
	\delta_t \defeq \frac{1}{\kappa}(b - Ax_t),\; \delta^{(j)}_{t + \half} \defeq \frac{1}{\kappa}\left(b - A\left(x_t + \frac{1}{p_j}\Delta_t^{(j)}\right) \right).
	\end{equation*}
	Then, we have
	\begin{equation*}
	\max\left(\norm{\delta_t}_{\infty}, \norm{\delta_{t + \half}^{(j)}}_{\infty}\right) \leq \frac{1}{4}.
	\end{equation*}
\end{lemma}
\begin{proof}
	First, the bound on $\delta_t$ follows by $\norm{b - Ax_t}_\infty \leq 2$ and $\kappa$ is sufficiently large. Note that we may also conclude a stronger bound, that $\norm{\delta_t}_\infty \leq \frac{1}{8}$. By triangle inequality, it suffices to show 
	\begin{equation*}
	\norm{\frac{1}{\kappa p_j}A\Delta_t^{(j)}}_{\infty} \leq \frac{1}{8} \text{ for all } j \in [m].
	\end{equation*}
	Firstly, observe that $\Delta_t^{(j)}$ is 1-sparse, and can be bounded by noting (where $\text{med}$ takes a median)
	\begin{equation*}
	\left[\hx\right]_j = \text{med}\left(-1, [\x]_j - \frac{1}{\kappa p_j}\left(\frac{\epsilon}{2}[\x]_j + s \cdot \aj^\top \y \right), 1\right),
	\end{equation*}
	so that by definition of $\Delta_t^{(j)} = \hx - x_t$,
	\begin{equation*}
	\norm{\Delta_t^{(j)}}_\infty \leq \frac{\epsilon}{2\kappa p_j}\left|[\x]_j\right| + \frac{s}{\kappa p_j}\left|\aj^\top y_t\right|.
	\end{equation*}
	Recall that
	\begin{equation}\label{eq:pjkappa}p_j(y_t) \geq \frac{\sqrt{\tilde{L}_j(y_t)}}{C\sqrt{ns} + \sqrt{mn\epsilon}} \geq \frac{8\sqrt{L_j(y_t)}}{\kappa}.\end{equation}
	Here, the first inequality was from \eqref{eq:pjdef}, and the second was from the definition of $\kappa$ and \eqref{eq:ljcomp}. We now bound the size of entries of $\frac{1}{\kappa p_j}A\Delta_t^{(j)}$, recalling $\norm{\x}_\infty \leq 1$:
	\begin{align*}
	\frac{1}{\kappa p_j}\norm{A \Delta_t^{(j)}}_\infty \leq \frac{\epsilon}{2\kappa^2 p_j^2} \norm{\aj}_\infty + \frac{s}{\kappa^2 p_j^2}  \norm{\aj}_\infty \left|\aj^\top y_t \right| \\
	= \frac{1}{\kappa^2} \frac{1}{p_j^2} \left(\frac{\epsilon}{2}\norm{\aj}_\infty + s\norm{\aj}_\infty|\aj^\top \y|\right) \\
	\leq \frac{1}{64 L_j(y_t)}\left(\frac{\epsilon}{2}\norm{\aj}_\infty + s\norm{\aj}_\infty|\aj^\top \y|\right) \leq \frac{1}{64}.
	\end{align*}
	The last line follows from \eqref{eq:pjkappa}. This yields the claim, as $L_j(y_t) = s\norm{\aj}_\infty |\aj|^\top y_t + \epsilon \norm{\aj}_\infty$.
\end{proof}

Leveraging this, the following lemma shows multiplicative stability of the simplex variables within a single iteration, which allows us to show that local smoothness estimates do not drift significantly. This proof is somewhat technical, and is deferred until the end of Section~\ref{ssec:runtime}, as it requires opening up our implementation, which will yield the fact that the simplex points are not too unstable.

\begin{lemma}
	\label{lem:stable}
	Coordinate-wise for any $j$, $y_{t + \half},\; y^{(j)}_{t + 1}$ multiplicatively approximate $y_t$ by a factor of at most 8. That is (where division is coordinate-wise), $\max\left(y_{t + \half} / y_t, y^{(j)}_{t + 1} / y_t\right) \leq 8$.
\end{lemma}
We also require the following (standard) local norms bound on the divergence of entropy. 
\begin{lemma}[Local norms]
	\label{lem:localnorm}
	Let $y, y'$ be on the simplex. Then for $V$ the divergence with respect to entropy, $V_y(y') \geq \half\norm{y - y'}^2_{\diag(\max(y, y'))^{-1}}$.
\end{lemma}
\begin{proof}
	Let $y_{\alpha} = (1 - \alpha)y + \alpha y'$. By a Taylor expansion, letting $h$ be entropy, we have
	\begin{equation*}
	V_y(y') = \int_0^1 \int_0^\beta (y' - y)\nabla^2 h(y_{\alpha}) (y' - y) d\alpha d\beta \geq \half\norm{y - y'}^2_{\diag(\max(y, y'))^{-1}}.
	\end{equation*}
\end{proof}

We now give a one-step convergence analysis of our algorithm, where use the definitions
\begin{align*}
g_j(z_t) &\defeq \left(\frac{1}{p_j} \left(\aj^\top y_t + \frac{\epsilon}{2s}[\x]_j\right)e_j,\; b - Ax_t + \frac{\epsilon}{4\log n} \log \y\right), \\
g_j(w_t^{(j)}) &\defeq \left(\frac{1}{p_j}\left(\aj^\top \hy + \frac{\epsilon}{2s} \left[\hx\right]_j\right)e_j,\; b - A\left(x_t + \frac{1}{p_j}\Delta_t^{(j)}\right) + \frac{\epsilon}{4\log n}\log \hy\right).
\end{align*}
\begin{lemma}
	\label{lem:threepoint}
	On any iteration $t$, we have (where expectations are over the randomness of the coordinate $j$ in the iteration)
	\begin{equation*}
	\E\left[\frac{1}{\kappa}\inprod{g_j(w_t^{(j)})}{w_t^{(j)} - \tz} \right] \leq \E\left[V_{z_t}(\tz) - V_{z_{t + 1}^{(j)}}(\tz)\right].
	\end{equation*}
\end{lemma}
\begin{proof}
	Applying the first-order optimality conditions defining the two steps, as well as \eqref{eq:threepoint} following from the definition of Bregman divergences,
	\begin{align*}
	\frac{1}{\kappa}\inprod{g_j(z_t)}{w_t^{(j)} - z_{t + 1}^{(j)}} &\leq V_{z_t}(z_{t + 1}^{(j)}) - V_{w_t^{(j)}}(z_{t + 1}^{(j)}) - V_{z_t}(w_t^{(j)})\\
	\frac{1}{\kappa}\inprod{g_j(w_t^{(j)})}{z_{t + 1}^{(j)} - \tz} &\leq V_{z_t}(\tz) - V_{z_{t + 1}^{(j)}}(\tz) - V_{z_t}(z_{t + 1}^{(j)}).
	\end{align*}
	Summing and rearranging slightly, we have
	\begin{align*}
	\frac{1}{\kappa}\inprod{g_j(w_t^{(j)})}{w_t^{(j)} - \tz} &\leq V_{z_t}(\tz) - V_{z_{t + 1}^{(j)}}(\tz)\\
	&+ \frac{1}{\kappa}\inprod{g_j(w_t^{(j)}) - g_j(z_t)}{w_t^{(j)} - z_{t + 1}^{(j)}} - V_{w_t^{(j)}}(z_{t + 1}^{(j)}) - V_{z_t}(w_t^{(j)}).
	\end{align*}
	Taking an expectation, we have the conclusion up to proving the following claim, where we recall $\kappa = m\epsilon + 8\sqrt{mn\epsilon} + 8C\sqrt{ns} + 16n$:
	\begin{equation*}
	\E\left[\inprod{g_j(w_t^{(j)}) - g_j(z_t)}{w_t^{(j)} - z_{t + 1}^{(j)}} \right] \leq \kappa\E\left[V_{w_t^{(j)}}(z_{t + 1}^{(j)}) + V_{z_t}(w_t^{(j)})\right].
	\end{equation*}
	We will instead show the stronger claim that this is true for any particular $j \in [m]$: 
	\begin{equation}
	\label{eq:strongerclaim}
	\inprod{g_j(w_t^{(j)}) - g_j(z_t)}{w_t^{(j)} - z_{t + 1}^{(j)}} \leq \kappa\left(V_{w_t^{(j)}}(z_{t + 1}^{(j)}) + V_{z_t}(w_t^{(j)})\right).
	\end{equation}
	We will roughly do so by splitting the left hand side into three pieces, and then bounding them separately. First, we rewrite it as (recalling $\Delta_t^{(j)} = \hx - \x$ is 1-sparse)
	\begin{equation}
	\label{eq:3guys}
	\begin{aligned}
	\inprod{g_j(w_t^{(j)}) - g_j(z_t)}{w_t^{(j)} - z_{t + 1}^{(j)}} \\
	= \frac{1}{p_j}\left(\inprod{\aj^\top (y_{t + \half} - y_{t})e_j}{\hx - \px} + \inprod{\aj^\top (y_{t + 1}^{(j)} - y_{t + \half})e_j}{\hx - x_t}\right) \\
	+ \frac{\epsilon}{2s p_j}\inprod{\left(\left[\hx\right]_j - [\x]_j\right)e_j}{\hx - \px} + \frac{\epsilon}{4 \log n}\inprod{\log \frac{\hy}{\y}}{\hy - \py}.
	\end{aligned}
	\end{equation}
	We bound the first term. By Lemma~\ref{lem:localnorm}, and as Lemma~\ref{lem:stable} gives coordinatewise $y_{t + 1}^{(j)}$, $y_{t + \half} \leq 8y_t$,
	\begin{align*}
	V_{w_t^{(j)}}(z_{t + 1}^{(j)}) + V_{z_t}(w_t^{(j)}) &\geq \frac{1}{16}\norm{y_{t + 1}^{(j)} - y_{t + \half}}_{\diag(y_t^{-1})}^2 + \frac{1}{2s}\norm{x_{t + 1}^{(j)} - x_{t + \half}^{(j)}}_2^2 \\
	&+ \frac{1}{16}\norm{y_{t + \half} - y_{t}}_{\diag(y_t^{-1})}^2 + \frac{1}{2s}\norm{x_{t + \half}^{(j)} - x_t}_2^2.
	\end{align*}
	We see that by $a^2 + b^2 \geq 2ab$ and Cauchy-Schwarz, 
	\begin{align*}
	\sqrt{\norm{\aj}_\infty |\aj|^\top y_t}\left(\frac{1}{16}\norm{y_{t + 1}^{(j)} - y_{t + \half}}_{\diag(y_t^{-1})}^2 + \frac{1}{2s}\norm{x_{t + \half}^{(j)} - x_t}_2^2\right) \\
	\geq \sqrt{|\aj^2|^\top y_t} \left(\frac{1}{\sqrt{8s}} \left|\left[x_{t + \half}^{(j)} - x_t\right]_j\right| \norm{y_{t + 1}^{(j)} - y_{t + \half}}_{\diag(y_t^{-1})} \right)\\
	= \left(\frac{1}{\sqrt{8s}} \left|\left[x_{t + \half}^{(j)} - x_t\right]_j\right|\right) \cdot \sqrt{\sum_i A_{ij}^2 [y_t]_i} \sqrt{\sum_i \frac{[y_{t + 1}^{(j)} - y_{t + \half}]_i^2}{[y_t]_i}}\\
	\geq \left(\frac{1}{\sqrt{8s}} \left|\left[x_{t + \half}^{(j)} - x_t\right]_j\right|\right) \cdot \sum_i |A_{ij}| |[y_{t + 1}^{(j)} - y_{t + \half}]_i|\\
	\geq \frac{1}{\sqrt{8s}} \inprod{\aj^\top (y_{t + 1}^{(j)} - y_{t + \half})e_j}{\hx - x_t}.
	\end{align*}
	Similarly, we have
	\begin{align*}
	\sqrt{\norm{\aj}_\infty |\aj|^\top y_t}\left(\frac{1}{16}\norm{y_{t + \half} - y_{t}}_{\diag{y_t^{-1}}}^2 + \frac{1}{2s}\norm{x_{t + 1}^{(j)} - x_{t + \half}^{(j)}}_2^2\right) \\
	\geq \frac{1}{\sqrt{8s}} \inprod{\aj^\top (y_{t + \half} - y_{t})e_j}{\hx - \px}.
	\end{align*}
	Therefore, by the three above equations,
	\begin{equation}
	\label{eq:firstguy}
	\begin{aligned}
	\frac{\sqrt{\norm{\aj}_\infty |\aj|^\top y_t} \sqrt{8s}}{p_j}\left(V_{w_t^{(j)}}(z_{t + 1}^{(j)}) + V_{z_t}(w_t^{(j)}) \right)\\ \geq \frac{1}{p_j}\left(\inprod{\aj^\top (y_{t + 1}^{(j)} - y_{t + \half})e_j}{\hx - x_t} + \inprod{\aj^\top (y_{t + \half} - y_{t})e_j}{\hx - \px}\right).
	\end{aligned}
	\end{equation}
	Now, we consider the second term. Directly applying strong-convexity and Cauchy-Schwarz gives
	\begin{equation}
	\label{eq:secondguy}
	\begin{aligned}
	\frac{\epsilon}{2s p_j}\inprod{\left(\left[\hx\right]_j - [\x]_j\right)e_j}{\hx - \px} \leq \frac{\epsilon}{2p_j} \left(\frac{1}{2s}\norm{\hx - \x}_2^2 + \frac{1}{2s}\norm{\hx - \px}_2^2\right) \\
	\leq \frac{\epsilon}{2p_j} \left(V_{\x}(\hx) + V_{\hx}(\px)\right).
	\end{aligned}
	\end{equation}
	Finally, we consider the third term. It is straightforward to note that for any convex $r$ (in this case, entropy), $\inprod{\nabla r(b) - \nabla r(a)}{b - c} \leq V_a(b) + V_b(c)$ for any three points $a, b, c$. Applying this,
	\begin{align}
	\label{eq:thirdguy}
	\frac{\epsilon}{4\log n}\inprod{\log \frac{\hy}{\y}}{\hy - \py} \leq \frac{\epsilon}{4\log n} \left(V_{\y}(\hy) + V_{\hy}(\py)\right).
	\end{align}
	Combining~\eqref{eq:firstguy},~\eqref{eq:secondguy},~\eqref{eq:thirdguy}, we obtain
	\begin{align*}
	\inprod{g_j(\hz) - g_j(\z)}{\hz - \pz} \leq \left(\frac{\sqrt{\norm{\aj}_\infty |\aj|^\top y_t} \sqrt{8s}}{p_j} + \frac{\epsilon}{2p_j}\right)\left(V_{\x}(\hx) + V_{\hx}(\px)\right) \\
	+ \left(\frac{\sqrt{\norm{\aj}_\infty |\aj|^\top y_t} \sqrt{8s}}{p_j} + \frac{\epsilon}{4 \log n}\right)\left(V_{\y}(\hy) + V_{\hy}(\py)\right).
	\end{align*}
	Finally, to prove~\eqref{eq:strongerclaim}, it remains to bound the size of the coefficients of the divergences by $\kappa$, and use nonnegativity of divergences. We claim the following holds:
	\begin{equation*}
	\frac{\sqrt{\norm{\aj}_\infty |\aj|^\top y_t} \sqrt{8s}}{p_j} + \frac{\epsilon}{2p_j} \leq \kappa = 16n + 8\sqrt{mn\epsilon} + 8C\sqrt{ns} + m\epsilon.
	\end{equation*}
	To see this, recall we assumed $p_j \geq \frac{1}{2m}$, and further that $\frac{1}{p_j} \leq \frac{C\sqrt{ns} + \sqrt{mn\epsilon}}{\sqrt{\tilde{L}_j(y_t)}}$ by \eqref{eq:pjdef}. Therefore,
	\begin{equation*}
	\frac{\sqrt{\norm{\aj}_\infty |\aj|^\top y_t} \sqrt{8s}}{p_j} + \frac{\epsilon}{2p_j} \leq \sqrt{8}(C\sqrt{ns} + \sqrt{mn\epsilon}) + m\epsilon \leq \kappa.
	\end{equation*}
	Similarly, it is easy to see that the following holds (corresponding to the coefficient of the divergences on the $y$ side), concluding the proof:
	\begin{equation*}
	\frac{\sqrt{\norm{\aj}_\infty |\aj|^\top y_t} \sqrt{8s}}{p_j} + \frac{\epsilon}{4 \log n} \leq \kappa.
	\end{equation*}
\end{proof}
We require the following helper lemma which upper bounds divergence via regret, which allows us to finally convert our regret bound into a divergence bound for the output iterate.
\begin{lemma}
	\label{lem:regtodiv}
	For any point $z$, we have $\inprod{g(z)}{z - \tz} \geq \frac{\epsilon}{4\log n} V_z(\tz)$, where we recall for $z = (x, y)$,
	\begin{equation*}
	g(z) = \left(A^\top y + \frac{\epsilon}{2s} x,\; b - Ax + \frac{\epsilon}{4\log n}\log y\right).
	\end{equation*}
\end{lemma}
\begin{proof}
	Recall that because $\tz$ is the saddle point of the convex-concave function $g$ is the gradient operator of, by first-order optimality,
	\begin{equation*}
	\inprod{g(\tz)}{z - \tz} \geq 0 \; \forall z.
	\end{equation*}
	Therefore, noting terms $\inprod{A^\top (y - \ty)}{x - \tx} - \inprod{A(x - \tx)}{y - \ty}$ cancel,
	\begin{align*}
	\inprod{g(z)}{z - \tz} \geq \inprod{g(z) - g(\tz)}{z - \tz} =  \frac{\epsilon}{2s}\norm{x - \tx}_2^2 + \frac{\epsilon}{4\log n}\inprod{\log \frac{y}{\ty}}{y - \ty}\\
	\geq \epsilon V_x(\tx) + \frac{\epsilon}{4\log n} V_y(\ty) \geq \frac{\epsilon}{4\log n} V_z(\tz).
	\end{align*}
	The last line used nonnegativity of the Bregman divergence and $\inprod{\nabla r(y) - \nabla r(\ty)}{y - \ty} \geq V_y(\ty)$.
\end{proof}

We now give the method in phase $k$, initialized at $(x_{k, 0}, y_{k, 0})$. Here, we briefly comment on inexactness issues. The algorithm requires maintenance of variable $y$ on the simplex, and various quantities which are functions of $y$, which we can only approximately compute (cheaply): the method outlined in Section~\ref{ssec:cheapiter} no longer applies, because updates to the variable are dense. Formally, we define $\yor$, a data structure which maintains an internal representation of the simplex variables. $\yor$ supports the following operations in each iteration $(k, t)$, in amoritized $\tO(1)$ time: 
\begin{itemize}
	\item $\yor.\texttt{Sample}()$: Samples $j \in [m]$ from $p_j(y_{k, t})$. Returns $(j, p_j(y_{k, t}))$.
	\item $\yor.\texttt{Coord}(i)$: Returns $[\breve{y}_{k, t}]_i$ such that $|[\breve{y}_{k, t}]_i - [y_{k, t}]_i| < n^{-100}$.
	\item $\yor.\texttt{Update-Half}(v)$: Updates the internal representation of $y_{k, t+ \half}$.
	\item $\yor.\texttt{Coord-Half}(i)$: Returns $[\breve{y}_{k, t + \half}]_i$ such that $|[\breve{y}_{k, t + \half}]_i - [y_{k, t + \half}]_i| < n^{-100}$.
	\item $\yor.\texttt{Update}(v)$: Updates the internal representation of $y_{k, t + 1}$.
\end{itemize}

We develop $\yor$ in Section~\ref{ssec:runtime}. The following is the algorithm for phase $k$.

\begin{enumerate}
\item Let $\kappa = m\epsilon + 8\sqrt{mn\epsilon} + 8C\sqrt{ns} + 16n$ where $C$ is the constant of Lemma~\ref{lem:sumsqrtlj}.
\item Let $T = \left\lceil\frac{8\kappa \log n}{\epsilon}\right\rceil$ be the number of iterations per phase.
\item Sample a stopping iteration uniformly at random $t_k^* \in [T]$.
	\item For iteration $t \in [t_k^* - 1]$: \begin{enumerate}
		\item Call $\yor.\texttt{Sample}$ to obtain $j$, $p_j(y_{k, t})$ (for shorthand, denoted $p_j$).
		\item For each non-zero entry $A_{ij}$ of $\aj$, call $\yor.\texttt{Coord}(i)$ to obtain $[\breve{y}_{k, t}]_i$.
		\item $x_{k, t + \half}\gets \argmin_{x \in [-1, 1]^m}\left\{\inprod{\frac{1}{\kappa p_j}(\aj^\top \breve{y}_{k, t} + \frac{\epsilon}{2s}[x_{k, t}]_j)e_j}{x} + V_{x_{k, t}}(x)\right\}$.
		\item $y_{k, t + \half} \gets \argmin_{y \in \Delta^n}\left\{\inprod{\frac{1}{\kappa}\left(b - Ax_{k, t} + \frac{\epsilon}{4\log n} \log y_{k, t}\right)}{y} + V_{y_{k, t}}(y) \right\}$.
		\item $\Delta_{k, t} \defeq x_{k, t + \half} - x_{k, t}$.
		\item For each non-zero entry $A_{ij}$ of $\aj$, call $\yor.\texttt{Coord}(i)$ to obtain $[\breve{y}_{k, t + \half}]_i$.
		\item $x_{k, t + 1} \gets \argmin_{x \in [-1, 1]^m}\left\{\inprod{\frac{1}{\kappa p_j}(\aj^\top \breve{y}_{k, t + \half} + \frac{\epsilon}{2s}[x_{k, t + \half}]_j)e_j}{x} + V_{x_{k, t}}(x)\right\}$.
		\item $y_{k, t + 1} \gets \argmin_{y \in \Delta^n}\left\{\inprod{\frac{1}{\kappa}\left(b - A\left(x_{k, t} + \frac{1}{p_j}\Delta_{k, t}\right) + \frac{\epsilon}{4\log n} \log y_{k, t + \half}\right)}{y} + V_{y_{k, t}}(y)\right\}$.
	\end{enumerate}
	\item For iteration $t = t_k^*$: \begin{enumerate}
		\item $\forall j \in [m]$, $\Delta^{(j)}_{k, t} \defeq \argmin_{x \in [-1, 1]^m}\left\{\inprod{\frac{1}{\kappa p_j}(\aj^\top y_{k, t} + \frac{\epsilon}{2s}[x_{k, t}]_j)e_j}{x} + V_{x_{k, t}}(x)\right\} - x_{k, t}$.
		\item Compute $\Delta_{k, t}^{(j)}$ for all $j \in [m]$. 
		\item Define $x_{k + 1, 0} = x_{k, t} + \sum_j \Delta_{k, t}^{(j)}$.
		\item Define $y_{k + 1, 0} = \argmin_{y \in \Delta^n}\left\{\inprod{\frac{1}{\kappa}\left(b - Ax_{k, t} + \frac{\epsilon}{4\log n} \log y_{k, t}\right)}{y} + V_{y_{k, t}}(y) \right\}$.
	\end{enumerate}
	\item Output $(x_{k + 1, 0}, y_{k + 1, 0})$.
\end{enumerate}

We remark that in each loop of step 4, steps (c), (e) and (g) are implemented directly in $\tilde{O}(1)$ time, step (d) is implemented implicitly using $\yor.\texttt{Update-Half}$, and step (h) is implemented implicitly using $\yor.\texttt{Update}$. We will discuss the efficient implementation of the procedures supported by $\yor$ in Section~\ref{ssec:runtime}.

We now come to the main export of this section, which shows that the expected divergence to the saddle point halves in every phase. In this lemma, we assume exact implementation of the steps; we discuss how to deal with inexactness issues in the analysis in Section~\ref{ssec:approxsum}.

\begin{lemma}
	\label{lem:mainsinglephase}
	Suppose phase $k$ is initialized with  $z_{k, 0} = (x_{k, 0}, y_{k, 0})$. Then, the output $z_{k + 1, 0} = (x_{k + 1, 0}, y_{k + 1, 0})$ satisfies (where expectations are taken over the randomness used in phase $k$)
	\begin{equation*}
	\E\left[V_{z_{k + 1, 0}}(\tx, \ty)\right] \leq \half V_{z_{k, 0}}(\tx, \ty).
	\end{equation*}
\end{lemma}
\begin{proof}
	Consider running for all of the $T = \left\lceil\frac{8 \kappa \log n}{\epsilon}\right\rceil$ iterations. Taking an expectation of Lemma~\ref{lem:threepoint} over the entire phase, telescoping, and using nonnegativity of  divergences, we obtain (for $\tz = (\tx, \ty)$)
	\begin{equation*}
	\E\left[\frac{1}{T}\sum_{t \in [T]} \inprod{g_{j_t}(w_{k, t})}{w_{k, t} - \tz}\right] \leq \frac{\kappa V_{z_{k, 0}}(\tz)}{T} \leq \frac{\epsilon}{8\log n} V_{z_{k, 0}}(\tz).
	\end{equation*}
	Here, $j_t$ is the coordinate sampled in the $t^{th}$ iteration. Applying Lemma~\ref{lem:average}, we instead have
	\begin{equation*}
	\E\left[\frac{1}{T}\sum_{t \in [T]} \inprod{g(\bar{w}_{k, t})}{\bar{w}_{k, t} - \tz}\right] \leq \frac{\epsilon}{8\log n} V_{z_{k, 0}}(\tz).
	\end{equation*}
	Now, because we randomly sampled a $t \in [T]$ to be the index $t_k^*$ and passed $\bar{w}_{t_k^*}$ to the $k + 1^{st}$ phase as $z_{k + 1, 0} = (x_{k + 1, 0}, y_{k + 1, 0})$, we obtain
	\begin{equation*}
	\E\left[\inprod{g(z_{k + 1, 0})}{z_{k + 1, 0} - \tz}\right] \leq \frac{\epsilon}{8\log n} V_{z_{k, 0}}(\tz).
	\end{equation*}
	The conclusion follows from applying Lemma~\ref{lem:regtodiv}.
\end{proof}

\subsubsection{Algorithm statement}
\label{ssec:multiphase}

We now state the full algorithm, which is composed of phases, each of which halves the expected divergence to the saddle point. 

\begin{enumerate}
	\item Initialize $x_{0, 0} = 0, y_{0, 0} = \frac{1}{n}\1$.
	\item Let $\Theta_0 = 1 + \log n$ be the initial divergence bound (Lemma~\ref{lem:range}). Let $K = \left\lceil\log_2\left(\frac{16s\Theta_0}{\epsilon^2}\right)\right\rceil$.
	\item For phase $0 \le k < K$: \begin{enumerate}
		\item Run the procedure in Section~\ref{ssec:singlephase}, initialized at $z_{k, 0}$, to produce the point $z_{k + 1, 0}$.
	\end{enumerate}
	\item Return $\hat{x} \defeq x_{K, 0}$.
\end{enumerate}

We now analyze the correctness and runtime of this algorithm. We assume the following lemma, which will be proven in Section~\ref{ssec:runtime}.

\begin{lemma}
	\label{lem:mainruntime}
	Every $n$ iterations of each phase can be implemented in $\tilde{O}(n)$ time. Furthermore, for each phase $k$, iteration $t_k^*$ can be implemented in $\tilde{O}(m)$ time.
\end{lemma}

\begin{theorem}
	\label{thm:main}
	The algorithm has runtime
	$$\tilde{O}\left(m + \frac{n + \sqrt{ns}}{\epsilon}\right),$$
	and satisfies $\E[\norm{\hat{x} - \tx}_2] \leq \frac{\epsilon}{2}$, where the expectation is over all randomness in the algorithm.
\end{theorem}
\begin{proof}
	To prove the first statement, note that the algorithm computes at most $K$ points of the form $\bar{w}_{k, t_k^*}$, and takes at most $KT$ steps. Thus, by Lemma~\ref{lem:mainruntime} this yields a runtime of
	\begin{equation*}
	\tilde{O}\left(KT\right) + \tilde{O}(mK) = \tilde{O}\left(\frac{\kappa}{\epsilon} + m\right) = \tilde{O}\left(\frac{m\epsilon + \sqrt{mn\epsilon} + \sqrt{ns} + n}{\epsilon}\right) = \tilde{O}\left(m + \frac{n + \sqrt{ns}}{\epsilon}\right).
	\end{equation*}
	We used that $\sqrt{mn/\epsilon}$ is never larger than $m + n/\epsilon$, as it is their geometric mean. To prove the second statement, we apply Lemma~\ref{lem:mainsinglephase} for $K \geq \log\left(\frac{16s\Theta_0}{\epsilon^2}\right)$:
	\begin{equation*}
	\E\left[\norm{\hat{x} - \tx}_2\right]^2 \leq \E\left[\norm{\hat{x} - \tx}_2^2\right] = 4s\E[V_{\hat{x}}(\tx)] \leq \frac{4s\Theta_0}{2^K} \leq \left(\frac{\epsilon}{2}\right)^2.
	\end{equation*}
	The first inequality used convexity of the square, and the second inequality repeatedly used Lemma~\ref{lem:mainsinglephase} and iterated expectations. This implies the desired bound. 
\end{proof}

We then see that $\hat{x}$ is our desired approximate minimizer, in expectation.

\begin{corollary}
	We have $\norm{\hat{x}}_\infty \leq 1$, and $\E[\norm{A\hat{x} - b}_\infty] \leq \textup{OPT} + \epsilon$.
\end{corollary}
\begin{proof}
	The first statement is immediate from the algorithm, since in each iteration $x_{k, t}$ lies in $[-1, 1]^m$, and for each $j$, $x_{k, t} + \Delta_{k, t}^{(j)}$ is also defined to lie in $[-1, 1]^m$, and the region decomposes coordinatewise. The second statement follows from $\norm{A}_\infty \leq 1$, and
	$$\E[\norm{A\hat{x} - b}_{\infty}] \leq \norm{A\tilde{x} - b}_{\infty} + \E[\norm{A(\hat{x} - \tx)}_{\infty}] \leq \textup{OPT} + \frac{\epsilon}{2} + \E[\norm{\hat{x} - \tx}_2] \leq \textup{OPT} + \epsilon.$$
\end{proof}

By Markov's inequality, this means that with half probability we have a $2\epsilon$-approximate minimizer. This can be boosted to probability $1 - \delta$ using $\log \frac{1}{\delta}$ independent runs, and it does not affect runtime asymptotically since computing objective value takes time $\tilde{O}(m)$.

\subsection{Runtime}
\label{ssec:runtime}

This section proves Lemma~\ref{lem:mainruntime}, which states that we can implement each iteration of each phase in amoritized $\tilde{O}(1)$ time for each $t \neq t_k^*$, and that we can implement the last iteration in $\tilde{O}(m)$ time. As discussed in Section~\ref{ssec:multiphase}, it suffices to show that $\yor.\left\{\texttt{Sample}, \texttt{Coord}, \texttt{Update}, \texttt{Update-Half}\right\}$ may be implemented in amoritized time $\tilde{O}(n)$ every $n$ iterations. In particular, assuming these operations are supported, it is simple to see that we can implement the updates to the $x$ variables in $\tilde{O}(1)$ time per iteration by sparsity. Finally, the last iteration can be implemented in $\tilde{O}(m)$ time simply by performing the updates to the $x$ variable $m$ times.

\subsubsection{Reducing sampling from and computing $p_j$ to sampling from and computing $\sqrt{y}$}
\label{ssec:samplereduce}
We first reduce the implementation of $\yor.\texttt{Sample}$ in an iteration $(k, t)$ to being able to efficiently sample proportional to $\sqrt{[y_{k, t}]_i}$ (we drop $(k, t)$ for simplicity). Recall we sample from
$$p_j = \frac{C\sqrt{ns}}{C\sqrt{ns} + \sqrt{mn\epsilon}} \cdot \frac{\sqrt{s\norm{\aj}_\infty}\sum_i \sqrt{|A_{ij}|y_i}}{\sum_j \sqrt{s\norm{\aj}_\infty}\sum_i \sqrt{|A_{ij}|y_i}} + \frac{\sqrt{mn\epsilon}}{C\sqrt{ns} + \sqrt{mn\epsilon}} \cdot \frac{\sqrt{\epsilon\norm{\aj}_\infty}}{\sum_j \sqrt{\epsilon\norm{\aj}_\infty}}.$$

First flip a coin which is heads with probability $C\sqrt{ns}/(C\sqrt{ns} + \sqrt{mn\epsilon})$. If it comes up tails, we sample a $j$ proportional to $\sqrt{\epsilon\norm{\aj}_\infty}$; clearly we may precompute all of these probabilities, and place them at the leaves of a binary tree (along with each of the subtree sums stored at roots of subtrees), flipping $\tilde{O}(1)$ appropriately biased coins to sample from this distribution. Next, in order to sample from a distribution over $j$ proportional to $\sqrt{s\norm{\aj}_\infty}\sum_i \sqrt{|A_{ij}|y_i}$, it clearly suffices to instead sample an $i$ proportional to $\sqrt{y_i}$, and then sample a $j$ proportional to $\sqrt{s\norm{\aj}_\infty|A_{ij}|}$; this latter distribution we can precompute.

We now discuss computing a particular $p_j$ in $\tilde{O}(1)$ time: we need to in fact compute the true $p_j$ which we sampled from, because otherwise we will not have an unbiased estimator. To do so, it clearly suffices to compute the conditional probabilities
\begin{equation*}
\frac{\sqrt{s\norm{\aj}_\infty}\sum_i \sqrt{|A_{ij}|y_i}}{\sum_j \sqrt{s\norm{\aj}_\infty}\sum_i \sqrt{|A_{ij}|y_i}},\; \frac{\sqrt{\epsilon\norm{\aj}_\infty}}{\sum_j \sqrt{\epsilon\norm{\aj}_\infty}}.
\end{equation*}
The latter of these is simple to pre-compute. To compute the former, let 
\begin{equation*}
q_{ij} \defeq \frac{\sqrt{s\norm{\aj}_\infty|A_{ij}|}}{\sum_i \sqrt{s\norm{\aj}_\infty|A_{ij}|}}.
\end{equation*}
We observe that for any $j$, at most $\tilde{O}(1)$ of the $q_{ij}$ are non-zero, and we can also precompute all the $q_{ij}$. Finally, the conclusion follows from
\begin{equation*}
\frac{\sqrt{s\norm{\aj}_\infty}\sum_i \sqrt{|A_{ij}|y_i}}{\sum_j \sqrt{s\norm{\aj}_\infty}\sum_i \sqrt{|A_{ij}|y_i}} =\sum_i \frac{\sqrt{y_i}}{\sum_i \sqrt{y_i}} \cdot q_{ij}.
\end{equation*}
Now, we only need to evaluate $\tilde{O}(1)$ nonzero summands. In conclusion, in order to sample from and compute $p_j$ in time $\tilde{O}(1)$ per iteration, it suffices to sample from and compute a probability distribution proportional to $\sqrt{y}$ (we remark that our sampling procedure will be exact).

\subsubsection{Sparse combinations}
\label{sssec:sparsecomb}

In this section we describe how to maintain $v_t$, $v_{t + \half}$ which are $\log y_t, \log y_{t + \half}$ up to an additive multiple of the ones vector, via a linear combination of sparsely updated vectors $q_t, r_t, s_t$. The reason for this representation is so that we may update the representation in time $\tilde{O}(1)$ per iteration, and further, for any coordinate $i$, we may compute $\exp([v_t]_i)$ in constant time by simply taking the appropriate linear combination of the vectors. The word ``sparse'' in this section denotes any vector with $\tilde{O}(1)$ nonzero entries. We begin by recalling the notation from Lemma~\ref{lem:updatebound},
\begin{align*}
\delta_t \defeq \frac{1}{\kappa}(b - Ax_t),\; \delta^{(j)}_{t + \half} \defeq \frac{1}{\kappa}\left(b - A\left(x_t + \frac{1}{p_j}\Delta_t^{(j)}\right) \right).
\end{align*}
Now we write the updates to $\hy$, $\py$ in the following form, for $c = \frac{\epsilon}{4\kappa \log n}$:
\begin{equation*}
\hy \propto \exp\left(\log \y - c \log \y - \delta_t \right),\; \py \propto \exp\left(\log y_t - c \log y_{t + \half} - \hd\right).
\end{equation*}
Letting vectors $v_t$, $v_{t + \half}$ satisfy $\y \propto \exp(v_t)$, $\hy \propto \exp(v_{t + \half})$ for all $t$, we have the recursion
\begin{equation*}
v_{t + \half} = (1 - c)v_t - \delta_t,\; v_{t + 1} = v_t - cv_{t + \half} - \hd.
\end{equation*}
Recalling the structure of these updates, we see that we can further decompose $\delta^{(j)}_{t + \half}$ into $\delta_t + \zeta_t$, where $\zeta_t = -\frac{1}{\kappa p_j} A\Delta_t^{(j)}$ is sparse. Next, observing $\norm{b - Ax_t}_\infty \leq 2$ and $\kappa \geq 16n$, we can assume $\norm{\delta_t}_\infty \leq \frac{1}{8n}$. We additionally note that $\delta_t - \delta_{t - 1} = -\frac{1}{\kappa}A(x_t - x_{t - 1})$ is sparse, since $x_t - x_{t - 1}$ is 1-sparse and $A$ has sparse columns. Altogether, this yields
\begin{align*}
v_{t + 1} = (1 - c + c^2) v_t - (1 - c)\delta_t - \zeta_t \\
\Rightarrow v_{t + 1} - v_t = (1 - c + c^2)(v_t - v_{t - 1}) - (1- c)(\delta_t - \delta_{t - 1}) - (\zeta_t - \zeta_{t - 1}) \\
\Rightarrow v_{t + 1} = (2 - c + c^2) v_t - (1 - c+ c^2) v_{t - 1} - (1 - c)(\delta_t - \delta_{t - 1}) - (\zeta_t - \zeta_{t - 1}) \\
\Rightarrow v_{t + 1} = c_1 v_t - c_2 v_{t - 1} - c_3 \mu_t - \nu_t.
\end{align*}
Here, we have defined $c_1 = 2 - c + c^2, c_2 = 1 - c+ c^2, c_3 = 1 - c, \mu_t = \delta_t - \delta_{t - 1}, \nu_t = \zeta_t - \zeta_{t - 1}$. Further, $c_1 \leq 2$ and $c_2 \leq 1$. Similarly, we can compute
\begin{align*}
v_{t + \half} = (1 - c)v_t - \delta_t \Rightarrow v_{t + \half} - v_{t - \half} = (1 - c)(v_t - v_{t - 1}) - (\delta_t - \delta_{t - 1}) \\
\Rightarrow v_{t + \half} = c_3 v_t - c_3 v_{t - 1} + v_{t - \half} - \mu_t.
\end{align*}
In matrix-vector multiplication notation, this update is (where $M$ is clearly full rank)
\begin{align*}
\begin{pmatrix} v_{t + 1} & v_{t + \half} & v_t \end{pmatrix} = \begin{pmatrix} v_t & v_{t - \half} & v_{t - 1} \end{pmatrix} M - \begin{pmatrix} c_3\mu_t + \nu_t & \mu_t & 0\end{pmatrix}, \\ M = \begin{pmatrix} c_1 & c_3 & 1 \\ 0 & 1 & 0 \\ -c_2 & -c_3 & 0\end{pmatrix}.
\end{align*}
Now, suppose we have maintained a representation
\begin{equation*}
\begin{pmatrix} v_t & v_{t - \half} & v_{t - 1} \end{pmatrix} = \begin{pmatrix} q_t & r_t & s_t \end{pmatrix} M^t.
\end{equation*}
We then require the update
\begin{equation*}
\begin{pmatrix} q_{t + 1} & r_{t + 1} & s_{t + 1} \end{pmatrix} M^{t + 1} = \begin{pmatrix} q_t & r_t & s_t \end{pmatrix} M^{t + 1} - (\begin{pmatrix} c_3\mu_t + \nu_t & \mu_t & 0\end{pmatrix}M^{-t - 1}) M^{t + 1}.
\end{equation*}
We can maintain $M^{-t - 1}$ in closed form by simply performing a single matrix multiplication of $3 \times 3$ matrices each iteration, so the updates to $q_{t + 1}, r_{t + 1}$ and $s_{t + 1}$ are sparse:
\begin{equation*}
\begin{pmatrix} q_{t + 1} & r_{t + 1} & s_{t + 1} \end{pmatrix} = \begin{pmatrix} q_t & r_t & s_t \end{pmatrix} - \begin{pmatrix} c_3\mu_t + \nu_t & \mu_t & 0\end{pmatrix}M^{-t - 1}.
\end{equation*}
\subsubsection{Maintaining the sum of exponentials}

The previous section states that we can maintain a representation of $v_t$ in $\tilde{O}(1)$ time per iteration, such that we can query for any $i$, the value $\exp([v_t]_i)$ in constant time (respectively, $\exp([v_{t + \half}]_i)$). Consequently, in order to support $\yor.\texttt{Coord}$ (respectively, $\yor.\texttt{Coord-Half}$), we need to be able to approximate
\begin{equation}\label{eq:task}\sum_{i \in [n]} \exp([v_t]_i)\end{equation}
multiplicatively by $1 + \frac{1}{n^{100}}$. In this section we will discuss how to do so over $n$ iterations in time $\tilde{O}(n)$. We then discuss how to sample from this distribution, and modify this maintenance to also support approximate coordinate queries from $y_{t + \half}$. We will not formally discuss how to extend this analysis to query and sample from a distribution proportional to $\sqrt{y_t}$, as required by \ref{ssec:samplereduce}, as it is an immediate generalization; we simply also implement $\yor$ with the vectors $\half v_t$, which clearly suffices. For the scope of this section, define the constant
$$c = \frac{\epsilon}{4\kappa\log n},\; \kappa > 16n.$$
The implementation problem is: for every iteration $t \in [n]$, we are given vectors $\delta_t, \zeta_t$, such that
\begin{itemize}
	\item $\norm{\zeta_t}_\infty \leq \frac{1}{8}$, and $\zeta_t$ is sparse.
	\item $\norm{\delta_t}_\infty \leq \frac{1}{8n}$.
	\item Vectors $v_t$ are defined recursively via $v_{t + 1} \defeq (1 - c) v_t + \delta_t + \zeta_t$.
	\item We are able to maintain a representation of $v_t$ as a linear combination $\alpha_t q_t + \beta_t r_t + \gamma_t s_t$, for sparsely changing $q_t, r_t, s_t$, and scalars $\alpha_t, \beta_t, \gamma_t$.
\end{itemize}
These bounds follow from the analysis in Lemma~\ref{lem:updatebound}. We also require the following fact on the effect of a certain ``squishing'' operation, which states that we may take any coordinate of $v_t$ which is significantly smaller than another, and raise it within a certain range.
\begin{lemma}
	\label{lem:projecteffect}
	Let $v \in \R^n$, and let $y \in \Delta^n$ be such that $y \propto \exp(v)$. Consider the following operation: let $i^*$, $i'$ be coordinates of $v$ such that $v_{i'} < v_{i^*} - \frac{16\log n}{\epsilon}$, and set $\hat{v} = v$ in every coordinate, except $\hat{v}_{i'} \leftarrow v_{i^*} - \frac{16\log n}{\epsilon}$. Then, for $\hat{y} \propto \exp(\hat{v})$, assuming $\epsilon < \frac{1}{7}$,
	\begin{equation*}
	V_{y}(\tilde{y}) - V_{\hat{y}}(\tilde{y}) > -n^{-100}.
	\end{equation*}
\end{lemma}
\begin{proof}
	We explicitly compute
	\begin{equation*}
	V_{y}(\tilde{y}) - V_{\hat{y}}(\tilde{y}) = \sum_{i} \tilde{y}_i\log \frac{\hat{y}_i}{y_i}.
	\end{equation*}
	Note that the only possible $i$ such that $\frac{\hat{y}_i}{y_i} \geq 1$ is $i = i'$. Furthermore, for every other coordinate $i$,
	\begin{equation*}
	\frac{y_{i}}{\hat{y}_{i}} = \frac{\frac{\exp(v_{i})}{\norm{\exp(v)}_1}}{\frac{\exp(\tilde{v}_{i})}{\norm{\exp(\tilde{v})}_1}} =  \frac{\norm{\exp(\tilde{v})}_1}{\norm{\exp(v)}_1} < \frac{1 + n^{ - \frac{16}{\epsilon}}}{1}.
	\end{equation*}
	Here, we used that $\exp(\hat{v}_{i'})$ can be at most $\exp(-\frac{16\log n}{\epsilon})$ of the sum, due to the contribution of the $v_{i^*}$ term, and all coordinates $i \neq i'$ have $\hat{v}_i = v_i$. Finally,
	\begin{equation*}
	\sum_{i} \tilde{y}_i \log \frac{\hat{y}_i}{y_i} \ge  -\sum_{i \neq i'} \tilde{y}_i \log (1 + n^{-100}) \geq -n^{-100}.
	\end{equation*}
\end{proof}

We assume that in the first iteration, we have spent $O(n)$ time computing $v_0$ explicitly, using our sparse representation, and squishing so its coordinates lie in the range $[0, \frac{16\log n}{\eps}]$.

\paragraph{The case $\zeta_t = 0$.}

We first handle the case when all of the $\zeta_t = 0$. At iteration $0$, suppose we have spent $O(n)$ time to compute $i^* = \argmax_{i} [v_0]_i$. Also, recall we guaranteed $[v_0]_{i^*} - [v_0]_i \leq \frac{16\log n}{\epsilon}$. We use the following fact:

\begin{fact}[Taylor expansion of exponential]\label{fact:tayexp}
	Let $|x| \leq \half$. Then, letting $\textup{Tay}_d(x)$ be the degree $d$ Taylor approximation of the exponential, we can bound $|\textup{Tay}_d(x) - \exp(x)| \leq \frac{1}{2^d}$.
\end{fact}

To approximate \eqref{eq:task} on iteration $t$, we will maintain a scalar $\sigma_t$ with the guarantee
\begin{equation}\label{eq:invariant}\norm{v_0 - \sigma_t\1 - v_t}_\infty \leq \half.\end{equation}
We will explicitly compute $[v_t]_{i^*}$ each iteration $t$, and set
\begin{equation}\label{eq:sigupdate}\sigma_{t + 1} = \sigma_t + c[v_t]_{i^*}.\end{equation}
First of all, we show the invariant \eqref{eq:invariant}.
\begin{lemma}Every iteration $t \le n$, and for all $i$, $|[v_t]_i - [v_t]_{i^*}| \leq \frac{17\log n}{\epsilon}$.\end{lemma}
\begin{proof}
	We claim that the range of the coordinates of $v_t$ is never larger than $\frac{17\log n}{\epsilon}$: certainly, this implies the conclusion. To show this, we inductively claim that the range satisfies
	$$\max_i [v_t]_i - \min_j [v_t]_j \leq \frac{16\log n}{\epsilon} + \frac{t}{4n}.$$
	Taking $t \leq n$ yields the result. Clearly for $t = 0$ this is true; now, for $t + 1$, recall $v_{t + 1} = (1 - c)v_t + \delta_t$. Let $i = \argmax_i [v_t]_i$, $j = \argmin_j [v_t]_j$. Then,
	$$[v_{t + 1}]_i - [v_{t + 1}]_j = (1 - c)([v_t]_i - [v_t]_j) + ([\delta_t]_i - [\delta_t]_j) \leq \frac{16\log n}{\epsilon} + \frac{t}{4n} + \frac{1}{4n}.$$
	Here we used the inductive guarantee and the range of $\delta_t$ (we may clearly assume $\log n/\eps > 1$).
\end{proof}

\begin{lemma}Every iteration $t \le n$, \eqref{eq:invariant}
	holds.\end{lemma}
\begin{proof}
	For some particular $i$, we show it holds; this implies the $\ell_\infty$ guarantee. Note that
	\begin{equation*}
	|[v_0]_i - \sigma_{t + 1}- [v_{t + 1}]_i| \leq |[v_0]_i - \sigma_{t} - [v_t]_i| + |([v_{t}]_i - [v_{t + 1}]_i) - c[v_t]_{i^*}|.
	\end{equation*}
	Here we used triangle inequality and the definitions of $\sigma, \bar{v}$. Now, we have
	$$|([v_{t}]_i - [v_{t + 1}]_i) - c[v_t]_{i^*}| \leq |c[v_t]_i - c[v_t]_{i^*}| + |[\delta_t]_i| \leq \frac{3}{8n} + \frac{1}{8n} \leq \frac{1}{2n}.$$
	Thus, inductively we have that
	$$|[v_0]_i - \sigma_{t} - [v_t]_i| \leq \frac{t}{2n}.$$
	Using $t \leq n$ yields the result.
\end{proof}
Finally, we describe how to compute an accurate approximation \eqref{eq:task} by Taylor expansion in $\tilde{O}(1)$ time per iteration. We approximate, for some $d = O(\log n)$,
\begin{equation}
\label{eq:tayapprox}
\begin{aligned}
\sum_{i \in [n]} \exp([v_t]_i) &= \sum_{i \in [n]} \exp\left([v_0]_i - \sigma_t\right) \exp\left([v_t]_i - ([v_0]_i - \sigma_t)\right)\\
&\approx \exp(- \sigma_t)\sum_{i \in [n]} \exp\left([v_0]_i\right) \text{Tay}_d\left(\alpha_t [q_t]_i + \beta_t[r_t]_i + \gamma_t[s_t]_i - ([v_0]_i - \sigma_t)\right).
\end{aligned}
\end{equation}
We now group by the degree of the Taylor expansion, $0 \leq k \leq d$, and each quintuple $0 \leq d_1 + d_2 + d_3 + d_4 + d_5 = k \leq d$:
\begin{align*}
\exp(-\sigma_t)\sum_{i \in [n]} \exp([\bar{v}_t]_i)\sum_{0 \leq k \leq d} \frac{(\alpha_t [q_t]_i + \beta_t[r_t]_i + \gamma_t[s_t]_i- ([v_0]_i - \sigma_t))^k}{k!} \\
= \exp(-\sigma_t)\sum_{i \in [n]} \exp([v_0]_i) \sum_{d_1, d_2, d_3, d_4, d_5} \frac{\binom{k}{d_1, d_2, d_3, d_4, d_5}}{k!} (\alpha_t)^{d_1} (\beta_t)^{d_2} (\gamma_t)^{d_3} (-1)^{d_4} [q_t]_i^{d_1} [r_t]_i^{d_2} [s_t]_i^{d_3} [v_0]_i^{d_4} [\sigma_t]^{d_5}\\
= \exp(-\sigma_t) \sum_{d_1, d_2, d_3, d_4, d_5} \frac{\binom{k}{d_1, d_2, d_3, d_4, d_5}}{k!} (\alpha_t)^{d_1} (\beta_t)^{d_2} (\gamma_t)^{d_3} (-1)^{d_4} \sum_{i \in [n]} \exp([v_0]_i) [q_t]_i^{d_1} [r_t]_i^{d_2} [s_t]_i^{d_3} [v_0]_i^{d_4} [\sigma_t]^{d_5}.
\end{align*}
Consider the complexity of computing the last expression. There are at most $(d + 1)^5 = O(d^5)$ quintuplets $d_1, d_2, d_3, d_4, d_5$ with $0 \leq d_1 + d_2 + d_3 + d_4 + d_5 \leq d$. For each quintuplet, we maintain
\begin{equation*}
\sum_{i \in [n]} \exp([v_0]_i) [q_t]_i^{d_1} [r_t]_i^{d_2} [s_t]_i^{d_3} [v_0]_i^{d_4} [\sigma_t]^{d_5}.
\end{equation*}
Because each of $q_t, r_t, s_t$, are sparsely changing, we can spend $\tilde{O}(d^5)$ time updating the relevant terms in each of these summations. Furthermore, $[\sigma_t]^{d_5}$ is simply a scalar so we can rescale its contribution to the entire sum in constant time. Now, in order to compute the overall sum, we can spend constant time computing each coefficient 
\[\frac{\binom{k}{d_1, d_2, d_3, d_4, d_5}}{k!} (\alpha_t)^{d_1} (\beta_t)^{d_2} (\gamma_t)^{d_3} (-1)^{d_4};\]
this takes $\tilde{O}(d^5)$ time altogether. Lastly, updating $\exp(-\sigma_t)$, the scaling of the entire sum, takes constant time, and computing the overall sum thus takes $\tilde{O}(d^5)$.

 Finally, we must argue that performing this procedure for $d = O(\log n)$ suffices for a multiplicative guarantee of $1 + \frac{1}{n^{O(1)}}$. Comparing the approximation in \eqref{eq:tayapprox} to the required \eqref{eq:task}, the only difference is each of the approximations
$$\exp([v_t]_i - ([v_0]_i - \sigma_t)) \approx \text{Tay}_d([v_t]_i - ([v_0]_i - \sigma_t)).$$
Because the left hand side is bounded between $\exp(\pm \half)$, an additive approximation is (up to constants) a multiplicative approximation as well. Further, Fact \ref{fact:tayexp} implies that $d = O(\log n)$ suffices for this quality of approximation, as desired.

\paragraph{Binomial heap data structures for $\zeta_t$.}

In this section, we reduce the general case to the case where $\zeta_t = 0$ via a binomial heap data structure, a fairly general reduction. We note that the analysis in the previous section also clearly holds when the number of iterations is less than $n$, and when there are less than $n$ coordinates. The main idea of the reduction is that we will maintain data structures for sets $\{S_k\}$ for $0 \leq k \leq \lceil\log n\rceil$, such that a $S_k$ either contains no elements, or between $2^{k - 1} + 1$ and $2^k$ elements. In particular, we maintain on every iteration
\begin{itemize}
	\item A hashmap which, for each $i \in [n]$, tracks which $S_k$ it belongs to.
	\item For each $S_k$,
	\begin{itemize}
		\item The cardinality of $S_k$.
		\item $\sum_{i \in [n]} \exp([v_0]_i) [q_t]_i^{d_1} [r_t]_i^{d_2} [s_t]_i^{d_3} [v_0]_i^{d_4} [\sigma_t]^{d_5}$, for each quintuplet $0 \leq d_1 + d_2 + d_3 + d_4 + d_5 \leq d$.
	\end{itemize}
\end{itemize}
The main difficulty is maintaining the invariant that there is at most one set of each rank (we call $k$ the ``rank'' of a nonempty $S_k$). To this end, if there are two sets $S_k, S_k'$ both with cardinality between $2^{k - 1} + 1$ and $2^k$ elements, e.g. of rank $k$, we allow the operation $\texttt{Merge}(S_k, S_k')$ which creates a new $S_{k + 1}$ of rank $k + 1$, containing all of the coordinates associated with either $S_k$ or $S_k'$. 

Whenever we perform a merge, we explicitly compute all coordinates involved in the merge, designate the largest as $i^*$ for the updates to $\sigma_t$ for that particular set, and squish if necessary to guarantee that the range of the set is at most $\frac{16\log n}{\epsilon}$; clearly, given our sparse representation $q_t, r_t, s_t, \alpha_t, \beta_t, \gamma_t$, we can appropriately modify a coordinate of say $q_t$ to handle the squishing. We also instantiate all relevant quintuplet sums for our particular set. Furthermore, if $|S_k'| + |S_k| \leq 2^{k + 1}$, $\texttt{Merge}$ will also spend $\tilde{O}(2^{k + 1} - |S_k'| - |S_k|)$ time to create ``initialization credits'', so that the sum of the initialization credits and the size of $S_{k + 1}$ is always exactly $2^{k + 1}$; these credits will be useful for our amoritized analysis. It takes time $\tilde{O}(2^{k + 1})$ to update the hashmap, reinstantiate all the relevant quintuplet sums, and create credits, for $S_{k + 1}$. We note we may need to recursively call $\texttt{Merge}$ if there was already a set of rank $k + 1$.

At the start of the $n$ iterations, we initialize a single set of rank $\lceil \log n\rceil$, and put all of the coordinates in this set (and pay any additional cost required for initialization credits), in time $\tilde{O}(n)$. Each iteration $t + 1$ will proceed in three stages. In the first stage, we compute the approximation \eqref{eq:tayapprox} to the sum of exponentials as in the previous section, ignoring the effect of $\zeta_t$. The complexity of this stage is at most $O(\log n)$ times its complexity in the previous section, because we may need to perform updates for each $S_k$; thus, it can be implemented in amoritized time $\tilde{O}(1)$.

In the second stage, for each coordinate in the support of $\zeta_t$, we delete it from its corresponding $S_k$ and instantiate a new set of rank 0, now explicitly factoring in the effect of $\zeta_t$. Furthermore, if this causes its corresponding $S_k$ to become rank $k - 1$, e.g. if before the deletion $S_k$ had $2^{k - 1} + 1$ elements, and there was already a set of rank $k - 1$, we will call the $\texttt{Merge}$ operation on the two sets of rank $k - 1$. The amoritized cost of the second stage is $\tilde{O}(1)$. To see this, every time we create a new set of rank 0, we spend $\tilde{O}(1)$ time to both initialize the rank 0 set, and pay for $\tilde{O}(1)$ ``deletion credits''. Now, whenever we must use the $\texttt{Merge}$ operation to create a new set of rank $k$, we can pay for the operation (which costs $\tilde{O}(2^k)$, both to merge and pay for new initialization credits) by using existing credits: between when $S_k$ was initialized and when it needed to be reinitialized due to becoming rank $k - 1$, the sum of its initialization credits and the deletion credits created by removing elements is at least $\tilde{O}(2^k)$. We call any such merges Type-1 merges.

In the third stage, we recursively call $\texttt{Merge}$, starting from the rank 0 sets, in order to maintain the invariant that there is at most one set of any given rank. We call any such merges Type-2 merges. We claim the amoritized cost of all Type-2 merges over all $n$ iterations is $\tilde{O}(n)$. Consider the number of times a rank $k$ set can be created through Type-2 merges: we claim it is upper bounded by $\tilde{O}\left(\frac{n}{2^k}\right)$. If this is true, overall the complexity of the third stage is at most
\begin{equation*}
\tilde{O}\left(\sum_{k = 0}^{\lceil\log n \rceil} 2^k \frac{n}{2^k}\right) = \tilde{O}(n).
\end{equation*}
The number of deletions due to the $\zeta_t$ throughout $n$ iterations is at most $\tilde{O}(n)$. Thus, it suffices to prove that between creations of rank $k$ sets due to Type-2 merges, there must have been at least $2^{k - 1}$ deletions. To see this, for each rank $k$, maintain a potential $\Phi_k$ for the sum of the cardinalities of all rank $l$ sets for $l < k$. Each deletion increases $\Phi_k$ by at most 1. Each Type-1 merge does not increase $\Phi_k$, because it can only cause coordinates to belong to sets which increase in rank. In order for a Type-2 merge to be used to create a rank $k$ set, $\Phi_k$ must have been at least $2^{k - 1} + 2$; after the merge, it is 0, because in its creation, all rank $l$ sets for $l < k$ must have been merged. Thus, for the potential to become large enough to require a merge again, there must have been at least $2^{k - 1}$ deletions, as desired.

Finally, we remark that the analysis of each constituent data structure, i.e. the case when $\zeta_t = 0$ for the supported coordinates, remains correct under deletions. In particular, \eqref{eq:sigupdate} may still use the original value of $[v_t]_{i^*}$ in its recursion, even if the coordinate $i^*$ is deleted; it is easy to see that by the original boundedness of the range of supported coordinates, the analysis still holds.

\paragraph{Maintaining $y_{t + \half}$.}

In order to compute coordinates of $y_{t + \half}$, we discuss approximating the sum
$$\sum_{i \in [n]} \exp\left(\left[v_{t + \half}\right]_i\right).$$
It is easy to see that because $v_{t + \half}$ and $(1 - c)v_t$ never vary by more than a small additive constant $\frac{1}{8n}$, and furthermore we also maintain a sparsely updated representation of $v_{t + \half}$ in terms of $q_t, r_t, s_t$, we may suitably modify the approximation \eqref{eq:tayapprox} to approximate this sum. In particular, we may compute an appropriate scaling $\sigma_{t + \half}$ by estimating all coordinates of $cv_t$ by scaling some particular coordinate, and estimate the coefficients of the quintuplet sums in terms of the coefficients in the linear combination. The complexity of this computation in each step is at most $O(d^5)$ in each step, which never asymptotically dominates.

\paragraph{Sampling from the sum of exponentials.}

Here, we discuss how to sample from the sum of the exponentials. We use the following fact about rejection sampling.

\begin{fact}[Rejection sampling]
	Suppose $P$ and $Q$ are probability distributions over $[n]$, and $\frac{P[i]}{Q[i]} \in [\half, 2]$ for all $i \in [n]$. Further, suppose we may sample from $P$ in time $\tilde{O}(1)$. The following strategy samples exactly from $Q$ in expected $O(1)$ time: sample a coordinate of $i$ according to $P$, and accept with probability $\frac{Q[i]}{2P[i]}$; repeat until acceptance.
\end{fact}

In our setting, in each iteration $P$ is the distribution over coordinates $i \in S_k$ proportional to $\exp([v_0]_i - \sigma_t)$, where we overload the definitions of $v_0$, $\sigma_t$ to refer to the point the set $S_k$ uses to approximate the Taylor expansion. We can sample from this distribution $P$ by maintaining for each of the $\tilde{O}(1)$ sets $S_k$, $\sum_{i \in S_k} \exp([v_0]_i - \sigma_t)$, by initializing it with the sum when $\sigma_t = 0$, and then appropriately scaling the sums each iteration. Further, we may initialize each set $S_k$ with a binary tree data structure for sampling proportional to $[v_0]_i$ for each $i \in S_k$, because the uniform scaling $\exp(-\sigma_t)$ does not affect this distribution. In conclusion, we sample from $P$ by first sampling a set $S_k$ proportional to its weight given by $P$, and then sampling a coordinate in the set appropriately. 

We then rejection sample from $P$ with respect to $Q$, the true distribution. By the invariant \eqref{eq:invariant}, this rejection sampling scheme meets the requirements to succeed in expected time $\tilde{O}(1)$, which yields the conclusion.

\subsubsection{Cleaning up: effects of approximate sums and squishing}
\label{ssec:approxsum}

We first prove Lemma~\ref{lem:stable}, using additional structure afforded by our data structure implementation.

\begin{proof}[Proof of Lemma~\ref{lem:stable}]
Recall the notation and bounds from Lemma~\ref{lem:updatebound},
\begin{align*}
\delta_t \defeq \frac{1}{\kappa}(b - Ax_t),\; \delta^{(j)}_{t + \half} \defeq \frac{1}{\kappa}\left(b - A\left(x_t + \frac{1}{p_j}\Delta_t^{(j)}\right) \right),\\ \norm{\delta_t}_\infty,\; \norm{\delta_{t + \half}^{(j)}}_\infty \leq \frac{1}{4}.
\end{align*}
We write the updates to $\hy$, $\py$ in the form, for $c = \frac{\epsilon}{4\kappa \log n}$:
\begin{equation*}
\hy \propto \exp\left(\log \y - c \log \y - \delta_t \right),\; \py \propto \exp\left(\log y_t - c \log y_{t + \half} - \hd\right).
\end{equation*}
Letting vectors $v_t$, $v_{t + \half}$ satisfy $\y \propto \exp(v_t)$, $\hy \propto \exp(v_{t + \half})$ for all $t$, the goal of this lemma is to show that $\norm{v_{t + 1} - v_t}_\infty$, $\norm{v_{t + \half} - v_t}_\infty$ are both bounded by $1$. Indeed, we have the recursion
\begin{equation*}
v_{t + \half} = (1 - c)v_t - \delta_t,\; v_{t + 1} = v_t - cv_{t + \half} - \hd.
\end{equation*}
Based on the bounds on $\delta_t$ and $\hd$, it suffices to show that $\norm{cv_t}_\infty$, $\norm{cv_{t + \half}}_\infty \le \frac{3}{4}$. By the squishing operations performed by the data structure, at the beginning of $n$ iterations (when the data structure is restarted), the range of $v_t$ is contained in $[0, 16\log n/\eps]$. 

Over the course of $n$ iterations, this fact is preserved for each particular data structure supporting a set of coordinates. Moreover, we recall that we used squishing whenever we initialize a new data structure in the binomial heap to maintain the fact that the additive range over all coordinates is $O(\log n/\eps)$. The final issue which may come up is the additive drift caused by the vectors $\delta_t$ or $\hd$; however, over the course of $n$ iterations, this can only shift the largest coordinate of $v_t$ by $n/4$. Altogether, it is clear we may assume $\norm{v_t}_\infty < \frac{n}{4} + \frac{33\log n}{\eps} \ll \frac{3}{4c}$; the conclusion follows. Similarly, we inductively have $\norm{v_{t + \half}}_\infty \ll \frac{3}{4c}$ by bounding its difference to $v_t$.
\end{proof}

We now consider the effect of only approximately maintaining the sums of exponentials in our algorithm, and applying squishing. In particular, the inequality in Lemma~\ref{lem:threepoint} only holds up to an additive constant. The additive error comes into play in two ways: the first-order optimality condition only holds up to the discrepancy between $y_t$ and $\breve{y}_t$, and each time we apply squishing affects the value of $\E\left[V_{z_{t + 1}}(\tz)\right]$. Regarding the former, all problem parameters and the number of phases of our algorithm are all bounded by a small polynomial in $n$, so the guarantees of $\yor.\texttt{Coord}$ mean that the cumulative error does not amount to more than $n^{-90} \ll \epsilon$ (we assume $\eps > n^{-3}$, else an interior point method achieves our stated runtime). Similarly, regarding the latter, Lemma~\ref{lem:projecteffect} implies that even if we squish $O(n)$ coordinates each iteration, the cumulative error in the Bregman divergence does not amount to more than $n^{-90}$. 
\subsection*{Acknowledgments}
This work was supported by NSF Graduate Fellowship DGE-1656518. We would like to thank Kent Quanrud for pointing out an error in the prior version of this manuscript. Some ideas in writing this new version were motivated by developments in the independent project \cite{CarmonJST19}; we would like to thank our collaborators Yair Carmon and Yujia Jin for helpful discussions.
\newpage
\bibliographystyle{alpha}
\bibliography{accel_linf}
\newpage
\begin{appendix}
\section{Missing proofs from \Cref{sec:intro} and \Cref{sec:overview}}
\label{appendix:a}

\subsection{Folklore bound on size of $\ell_\infty$-strongly-convex functions}
\label{sec:folklore}

In this section, we prove the following claim which occurs in the literature, but does not seem to usually be formally shown:

\begin{lemma}
Suppose $\psi$ is $1$-strongly convex with respect to the $\ell_\infty$ norm on $[-1, 1]^n$. Then,
\begin{equation*}
\max_{x \in [-1, 1]^n} \psi(x) - \min_{x \in [-1, 1]^n} \psi(x) \geq \frac{n}{2}
\end{equation*}
Furthermore, this lower bound is tight, i.e. there is a 1-strongly convex function in the $\ell_\infty$ norm for which equality holds.
\end{lemma}

\begin{proof}
We will prove this by iteratively constructing a set of points $x_0, x_1, \ldots x_n \in [-1, 1]^n$ such that for all $i$ with $0 \leq i \leq n - 1$, we have
\begin{equation*}
\psi(x_i) \leq \psi(x_{i + 1}) - \frac{1}{2}
\end{equation*}
and consequently,
\begin{equation*}
\psi(x_0) \leq \psi(x_n) - \frac{n}{2}
\end{equation*}
Let $e_i$ be the $i^{th}$ standard basis vector, namely the $n$-dimensional vector which is 1 in the $i^{th}$ coordinate and 0 elsewhere. Let $x_0 = (0, 0, \ldots 0)$, the $n$-dimensional point which is 0 in every coordinate. Let $x^+_1 = x_0 + e_1$ and let $x^-_1 = x_0 - e_1$, such that $x_0 = \frac{1}{2}x^+_1 + \frac{1}{2}x^-_1$. By strong convexity,
\begin{equation*}
\psi(x_0) \leq \frac{1}{2} \psi(x^+_1) + \frac{1}{2} \psi(x^-_1) - \frac{1}{8}\norm{x^+_1 - x^-_1}_\infty^2 = \frac{1}{2} \psi(x^+_1) + \frac{1}{2} \psi(x^-_1) - \frac{1}{2}
\end{equation*}
Consequently, it must be the case that at least one of
\begin{align*}
\psi(x_0) &\leq \psi(x^+_1) - \frac{1}{2} \\
\psi(x_0) &\leq \psi(x^-_1) - \frac{1}{2}
\end{align*}
holds. Let $x_1$ be the point $x^+_1$ or $x^-_1$ for which this holds. \\

More generally, suppose we have constructed $x_0, x_1, \ldots x_i$ in this fashion, such that $x_i$ is 0 in the coordinates $i + 1, i + 2, \ldots n$. Then, let $x_{i + 1}^+ = x_i + e_{i + 1}$ and let $x_{i + 1}^- = x_i - e_{i + 1}$, such that $x_i = \frac{1}{2} x_{i + 1}^+ + \frac{1}{2} x_{i + 1}^-$. Again by strong convexity, we have that at least one of 
\begin{align*}
\psi(x_i) &\leq \psi(x^+_{i + 1}) - \frac{1}{2} \\
\psi(x_i) &\leq \psi(x^-_{i + 1}) - \frac{1}{2}
\end{align*}
holds, and therefore we can pick one of the points $x^+_{i + 1}, x^-_{i + 1}$ to be the point $x_{i + 1}$. We can clearly iteratively construct a point $x_n$ in this fashion, proving the claim. \\

To show that the lower bound is tight, consider $\psi(x) = \frac{1}{2} \norm{x}_2^2$. Clearly this function has range $\frac{n}{2}$ over $[-1, 1]^n$. Furthermore, for all $x \in [-1, 1]^n$, and arbitrary vector $z$, we have
\begin{equation*}
z^\top \nabla^2 \psi(x) z = z^\top I z = \norm{z}_2^2 \geq \norm{z}_\infty^2
\end{equation*}
where this second-order condition is well-known to be equivalent to 1-strong convexity, for twice-differentiable functions.
\end{proof}

\subsection{Reduction from general box-constrained $\ell_\infty$ regression to Definition~\ref{def:boxedlinfreg}}
\label{appendix:reductionbox}

In this section, we describe a general reduction from unconstrained $\ell_\infty$ regression and more arbitrary box constraints to the setting where the domain of the argument is $[-1, 1]^m$, proving Corollary~\ref{corr:generalbox}. Consider first the problem of solving the generalized box-constrained regression problem
\begin{equation}\label{eq:genbox}\min_{x \in [-r, r]^m} \norm{Ax - b}_\infty,\end{equation}
for some $r > 0$. By performing the change of variables $\tx = x/r$, $\tilde{b} = b/r$, it suffices to find an $\eps/r$-approximate minimizer to
\begin{equation}\label{eq:genboxscaled}\min_{\tx \in [-1, 1]^m}\norm{A\tx - \tilde{b}}_\infty, \end{equation}
which under the change of variables $x \gets r\tx$ recovers an $\eps$-approximate minimizer to the original problem. To see this, let $x^*$ be the minimizer to \eqref{eq:genbox}; it is clear under a simple rescaling and linearity of norms that $\tx^*\defeq x^*/r$ is the minimizer to \eqref{eq:genboxscaled}. Next, let $\tx$ be any point in $[-1, 1]^m$ with
\[\norm{A\tilde{x} - \tilde{b}}_\infty - \norm{A\tx^* - \tilde{b}}_\infty \le \frac{\eps}{r}.\]
By linearity of norms, we see that $x = r\tx$ has $x \in [-r, r]^m$ and
\[\norm{Ax - b}_\infty - \norm{Ax^* - b}_\infty \le \eps,\]
i.e. $x$ is an $\eps$-approximate minimizer to \eqref{eq:genbox}. To bound the complexity of solving \eqref{eq:genboxscaled} to $\eps/r$ additive accuracy, it suffices to invoke Theorem~\ref{thm:accelboxlinfreg}.

Next, to deal with the unconstrained case with the promise $\norm{x_0 - x^*}_\infty \le r$, it suffices to perform a change of variables $b' \gets b - Ax_0$, $x' \gets x - x_0$, and solve the problem
\[\min_{x' \in [-r, r]^m} \norm{Ax' - b'}_\infty = \min_{x' \in [-r, r]^m} \norm{A(x - x_0) - (b - Ax_0)}_\infty = \min_{x \in \R^m} \norm{Ax - b}_\infty.\]
To see the last inequality, we use the guarantee that $\norm{x_0 - x^*}_\infty \le r$, i.e. $x^* - x_0$ is a valid point $x'$. We then can invoke the general box-constrained case with radius $r$.

Finally, we remark that similar additive shifts and rescalings allow us to handle the more general box constraint $\prod_{j \in [m]} [\ell_j, r_j]$ with appropriate (weighted) dependences on the quantities $r_j - \ell_j$.

\subsection{Convergence rates of first-order methods}

In this section, we give guarantees for the convergence rates of the classical unaccelerated first-order methods of gradient descent in general norms and coordinate descent. %

\subsubsection{Gradient descent in general norms}
\label{sec:gdgeneralnorm}

We briefly review the basic guarantees of gradient descent applied to a convex function $f$ which is $L$-smooth in an arbitrary norm $\| \cdot \|$. The general framework of gradient descent initializes at some point $x^0$ and iteratively maximizes the primal progress using the upper bound guaranteed by the smoothness. In particular, we perform the following update:
\begin{equation*}
x^{k + 1} \leftarrow \textrm{argmin}_y \Big\{ f(x^k) + \nabla f(x^k)^\top (y-x^k) + \frac{L}{2} \|y - x^k\|^2 \Big\}
\end{equation*}
The $O(\frac{1}{T})$ convergence rate of gradient descent is well-known in the literature. We state the convergence guarantee here.

\begin{lemma}
\label{lem:gdconvergence}
Let $x^T$ be the result of running gradient descent for $T$ iterations. Then for the global minimizer $x^*$, we have $f(x^T) - f(x^*) \leq \frac{2LR^2}{T}$, where $R = \textup{max}_{y : f(y) \leq f(x^0)} \|y - x^*\|$.
\end{lemma}

\subsubsection{Coordinate descent}

Next, we briefly review the basic guarantees of randomized coordinate descent when applied to a convex function $f$ which is $L_j$-smooth in the $j^{th}$ coordinate. Here, we analyze the convergence rate of the simple unaccelerated variant of coordinate descent where coordinate $j$ is sampled with probability $\frac{L_j}{S}$, where $S \defeq \sum_j L_j$. In particular, we perform the following update after sampling a coordinate $j$:
\begin{equation*}
x^{k + 1} \leftarrow \textrm{argmin}_y \Big\{ f(x^k) + \nabla_j f(x^k)^\top (y-x^k) + \frac{L_j}{2} |y_j - x^k_j|^2 \Big\} = x^k - \frac{1}{L_j} \nabla_j f(x^k)
\end{equation*}
Here, we give the convergence rate of this simple coordinate descent algorithm.

\begin{lemma}
\label{lem:cdconvergence}
Let $x^T$ be the result of running gradient descent for $T$ iterations. Then for the global minimizer $x^*$, we have $f(x^T) - f(x^*) \leq \frac{2SR^2}{T}$, where $R = \textup{max}_{y : f(y) \leq f(x^0)} \|y - x^*\|_2$.
\end{lemma}

We remark that for any randomized iterative method for minimizing a convex function which converges in expectation, it is easy to use Markov's inequality to bound the convergence with constant probability. For example, if an algorithm terminates with a $\epsilon$-approximate minimizer on expectation, with probability at least $\frac{1}{2}$ it terminates with a $2\epsilon$-approximate minimizer. Thus, if one desires a high probability result for the approximate minimization, the runtime only incurs a logarithmic multiplicative loss in the failure probability.

\subsection{Proof of \Cref{lem:gdconvergence}}

First we give an intermediate progress bound which will be useful in the final proof.

\begin{lemma}
\label{lem:gdprogress}
	$f(x^{k}) - f(x^{k + 1}) \geq \frac{1}{2L} \|\nabla f(x^k)\|_*^2$
\end{lemma}

\begin{proof}
We will prove that $\textrm{min}_y \Big\{ \nabla f(x)^\top (y - x) + \frac{L}{2} \|y - x\|^2 \Big\} \leq -\frac{1}{2L} \|\nabla f(x)\|_*^2$; clearly this yields the desired claim. Let $z$ be such that $\|z\| = 1$ and $z^\top \nabla f(x) = \|\nabla f(x)\|_*$, by the definition of dual norm; let $y = x - \frac{\|\nabla f(x)\|_*}{L} z$. Then, 
\begin{equation*}
\nabla f(x)^\top (y - x) + \frac{L}{2} \|y - x\|^2 = -\left(\frac{\|\nabla f(x)\|_*}{L}\right) z^\top \nabla f(x) + \frac{L}{2} \frac{\|\nabla f(x)\|_*^2}{L^2} \|z\|^2 = -\frac{1}{2L} \|\nabla f(x)\|_*^2
\end{equation*}
Thus, the minimizer of the upper bound yields the desired progress result.
\end{proof}

Next, we prove \Cref{lem:gdconvergence}.

\begin{proof}
Let $\epsilon_k \defeq f(x^k) - f(x^*)$. Note that by convexity and Cauchy-Schwarz, we have 
\begin{equation*}
f(x^k) - f(x^*) \leq (\nabla f(x^k))^\top (f(x^k) - f(x^*)) \leq \|\nabla f(x^k)\|_* \|x^k - x^*\|
\end{equation*}
Thus, we have the two equations $\epsilon_k - \epsilon_{k + 1} \geq \frac{1}{2L} \|\nabla f(x^k)\|_*^2$ and $\epsilon_k \leq R \|\nabla f(x^k)\|_*$. Combining the two, it's easy to see that
\begin{equation*}
\epsilon_k^2 \leq 2LR^2(\epsilon_k - \epsilon_{k + 1}) \leftrightarrow \Big( \frac{1}{\epsilon_{k+1}} - \frac{1}{\epsilon_k} \Big) \geq \frac{\epsilon_k}{2LR^2 \epsilon_{k + 1}} \geq \frac{1}{2LR^2}
\end{equation*}
Thus, telescoping we have $\frac{1}{\epsilon_T} \geq \frac{T}{2LR^2}$, which yields the desired rate of convergence.
\end{proof}

\subsection{Proof of \Cref{lem:cdconvergence}}
\label{sec:cdproof}

The progress of a step in the $j^{th}$ coordinate is thus lower bounded by $-\frac{1}{2L_j} |\nabla_j f(x^k)|^2$, which can be verified by computing the upper bound on $f(x^{k + 1})$. The analysis of convergence follows directly from the following result on the expected progress of a single step.

\begin{lemma}
\label{lem:cdprogress}
$f(x^k) - \mathbb{E}_k[f(x^{k + 1})] \geq \frac{1}{2S} \|\nabla f(x^k)\|_2^2$
\end{lemma}

\begin{proof}
We directly compute the expectation. We have 
\begin{equation*}
\mathbb{E}[f(x^{k + 1})] = \sum_j \frac{L_j}{S} \Big( f(x^k) - \frac{1}{2L_j} |\nabla_j f(x^k)|^2 \Big) = f(x^k) - \frac{1}{2S} \sum_j |\nabla_j f(x^k)|^2
\end{equation*}
\end{proof}

Thus, we can immediately plug in this expected progress result into the convergence rate proof of gradient descent, and obtain the desired result.
\section{Missing proofs from \Cref{sec:maxflow}}
\label{appendix:b}

\subsection{Reducing undirected maximum flow to $\ell_\infty$ regression}

In this section, we prove \Cref{lemma:flowreduction}, via giving the reduction and analyzing its convergence. First, suppose we have a subroutine, $\textsc{Almost-Route}$, which takes in matrices $R$ (an $\alpha = \tilde{O}(1)$-congestion approximator), $B$ (an edge-incidence matrix), $U$ (the capacities of edges), $\alpha$, an error tolerance $\epsilon$, and a demand vector $d$, and returns some $x$ such that 
\begin{equation}
2\alpha\|RBUx - Rd\|_\infty + \|x\|_\infty \leq (1 + \epsilon)(2\alpha\|RBUx^* - Rd\|_\infty + \|x^*\|_\infty) \defeq (1 + \epsilon) \textsc{OPT}(d)
\end{equation}
Here, under a change of variables we have that $x = U^{-1} f$. Note that we are writing with an $\epsilon$-multiplicative approximation to $\textsc{OPT}$ instead of an additive one. We do this without loss of generality: assume we have scaled the problem appropriately so that the optimal value is 1, which it will be when we find the true maximum flow instead of the minimum congestion flow. We can find this optimal value via a binary search, as we argued before, losing a $\tilde{O}(1)$ factor in the runtime.

Now, we show a key property of the function we try to minimize. Intuitively, the next lemma says that if we are able to $\epsilon$-approximately minimize our regression problem, the cost of routing the residual demands $d - Bf$ is only an $\epsilon$ fraction of routing the original demands, allowing us to quickly recurse. This is a restatement of Lemma 2.2 in \cite{Sherman13}.

\begin{lemma}
\label{lemma:residual}
Define the change of variables $Ux = f$. Suppose $2\alpha\|RBUx - Rd\|_\infty + \|x\|_\infty = 2\alpha\|R(d - Bf)\|_\infty + \|U^{-1} f\|_\infty \leq (1 + \epsilon) \textsc{OPT}(d)$. Then, $\|R(d - Bf)\|_\infty \leq \epsilon\|Rd\|_\infty$.
\end{lemma}

\begin{proof}
Let $f'$ be the optimal routing of the residual demands $d - Bf$, namely the argument which achieves $\textsc{OPT}(d - Bf)$. Then, $Bf' = d - Bf$, and by the definition of a congestion approximator,
\begin{equation}
\textsc{OPT}(d - Bf) = \|U^{-1} f'\|_\infty + 2\alpha\|R((d - Bf) - Bf')\|_\infty = \|U^{-1} f'\|_\infty \leq \alpha\|R(d - Bf)\|_\infty
\end{equation}
For simplicity we write $d' \defeq d - Bf$. Furthermore, we have by assumption of the quality of the initial solution $f$,
\begin{align}
\textsc{OPT}(d) + \alpha\|Rd'\|_\infty & \leq \|U^{-1}(f + f')\|_\infty + 2\alpha\|R(d - B(f + f'))\|_\infty + \alpha\|Rd'\|_\infty \\
& \leq \|U^{-1} f\|_\infty + \|U^{-1}f'\|_\infty + \alpha\|Rd'\|_\infty \\
& \leq \|U^{-1} f\|_\infty + 2\alpha\|Rd'\|_\infty \leq (1 + \epsilon) \textsc{OPT}(d)
\end{align}
Here, we used that $d = B(f + f')$ and our bound $\|U^{-1}f'\|_\infty \leq \alpha\|Rd'\|_\infty$. Subtracting $\textsc{OPT}(d)$, and noting that $\textsc{OPT}(d) \leq \alpha\|Rd\|_\infty$, we have the desired claim.
\end{proof}

Now, we give the full reduction to calling $\textsc{Almost-Route}$. Note that it was shown in \cite{Sherman13} that routing through a maximal spanning tree yields an $O(m)$-congestion approximator.

\begin{figure}[ht]
\noindent
\centering
\fbox{
\begin{minipage}{6in}
    \noindent $f^{final} = \textsc{Flow-To-Regress} (G, d, \epsilon)$
\begin{enumerate}
	\item Let $T = \log 2m$.
    \item Initialize $d^0 = d$. Initialize $f^0 = U\textsc{Almost-Route}(R, B, U, d^0, \alpha, \epsilon)$.
    \item Let $f^{final} = f^0$.
	\item Iterate for $k = 1, 2, \ldots T$:
    \begin{enumerate}
    	\item Let $d^k = d^{k - 1} - Bf^{k - 1}$.
        \item Let $f^k = U\textsc{Almost-Route}(R, B, U, D^k, \alpha, \frac{1}{2})$.
        \item Let $f^{final} = f^{final} + f^k$.
    \end{enumerate}
  \item Let $f^{T + 1}$ be an (exact) routing of $d^k - Bf^k$ in a maximal spanning tree. Let $f^{final} + f^{T + 1}$.
  \item Return $f^{final}$
\end{enumerate}
\end{minipage}
}
\caption{The reduction from solving the approximate maximum flow problem to solving $\tilde{O}(1)$ approximate regression problems.}
\label{fig:flowtoregress}
\end{figure}

We now need to prove the correctness of our algorithm. This is a restatement of ideas presented in \cite{Sherman13}.

\begin{lemma}
The output of $\textsc{Flow-To-Regress}$ is an $\epsilon$-approximate solution to the minimum congestion flow problem.
\end{lemma}

\begin{proof}
By the guarantees of $\textsc{Almost-Route}$, we have the following guarantees:
\begin{align}
\|U^{-1}f^0\|_\infty + 2\alpha\|Rd^1\|_\infty \leq (1 + \epsilon)\textsc{OPT}(d), \\
\|U^{-1}f^k\|_\infty + 2\alpha\|Rd^{k + 1}\|_\infty \leq \frac{3}{2} \textsc{OPT}(d^k) \leq \frac{3}{2} \alpha\|Rd^k\|_\infty, k \geq 1.
\end{align}
Now, using the second inequality and repeatedly applying it to the first, we have the following guarantee:
\begin{equation}
\frac{1}{2}\alpha\|Rd^1\|_\infty + \|U^{-1}f^0\|_\infty + \ldots + \|U^{-1}f^{T}\|_\infty \leq (1 + \epsilon) \textsc{OPT}(d).
\end{equation}
It suffices to note that by our choice of $T$ and seeing that by applying \Cref{lemma:residual} $T$ times, we have $\alpha\|Rd^{T + 1}\|_\infty \leq \frac{1}{2m} \alpha\|Rd^1\|_\infty$. Thus because we routed $d^{T + 1}$ exactly through a $m$-congestion approximator, we have $\|U^{-1} f^{T + 1}\|_\infty \leq \frac{1}{2}\alpha\|Rd^1\|_\infty$. Finally, $Bf^{final} = d$, and
\begin{align}
\|U^{-1} f^{final}\|_\infty & \leq \|U^{-1}f^{T + 1}\|_\infty + \|U^{-1}f^0\|_\infty + \ldots + \|U^{-1}f^{T}\|_\infty \\
& \leq \|U^{-1} f^{final}\|_\infty \leq \frac{1}{2}\alpha\|Rd^1\|_\infty + \|U^{-1}f^0\|_\infty + \ldots + \|U^{-1}f^{T}\|_\infty \\
& \leq (1 + \epsilon) \textsc{OPT}(d).
\end{align}
\end{proof}

\begin{lemma}
The runtime of our routine $\textsc{Flow-To-Regress}$ is the cost of solving the first associated regression problem, $2\alpha \|RBUx - Rd\|_\infty + \|x\|_\infty$, to an $\epsilon$ approximation, plus an additional $\tilde{O}(m)$ additive overhead.
\end{lemma}

\begin{proof}
We analyze the time of each of the calls to $\textsc{Almost-Route}$. Clearly, the first call is the cost of solving the first associated regression problem. 

Note that we have flexibility in terms of how to implement $\textsc{Almost-Route}$; for all remaining calls, we consider the implementation in the form of unaccelerated gradient descent in the $\ell_\infty$ norm. The runtime as we demonstrated in \Cref{sec:gdgeneralnorm} for each round $k$ is 
\begin{equation}
\frac{m\|f^k_*\|_\infty^2 \|\alpha RBU\|_\infty^2}{(\frac{1}{2})^2} = \tilde{O}(m)
\end{equation}
where $f^k_*$ is the optimal solution to the $k^{th}$ regression problem. Here, we used the known properties of $\alpha RBU$, as well as the fact that the implications of \Cref{lemma:residual} allow us to bound the $\ell_\infty$ norm of the optimal solution by $O(1)$ as well.
\end{proof}

As a final note in the proof of \Cref{lemma:flowreduction}, observe that to optimize the first objective $2\alpha\norm{Ax - b}_\infty + \norm{x}_\infty$ it suffices to binary search over values $r \geq \|U^{-1} f\|_\infty = \|x\|_\infty$, and solve the associated regression problem $\|Ax - b\|_\infty$ over $x \in [-r, r]^m$. More formally, since $r$ is our guess of $\text{OPT}$ to the original flow problem, we repeatedly solve the problem over $[-r, r]^m$ to $\eps r$ additive error; if the conclusion is that the optimal value cannot be 0 (i.e. the additive approximation is larger than $\eps r$), then we conclude that this value of $r$ is not routable. This only incurs a multiplicative loss in the runtime by a factor of $\tilde{O}(1)$, due to the binary search. By normalizing $x, b$ appropriately, it suffices to consider the case where $r = 1$, and solve to $\eps$ additive error (see Appendix~\ref{appendix:reductionbox}). Finally, we need only consider the case where $\norm{b}_\infty \le \norm{A}_\infty$, as $x \in [-1, 1]^{\infty}$ the unit box, so the demands are clearly not routable otherwise.

\subsection{Reducing directed maximum flow to undirected maximum flow}
\label{ssec:dirtoundir}

In this section, we give an overview of the main result in \cite{Lin09}. In particular, we prove the following statement, which is used in our algorithms for finding exact maximum flows in unit-capacity graphs.

\begin{lemma}[Summary of results in \cite{Lin09}]
Suppose we wish to find an $s-t$ maximum flow in a unit-capacity directed (multi)graph $G$ with $m$ edges and maximum flow value $F$. Then, it suffices to find the $s-t$ maximum flow $f_{max}$ in an undirected (multi)graph $G'$ with $O(m)$ edges, such that edges of $G'$ have capacity $\frac{1}{2}$, and the maximum flow in $G'$ has value $F + \frac{m}{2}$. Furthermore, we are able to initialize the undirected maximum flow algorithm in $G'$ with some $f_{init}$ such that $\norm{f_{init} - f_{max}}_2^2 = F$.
\end{lemma}

\begin{proof}
First, we give the construction of the undirected graph $G'$. For every directed edge $(u, v)$ of weight 1 in $G$, $G'$ has the undirected edges $(s, v)$, $(v, u)$, and $(u, t)$ of weight $\frac{1}{2}$. Clearly, $G'$ has $O(m)$ edges, since each edge in $G$ is replaced with 3 edges in $G'$.

Next, we give the (algorithmic) proof that one can recover a maximum flow in $G$ from a maximum flow in $G'$, and that the maximum flow in $G'$ has value $F + \frac{m}{2}$. Consider the following algorithm.

\begin{figure}[ht!]
\noindent
\centering
\fbox{
\begin{minipage}{6in}
    \noindent $f = \textsc{UMF-to-DMF} (G)$
\begin{enumerate}
	\item Let $G'$ be the undirected graph with edges $(s, v), (v, u), (u, t)$ of weight $\frac{1}{2}$ for every directed edge $(u, v)$ in $G$.
    \item Let $f_{init}$ be the flow which puts $\frac{1}{2}$ units of flow on each of the $(s, v), (v, u), (u, t)$.
    \item Compute $f_{final}$, the maximum flow of $G'$.
    \item Return $f_{final} - f_{init}$.
\end{enumerate}
\end{minipage}
}
\caption{Recovering a maximum flow in directed $G$ via a maximum flow in undirected $G'$.}
\label{fig:dir-to-undir}
\end{figure}

We will now prove correctness of the algorithm $\textsc{UMF-to-DMF}$, namely that $f_{final} - f_{init}$ is a maximum flow in graph $G$. To do so, we show that $f_{final}$ has value $\frac{m}{2} + F$, and that $f_{final} - f_{init}$ puts flow only in the $(u, v)$ direction and does not put any flow on any new edges $(s, v)$ or $(u, t)$. Note that this immediately implies the statement $\norm{f_{init} - f_{max}}_2^2 = F$.

We begin by showing that $f_{final}$ has value $\frac{m}{2} + F$. The residual graph of $G'$ with respect to the flow $f_{init}$ is the directed graph $G$. Thus, the maximum flow in the residual graph has value $F$ by assumption, and the flow $f_{init}$ has value $\frac{m}{2}$, yielding the conclusion.

Next, we show that for every edge $(u, v)$ in $G'$ which resulted from a directed edge $(u, v)$ in $G$, $f_{final} - f_{init}$ puts flow only in the $(u, v)$ direction, and does not violate the capacity constraint. This is simple to see because $f_{final}$ puts a flow with value in $\{-\frac{1}{2}, 0, \frac{1}{2}\}$ in the $(u, v)$ direction, and $-f_{init}$ puts a flow with value $\frac{1}{2}$ in the $(u, v)$ direction; adding yields the result.

Finally, we show that $f_{final} - f_{init}$ puts no flow on any of the new edges $(s, v)$ (the same statement holds for edges $(u, t)$ by a similar argument). Again, $f_{final}$ puts a flow with value in $\{-\frac{1}{2}, 0, \frac{1}{2}\}$ in the $(v, s)$ direction, and $-f_{init}$ puts a flow with value $\frac{1}{2}$ in the $(v, s)$ direction, thus $f_{final} - f_{init}$ puts a flow with nonnegative value in the $(v, s)$ direction. If this value was strictly positive, it would be part of a path in the flow decomposition sending flow into $s$, contradicting the maximality of $f_{final}$.

\end{proof} \end{appendix}

\end{document}